\documentclass[11pt,a4paper,reqno]{amsart}
\usepackage{amsthm,amsmath,amsfonts,amssymb,amsxtra,appendix,bookmark,euscript,delarray,dsfont,mathrsfs}
\usepackage{enumerate}
\usepackage{amssymb}
\usepackage[utf8]{inputenc}
\usepackage{xcolor}
\usepackage{graphicx}
\usepackage{upgreek}
\usepackage[normalem]{ulem}
\usepackage{ stmaryrd }
\usepackage{array}
\usepackage{longtable}
\newcolumntype{L}[1]{>{\raggedright\arraybackslash}p{#1}}
\usepackage{url}
\usepackage{verbatim}

\allowdisplaybreaks

\usepackage{tikz}
\usetikzlibrary{decorations.pathreplacing, calc}

\theoremstyle{plain}
\newtheorem{theorem}{Theorem}[section]
\newtheorem{lemma}[theorem]{Lemma}
\newtheorem{remark}[theorem]{Remark}

\newtheorem{proposition}[theorem]{Proposition}

\newtheorem{corollary}[theorem]{Corollary}
\newtheorem{assumption}{Assumption}[section]

\numberwithin{equation}{section}

\renewcommand{\d}{\mathrm{d}}
\newcommand{\dd}{\,\mathrm{d}}
\newcommand{\ii}{\mathrm{i}}

\DeclareMathOperator{\tr}{Tr}

\renewcommand{\leq}{\leqslant}
\renewcommand{\geq}{\geqslant}
\renewcommand{\ge}{\geqslant}
\renewcommand{\le}{\leqslant}
\def\nn{\nonumber}
\renewcommand{\to}{\rightarrow}
\renewcommand{\phi}{\varphi}

\newcommand{\norm}[1]{\left\lVert #1 \right\rVert}
\newcommand{\abs}[1]{\left\lvert#1\right\rvert}
\renewcommand{\phi}{\varphi}
\newcommand{\ket}[1]{|#1 \rangle}
\newcommand{\bra}[1]{\langle #1 |}
\newcommand		{\lt}			{\left}
\newcommand		{\rt}			{\right}
\newcommand		{\bangle}[1]	{\lt\langle #1\rt\rangle}
\newcommand		{\inprod}[2]	{\bangle{#1, #2}}

\newcommand\R{{\ensuremath {\mathbb R} }}
\newcommand{\CC}{{\ensuremath {\mathbb C} }}
\newcommand\Z{{\ensuremath {\mathbb Z} }}
\newcommand\1{{\ensuremath {\mathds 1} }}
\newcommand{\gH}{\mathfrak{H}}
\newcommand{\dG}{\d\Gamma}
\newcommand{\gX}{\mathfrak{X}}
\newcommand{\gF}{\mathfrak{F}}
\newcommand{\bH}{\mathbb{H}}
\newcommand{\bW}{\mathbb{W}}
\newcommand{\gS}{\mathfrak{S}}
\newcommand{\cN}{\mathcal{N}}
\newcommand{\cD}{\mathcal{D}}
\newcommand{\cU}{\mathcal{U}}
\newcommand{\cH}{\mathcal{H}}
\newcommand{\cM}{\mathcal{M}}
\newcommand{\cE}{\mathcal{E}}
\newcommand{\bT}{\mathbb{T}}
\newcommand{\cX}{\mathcal{X}}
\newcommand{\cl}{\mathrm{cl}}
\newcommand{\ren}{\mathrm{ren}}

\title[Nonlocal cubic density Gibbs states with three-body interactions]{Nonlocal Cubic Density Gibbs Measures from Bosonic Gibbs States with Three-Body Interactions}
\author[H. Liang]{Hao Liang}
\address{School of Mathematical Sciences, Peking University, Beijing, 100871, China}
\email{leunghao@stu.pku.edu.cn}

\begin{document}

\begin{abstract}
We study the high-temperature mean-field limit of grand-canonical bosonic
Gibbs states on the torus with renormalized nonlocal three-body interactions.  In dimensions two and three, we
construct the limiting nonlinear classical Gibbs measure and prove
convergence of the relative free energy and of the reduced density matrices
of every fixed order; in three dimensions, a smallness condition on the
interaction is imposed.  The proof combines the density-channel
representation of the interaction with a coherent-state variational method
based on the upper-symbol representation of the free Gibbs state.
The same framework also contains, as a special case, the
homogeneous positive-type model studied by Lewin, Nam, and Rougerie \cite{LNR21}.
\end{abstract}



\maketitle

\tableofcontents

\section{Introduction}

A nonlinear Gibbs measure is formally written as 
\[
    \dd\mu(u)
    =
    z^{-1}e^{-\cE(u)}\dd\mu_0(u),
\]
where $\mu_0$ is a Gaussian reference measure, $\cE$ is a nonlinear
energy, and $z$ is the corresponding partition function.  In the
infinite-dimensional setting this expression is only symbolic: a typical
field sampled from $\mu_0$ is too singular for the nonlinear energy to be
defined pointwise, and both the interaction and the Gibbs weight generally
require a suitable renormalization.  

Nonlinear Gibbs measures arise naturally in constructive quantum field
theory, see, for example \cite{nelson1973construction,nelson2007probability,guerra1975p,simon2015p}. 
In suitable settings, nonlinear Gibbs measures are invariant under the
flows of nonlinear Schrödinger equations
\cite{bourgain1994periodic,bourgain1996invariant,bourgain1997invariant,
bourgain2000invariant}.  They have also been used to construct global
solutions with rough random initial data; see, for example,
\cite{burq2008randomA,burq2008randomB,bourgain2014almostA,
bourgain2014almostB,bringmann2024invariant,deng2012two,
deng2021invariant,deng2024invariant}.

In this paper, we study the emergence of nonlinear classical Gibbs measures
from grand-canonical bosonic quantum Gibbs states.
More precisely, let
\[
    \Gamma_\lambda
    =
    Z_\lambda^{-1}e^{-\bH_\lambda^{\ren}}
\]
denote the interacting quantum Gibbs state.  Here
$\lambda\downarrow0$ parametrizes the high-temperature mean-field regime;
see \cite{frohlich2020path} for a discussion of this scaling and related
asymptotic regimes.  The renormalized many-body Hamiltonian
$\bH_\lambda^{\ren}$ is defined in
Section~\ref{sec:main-results}. The corresponding free Gibbs state and
partition function are denoted by $\Gamma_0$ and $Z_0$, respectively.
On the classical side, the limiting nonlinear Gibbs measure has the form
\[
    \dd\mu(u)
    =
    z^{-1}e^{-\cD_0[u]}\dd\mu_0(u),
\]
where $\mu_0$ is the Gaussian free field, $\cD_0$ is the renormalized
classical interaction, and $z$ is the classical partition function.  The
derivation problem is to prove
\[
    -\log\frac{Z_\lambda}{Z_0}
    \longrightarrow
    -\log z
\]
and, for every fixed $k\geq1$,
\[
    k!\lambda^k\Gamma_\lambda^{(k)}
    \longrightarrow
    \gamma_\mu^{(k)},
    \qquad
    \gamma_\mu^{(k)}
    =
    \int
    \ket{u^{\otimes k}}\bra{u^{\otimes k}}
    \dd\mu(u).
\]
Thus the relative quantum free energy converges to the classical free
energy, while the fixed-order reduced density matrices converge to the correlation
operators of the nonlinear Gibbs measure.  The precise statement, including
the topology of convergence and the assumptions in dimensions two and
three, is given in Theorem~\ref{thm:quantum-classical-main}.

This question belongs to a broader program initiated in
\cite{LNR15} and subsequently developed in a variety of settings
\cite{frohlich2017gibbs,lewin2018gibbs,frohlich2019microscopic,LNR21,
frohlich2022mean,sohinger2022microscopic,FKSS23,
NZZ25,nam2026derivation,lu2026derivation,caraci2026euclidean,
JR26,farhatPerturbativeMicroscopicDerivation2026,NYZ26}. Several complementary approaches have been developed, including
variational, perturbative, functional-integral, and stochastic methods.

Among these works, \cite{LNR21} is particularly relevant to the present
paper. It treats defocusing two-body interactions in dimensions two
and three and, in the homogeneous torus setting, assumes positivity and
suitable summability of the Fourier coefficients of the interaction.  We
observe that these assumptions reveal a translated density-channel
structure: the two-body interaction becomes a quadratic functional of such
channels, and its Fourier-space renormalization is equivalent to centering
the same channels in position space.  This reformulation is the motivation
for the general channel framework and its cubic extension developed in the
next subsection.

\subsection{Motivation and the channel formulation}

Let $d\in\{2,3\}$ and
$\bT^d=(\R/2\pi\Z)^d$.  A real-valued translation-invariant two-body
interaction $w$ on $\bT^d$ is said to be of positive type if its Fourier
coefficients satisfy
\[
    \widehat w(k)\geq0,
    \qquad
    \widehat w(-k)=\widehat w(k),
    \qquad
    k\in\Z^d,
\]
where
\[
    w(x)
    =
    \frac1{(2\pi)^d}
    \sum_{k\in\Z^d}
    \widehat w(k)e^{\ii k\cdot x}.
\]
In the torus setting of \cite[(3.6)]{LNR21}, the stronger weighted
summability condition
\[
    \sum_{k\in\Z^d}
    (1+|k|^2)\widehat w(k)<\infty
\]
is imposed.  Positivity of $\widehat w$ is the natural defocusing
assumption for a quadratic density interaction.  It also reveals a useful
position-space structure which motivates the present work.

Let
\[
    e_k(x)=(2\pi)^{-d/2}e^{\ii k\cdot x},
    \qquad k\in\Z^d,
\]
and, for every $k\in\Z^d$, define the complex-valued mode channel
\begin{equation}\label{eq:intro-mode-channels}
    J_k=(2\pi)^{-d/2}\widehat w(k)^{1/2}e_k.
\end{equation}
Writing
\[
    (\tau_r f)(x)=f(x-r),
\]
a direct calculation gives
\begin{align}
    \sum_{k\in\Z^d}
    \int_{\bT^d}
    J_k(x-r)\overline{J_k(y-r)}\dd r
    =
    \frac1{(2\pi)^d}
    \sum_{k\in\Z^d}
    \widehat w(k)e^{\ii k\cdot(x-y)}
    =
    w(x-y).
\end{align}
Moreover,
\begin{equation}\label{eq:intro-mode-channel-norm}
    \sum_{k\in\Z^d}
    \norm{J_k}_{L^\infty}^2
    =
    (2\pi)^{-2d}
    \sum_{k\in\Z^d}\widehat w(k).
\end{equation}
Thus, for positive-type interactions, absolute summability of the Fourier
coefficients is equivalent to square-summability of the $L^\infty$-norms
of the associated mode channels.

This observation suggests replacing the Fourier-mode family by a general
measurable family of channels.  Let $(\Omega,\nu)$ be a measure space and
let $J_\omega\in L^\infty(\bT^d)$ be complex-valued.  The quadratic
channel kernel is
\begin{equation}\label{eq:intro-general-quadratic-channel}
    w(x-y)
    =
    \int_\Omega\dd\nu(\omega)
    \int_{\bT^d}
    J_\omega(x-r)\overline{J_\omega(y-r)}\dd r,
    \qquad
    \int_\Omega
    \norm{J_\omega}_{L^\infty}^2\dd\nu(\omega)<\infty.
\end{equation}
The mode decomposition \eqref{eq:intro-mode-channels} is a particular case
of \eqref{eq:intro-general-quadratic-channel}.  In particular, the quadratic
channel framework contains the homogeneous positive-type torus interactions
of \cite{LNR21}; for results on the relative free energy and fixed-order reduced density matrices
considered here, only the unweighted summability appearing in
\eqref{eq:intro-mode-channel-norm} is needed.

The channel formulation also identifies the position-space centering used
in this paper with the Fourier-space renormalization of \cite{LNR21}.  To
state this equivalence, let $\cM_0^A$ denote the centered quadratic
Gaussian field associated with a one-body operator $A$, as defined
precisely in Section~\ref{sec:classical-model}.  Since
\[
    \tau_rJ_k:=J_k(\cdot-r)
    =
    \widehat w(k)^{1/2}\overline{e_k(r)}\,e_k,
\]
linearity of the centering and Parseval's identity yield
\begin{equation}\label{eq:intro-classical-renormalization-equivalence}
    \frac12
    \sum_{k\in\Z^d}
    \int_{\bT^d}
    \left|
    \cM_0^{\tau_rJ_k}[u]
    \right|^2\dd r
    =
    \frac12
    \sum_{k\in\Z^d}
    \widehat w(k)
    \left|
    \cM_0^{e_k}[u]
    \right|^2.
\end{equation}
The corresponding quantum identity is obtained in the same way.  With
$\rho_k=\dG(e_k)$ denoting the density Fourier modes and $\gamma_0$
the one-body density matrix of the free quantum Gibbs state, both defined in
Section~\ref{sec:main-results}, one has
\begin{align}
    \frac12
    \sum_{k\in\Z^d}
    \int_{\bT^d}
    \left|
    \lambda\dG(\tau_rJ_k)
    -
    \lambda\tr_{\gH}\big((\tau_rJ_k)\gamma_0\big)
    \right|^2\dd r
    =
    \frac12
    \sum_{k\in\Z^d}
    \widehat w(k)
    \left|
    \lambda\rho_k
    -
    \lambda\langle\rho_k\rangle_0
    \right|^2.
    \label{eq:intro-quantum-renormalization-equivalence}
\end{align}
Here $|X|^2=X^\dagger X$.  The right-hand sides of
\eqref{eq:intro-classical-renormalization-equivalence} and
\eqref{eq:intro-quantum-renormalization-equivalence} are the Fourier-space
renormalizations used in the homogeneous setting of \cite{LNR21}, whereas
the left sides perform the same operation directly on translated density
channels. Thus the two renormalizations are equivalent under the Fourier transform.

Once a positive two-body interaction is written as a quadratic functional
of translated density channels, a higher-order extension becomes natural.
The cubic analogue is generated by real-valued nonnegative channels through
\[
    v(x-y,x-z)
    =
    \int_\Omega\dd\nu(\omega)
    \int_{\bT^d}
    J_\omega(x-r)J_\omega(y-r)J_\omega(z-r)\dd r.
\]
The integration in $r$ preserves translation invariance, while
nonnegativity of $J_\omega$ is the defocusing condition appropriate to an
odd density power.  The precise assumptions on the channels, the associated
renormalized Hamiltonian, and the limiting classical Gibbs measure are
introduced in Section~\ref{sec:main-results}.

The main results of the paper concern this cubic model.  We prove the 
convergence of the relative free energy and the Hilbert--Schmidt
convergence of reduced density matrices of every fixed order; see
Theorem~\ref{thm:quantum-classical-main}. 

\subsection{Organization of the paper}

The remainder of the paper is organized as follows.  Section~\ref{sec:main-results}
introduces the many-body framework, defines the quantum and classical models, and
states the main convergence theorem.  Section~\ref{sec:strategy} outlines the
coherent-state variational argument, first in the simpler quadratic channel
setting and then for the three-body interaction considered in the main result.

Sections~\ref{sec:classical-model} and
\ref{sec:convergence-classical-partition} construct the nonlinear classical
Gibbs measure and establish the convergence of the associated partition
functions.  Section~\ref{sec:quantum-model} develops the corresponding quantum
Gibbs states and the estimates needed for the semiclassical comparison.
Finally, Sections~\ref{sec:Q-C-free-energy} and \ref{sec:Q-C-density} prove,
respectively, the convergence of the relative free energies and of the fixed-order
reduced density matrices.  

\medskip

\subsection*{Acknowledgments}
The author is grateful to Zhenfu Wang for many helpful discussions, to Phan Th{\`a}nh Nam for valuable comments on the manuscript and for sharing related work, and to Nicolas Rougerie for helpful comments and discussions, particularly on the second Berezin--Lieb inequality. This work was partially supported by the National Key R\&D Program of China (Project No.~2024YFA1015500), and the author was partially supported by the NSFC (Grant Nos.~12595282 and 12171009).

\bigskip

\section{Setting and main results}\label{sec:main-results}

Throughout the paper, the notation \(A\lesssim_p B\) means that
\(A\leq C_p B\) for some constant \(C_p>0\) depending only on \(p\).
Unless otherwise specified, \(\norm{\cdot}\) denotes the operator norm.
When no confusion can arise, a function \(f\in L^\infty\) is identified
with the corresponding multiplication operator.  For a Hilbert space
\(\gH\), we denote by \(\gS^1(\gH)\) and \(\gS^2(\gH)\) the spaces of
trace-class and Hilbert--Schmidt operators on \(\gH\), respectively, and
we omit the underlying Hilbert space from the notation when it is clear
from the context.

\subsection{Many-body framework}

We recall the grand-canonical bosonic framework used throughout the paper.
Our conventions follow \cite[Section~2.1]{nam2026derivation}. Let $d\in\{2,3\}$ and 
\[
    \gH=L^2(\bT^d),
    \qquad
    \bT^d=(\R/2\pi\Z)^d,
\]
and let
\[
    \gF
    =
    \bigoplus_{n=0}^{\infty}\gH^{\otimes_s n},
    \qquad
    \gH^{\otimes_s0}=\CC .
\]
Here $\otimes_s$ denotes the symmetric tensor product.  For a one-body operator $A$ on $\gH$, its additive
second quantization is
\[
    \dG(A)
    =
    0\oplus
    \bigoplus_{n\geq1}
    \sum_{j=1}^n A_j,
\]
where $A_j$ acts as $A$ on the $j$-th variable and as the identity
on the other variables. We also use the multiplicative
second quantization
\[
    \Gamma(B)
    :=
    1\oplus
    \bigoplus_{n\geq1}
    B^{\otimes n}\big|_{\gH^{\otimes_s n}}
\]
on its natural domain.  Thus
$\Gamma(B_1)\Gamma(B_2)=\Gamma(B_1B_2)$ whenever the products are
defined, and $\Gamma(e^B)=e^{\dG(B)}$ for bounded self-adjoint $B$.
In particular,
\[
    \cN=\dG(\1)
\]
is the number operator.

For $f\in L^2(\bT^d)$, the annihilation and creation operators
$a(f)$ and $a^\dagger(f)$ are defined on the finite-particle
subspace by
\[
\begin{aligned}
    (a(f)\psi)(x_1,\ldots,x_{n-1})
    &=
    \sqrt n
    \int_{\bT^d}
    \overline{f(x)}
    \psi(x_1,\ldots,x_{n-1},x)\dd x,
    \\
    (a^\dagger(f)\psi)(x_1,\ldots,x_{n+1})
    &=
    \frac1{\sqrt{n+1}}
    \sum_{j=1}^{n+1}
    f(x_j)
    \psi(x_1,\ldots,x_{j-1},x_{j+1},\ldots,x_{n+1}).
\end{aligned}
\]
Equivalently, in the distributional notation,
\[
    a(f)=\int_{\bT^d}\overline{f(x)}a_x\dd x,
    \qquad
    a^\dagger(f)=\int_{\bT^d}f(x)a_x^\dagger\dd x .
\]
They satisfy the canonical commutation relations
\[
    [a(f),a(g)]=0,\qquad
    [a^\dagger(f),a^\dagger(g)]=0,\qquad
    [a(f),a^\dagger(g)]=\inprod{f}{g}.
\]

We use the normalized Fourier basis
\[
    e_k(x)=(2\pi)^{-d/2}e^{\ii k\cdot x},
    \qquad
    k\in\Z^d.
\]
We write
\[
    a_k=a(e_k),
    \qquad
    a_k^\dagger=a^\dagger(e_k),
\]
so that
\[
    [a_k,a_\ell]=0,
    \qquad
    [a_k^\dagger,a_\ell^\dagger]=0,
    \qquad
    [a_k,a_\ell^\dagger]=\delta_{k,\ell}.
\]

A state on $\gF$ is a positive trace-class operator $\Gamma$ with
$\tr_{\gF}\Gamma=1$.  Writing
$\Gamma=\bigoplus_{n\geq0}\Gamma_n$ with respect to the particle-number
decomposition, we use the fixed-order reduced density matrices normalized by
\[
    \Gamma^{(k)}
    =
    \sum_{n\geq k}
    \binom nk
    \tr_{k+1\to n}\Gamma_n,
    \qquad
    k\geq1.
\]
Equivalently, for every bounded self-adjoint operator $A_k$ on
$\gH^{\otimes_s k}$,
\[
    \tr_{\gH^{\otimes_s k}}(A_k\Gamma^{(k)})
    =
    \tr_{\gF}(\mathbb A_k\Gamma),
\]
where $\mathbb A_k$ is the usual second-quantized lift of $A_k$ defined by
\[
    \mathbb{A}_{k}:=0\oplus\cdots\oplus\bigoplus_{n=k}^{\infty}\Big( \sum_{1\leq i_{1}<\cdots<i_{k}\leq n}(A_{k})_{i_{1},..., i_{k}} \Big)  .
\]
In particular,
\[
    \tr_{\gH}(\Gamma^{(1)})
    =
    \tr_{\gF}(\cN\Gamma).
\]
Define the one-body operator
\[
    h=-\Delta+1,\qquad \dG(h)=\sum_{k\in\Z^d}h(k)a_k^{\dagger}a_k,
\]
where $h(k)=|k|^2+1$. Throughout the paper,
\[
    x_-:=\max\{-x,0\},
    \qquad x\in\R,
\]
denotes the negative part.

\subsection{Free Gibbs states}

For $\lambda>0$, the free Gibbs state is
\begin{equation}\label{eq:def-free-gibbs-state}
    \Gamma_0=Z_0^{-1}e^{-\lambda\dG(h)},\qquad
    Z_0=\tr_{\gF}(e^{-\lambda\dG(h)}).
\end{equation}
Its one-body reduced density matrix is
\[
    \gamma_0=\frac{1}{e^{\lambda h}-1}.
\]

We use
\[
    N_0=N_0(\lambda):=\tr_{\gF}(\cN \Gamma_0)=\tr_{\gH}(\gamma_0)
    =
    \sum_{p\in\Z^d}\frac1{e^{\lambda h(p)}-1}
\]
to denote the expected particle number of the free Gibbs state. 

\subsection{Renormalized three-body interaction and Gibbs states}

For a translation-invariant three-body interaction, the unrenormalized
many-body Hamiltonian has the form
\[
    \bH_\lambda
    =
    \lambda\dG(h)+\bW_\lambda,
\]
where
\[
\begin{aligned}
    \bW_\lambda
    &=
    \frac{\lambda^3}{6}
    \int_{\bT^{3d}}
    v(x-y,x-z)
    a_x^\dagger a_y^\dagger a_z^\dagger
    a_z a_y a_x
    \dd x\dd y\dd z .
\end{aligned}
\]
Equivalently, using the Fourier convention fixed above, one may write
\[
    \bW_\lambda
    =
    \frac{\lambda^3}{6(2\pi)^{2d}}
    \sum_{k,\ell\in\Z^d}\widehat v(k,\ell)
    \sum_{p,q,s\in\Z^d}
    a_{p+k+\ell}^\dagger
    a_{q-k}^\dagger
    a_{s-\ell}^\dagger
    a_s a_q a_p .
\]

Let $(\Omega,\nu)$ be a $\sigma$-finite measure
space.  For $\nu$-a.e. $\omega$, let $J_\omega$ be a function on $\bT^d$. For $r\in\bT^d$, we write
\[
    (\tau_rJ_\omega)(x)=J_\omega(x-r).
\]
We consider three-body potentials generated by translated density channels;
namely, $v$ is assumed to have the form
\begin{equation}\label{eq:def-int-v}
    v(x-y,x-z)
    =
    \int_\Omega\dd\nu(\omega)
    \int_{\bT^d}
    J_\omega(x-r)J_\omega(y-r)J_\omega(z-r)\dd r .
\end{equation}
In Fourier variables this gives
\[
    \widehat v(k,\ell)
    =
    \int_\Omega
    \widehat{J_\omega}(k+\ell)
    \widehat{J_\omega}(-k)
    \widehat{J_\omega}(-\ell)
    \dd\nu(\omega).
\]

The quantum renormalization is performed at the level of each translated
density channel.  Since the free state is translation invariant, the number
\[
    m_{\lambda,\omega}
    =
    \lambda\tr_{\gH}\big((\tau_rJ_\omega)\gamma_0\big)
    =
    \frac{\lambda\widehat{J_\omega}(0)}{(2\pi)^d}N_0
\]
is independent of $r$.  We define
\[
    \bW_\lambda^{\ren}
    =
    \frac16
    \int_\Omega\dd\nu(\omega)
    \int_{\bT^d}
    \left[
    \lambda\dG(\tau_rJ_\omega)-m_{\lambda,\omega}
    \right]^3\dd r
\]
and
\begin{equation}\label{eq:def-H-ren}
    \bH_\lambda^{\ren}
    =
    \lambda\dG(h)+\bW_\lambda^{\ren}.
\end{equation}
The interacting Gibbs state is defined by
\[
    \Gamma_\lambda
    =
    Z_\lambda^{-1}e^{-\bH_\lambda^{\ren}},
    \qquad
    Z_\lambda=\tr_{\gF}(e^{-\bH_\lambda^{\ren}}).
\]

\subsection{Classical Gibbs measures}

Let $\mu_0$ be the centered complex Gaussian free field on $\bT^d$
with covariance $h^{-1}$.  For a bounded operator $A$, we denote by
$\cM_0^A$ the centered quadratic field defined by the cutoff limit
\[
    \cM_{0,K}^A[u]
    =
    \inprod{P_Ku}{AP_Ku}
    -
    \mathbb E_{\mu_0}\inprod{P_Ku}{AP_Ku},
    \qquad
    P_K=\mathds1_{\{h\leq K\}}.
\]
In particular, the density channel associated with $\tau_rJ_\omega$
is renormalized by
\[
    \cM_0^{\tau_rJ_\omega}[u]
    =
    \lim_{K\to\infty}
    \cM_{0,K}^{\tau_rJ_\omega}[u],
\]
where the limit is understood in the $L^p(\mu_0)$-sense established
below, see Lemma~\ref{lem:Mass-Lp}.

We define the classical cubic channel interaction by
\[
    \cD_0[u]
    =
    \frac16
    \int_\Omega\dd\nu(\omega)
    \int_{\bT^d}
    \big(
    \cM_0^{\tau_rJ_\omega}[u]
    \big)^3\dd r ,
\]
and the corresponding nonlinear Gibbs measure by
\begin{equation}\label{eq:def-nonlin-gibbs-measure}
    \dd\mu(u)
    =
    z^{-1}e^{-\cD_0[u]}\dd\mu_0(u),
    \qquad
    z=
    \int e^{-\cD_0[u]}\dd\mu_0(u).
\end{equation}
The existence of $\cD_0$, the finiteness of $z$, and the convergence
of the corresponding cutoff measures are proved in
Section~\ref{sec:classical-model}.

For $k\geq1$, the associated classical $k$-particle correlation
operator is denoted by
\begin{equation}\label{eq:classical-correlation-formal}
    \gamma_\mu^{(k)}
    =
    \int
    \ket{u^{\otimes k}}\bra{u^{\otimes k}}
    \dd\mu(u),
\end{equation}
where the integral is understood as the Hilbert--Schmidt limit of its
Fourier cutoffs, see the proof of Proposition~\ref{prop:classical-correlation-operators}.

We will use the Gibbs variational principle, and we recall it here. The quantum Gibbs variational principle gives
\[
    -\log\frac{Z_\lambda}{Z_0}
    =
    \inf_{\substack{\Gamma\geq0\\ \tr_{\gF}\Gamma=1}}
    \left\{
    \cH(\Gamma,\Gamma_0)
    +
    \tr_{\gF}\big(\bW_\lambda^{\ren}\Gamma\big)
    \right\},
\]
where the quantum relative entropy is
\[
    \cH(\Gamma,\Gamma')
    =
    \begin{cases}
    \tr_{\gF}\!\left[
    \Gamma\big(\log\Gamma-\log\Gamma'\big)
    \right],
    & \operatorname{supp}\Gamma\subseteq\operatorname{supp}\Gamma',
    \\[0.3em]
    +\infty,
    & \text{otherwise}.
    \end{cases}
\]
The corresponding classical variational principle is
\[
    -\log z
    =
    \inf_{\nu}
    \left\{
    \cH_{\cl}(\nu,\mu_0)
    +
    \int\cD_0[u]\dd\nu(u)
    \right\},
\]
where the infimum is taken over probability measures and
\[
    \cH_{\cl}(\nu,\mu)
    =
    \begin{cases}
    \displaystyle
    \int
    \log\!\left(\frac{\dd\nu}{\dd\mu}\right)\dd\nu,
    & \nu\ll\mu,
    \\[0.6em]
    +\infty,
    & \text{otherwise}.
    \end{cases}
\]
The minimizers in these two variational principles are the corresponding
quantum Gibbs state and classical Gibbs measure, respectively.

\subsection{Main results}

The following quantity will appear in the smallness condition for the three-dimensional case:

\begin{equation}
    \Theta_{d,\delta}
    =
    \frac{1+\delta}{6}(2\pi)^d
    \int_\Omega\norm{J_\omega}_{L^\infty}^3\dd\nu(\omega).
\end{equation}

\begin{assumption}\label{ass:J}
The measure space $(\Omega,\nu)$ is $\sigma$-finite.  The family
$(J_\omega)_{\omega\in\Omega}$ admits a jointly measurable representative
\[
    J:\Omega\times\bT^d\longrightarrow[0,\infty),
    \qquad
    J(\omega,x)=J_\omega(x),
\]
and the map $\omega\mapsto\norm{J_\omega}_{L^\infty(\bT^d)}$ is
measurable.  For $\nu$-a.e. $\omega$, $J_\omega$ is real-valued,
nonnegative, and belongs to $L^\infty(\bT^d)$.  Moreover,
\begin{equation}\label{eq:J-Linf-assumption}
    \int_\Omega \norm{J_\omega}_{L^\infty(\bT^d)}^3\dd\nu(\omega)<\infty.
\end{equation}
If $d=3$, we assume in addition that the three-body interaction is small
enough in the sense that $\Theta_{3,\delta}<1/43200$ for some $\delta>0$.
\end{assumption}

We can now state the main theorem.

\medskip

\begin{theorem}[Quantum-to-classical limit]\label{thm:quantum-classical-main}
Let $d=2,3$. Let the three body interaction $v$ satisfies \eqref{eq:def-int-v} and assume Assumption~\ref{ass:J}. Consider the Gibbs state $\Gamma_\lambda=Z_\lambda^{-1}e^{-\bH_{\lambda}^{{\rm ren}}}$ associated with the renormalized Hamiltonian $\bH_{\lambda}^{{\rm ren}}$ in \eqref{eq:def-H-ren} with $h=-\Delta+1$. 
Let $\Gamma_0=Z_0^{-1}e^{-\lambda\dG(-\Delta+1)}$ be the free Gibbs state as in \eqref{eq:def-free-gibbs-state}. Let $\mu_0$ be the Gaussian free field with covariance $(-\Delta+1)^{-1}$ and let $\dd\mu=z^{-1}e^{-\cD_0[u]}\dd\mu_0$ be the associated nonlinear Gibbs measure as in \eqref{eq:def-nonlin-gibbs-measure}. Then,
\begin{equation}\label{eq:free-energy-main-limit}
    -\log\frac{Z_\lambda}{Z_0}
    \to
    -\log z=-\log\Big( \int e^{-\cD_0[u]}\dd\mu_0(u) \Big)
    \qquad(\lambda\downarrow0).
\end{equation}
Moreover, for every fixed $k\geq1$,
\begin{equation}\label{eq:density-main-limit}
    k!\lambda^k\Gamma_\lambda^{(k)}
    \to
    \gamma_\mu^{(k)}=\int \ket{u^{\otimes k}}\bra{u^{\otimes k}}\dd\mu (u)
    \qquad\text{in }\gS^2(\gH^{\otimes_s k}).
\end{equation}
\end{theorem}
The two assertions are proved in Sections~\ref{sec:Q-C-free-energy}
and~\ref{sec:Q-C-density}, respectively.

\medskip

\begin{remark}
    Our method also applies to quadratic channel interactions, which contains the
positive-type interactions on the torus considered in \cite{LNR21}. In the quadratic case, positivity removes the need
for exponential estimates and the proof is much shorter, see Section~\ref{sec:strategy}. 
\end{remark}

\begin{remark}[On the local limit]\label{rem:local-Phi6}
Formally replacing the nonlocal three-body kernel by
delta function leads, up to lower-order counterterms, to the
local \(\Phi_d^6\) measure.  In \(d=3\), this is energy-critical, and
the smallness condition in Assumption~\ref{ass:J} is incompatible with the
approximate-identity scaling required for a direct delta function-limit; while in
\(d=2\), no smallness condition is imposed, and the quantitative errors,
which are controlled by powers of \(\lambda^2N_0\) and terms vanishing with
\(\lambda\) (see Section~\ref{sec:Q-C-free-energy} and \ref{sec:Q-C-density}), leave room for a joint limit
\(\varepsilon=\varepsilon(\lambda)\downarrow0\).  We do not pursue this
problem here. For related results, we refer to the work of Nam, Yang and Zhu \cite{NYZ26}, concerning $P(\Phi)_2$ measure arising from bosonic Gibbs states with general $p$-body interactions in dimension two, and the references therein.
\end{remark}

\medskip

Combining the pointwise kernel domination in
Proposition~\ref{prop:Ginibre-domination} with the Hilbert--Schmidt
convergence above, one may further upgrade the convergence of the
fixed-order reduced density matrices to convergence of their integral kernels, which yields the same range of integral-kernel $L^p$-convergence as in
\cite{frohlich2022mean}.  More precisely, we have

\begin{corollary}\label{co:Lr-kernal-convergence}
  Let the integral kernels be
viewed as functions on \((\bT^d)^{2k}\), then
\begin{equation}\label{eq:kernal-Lr-convergence}
    k!\lambda^k
    \Gamma_\lambda^{(k)}(\,\cdot\,;\,\cdot\,)
    \longrightarrow
    \gamma_\mu^{(k)}(\,\cdot\,;\,\cdot\,)
    \quad\text{in }L^r((\bT^d)^{2dk})
\end{equation}
for \(1\leq r<\infty\) when \(d=2\), and for \(1\leq r<3\) when \(d=3\).
\end{corollary}
The proof is given in Section~\ref{sec:Q-C-density}.

\bigskip

\section{Strategy of the proof}\label{sec:strategy}

We explain the main idea of the proof before entering the technical
estimates.  The argument is variational, as in \cite{LNR21}, and relies on
the Gibbs variational principle.

The main difference is that we do not
extract a classical measure from the interacting Gibbs state through lower
symbols and a quantitative de Finetti theorem.  Instead, we use the
coherent-state representation of the free Gibbs state,
\[
    \Gamma_0
    =
    \int
    \left|\xi(u/\sqrt\lambda)\right\rangle
    \left\langle\xi(u/\sqrt\lambda)\right|
    \dd\mu_\lambda(u),
\]
where $\mu_\lambda$ is the Gaussian measure with covariance
$\mathsf C_\lambda=\lambda(e^{\lambda h}-1)^{-1}$, see Lemma~\ref{lem:free-coherent-representation}.  The measure
$\mu_\lambda$ is a regularized version of $\mu_0$, and the quantum
interaction is compared directly with its classical coherent-state symbol.

We first illustrate the method in the simpler quadratic channel setting.
In this case the channels $J_\omega$ may be complex-valued. The corresponding positive-type two-body interaction is
\[
    w(x-y)
    =
    \int_\Omega\dd\nu(\omega)
    \int_{\bT^d}
    J_\omega(x-r)\overline{J_\omega(y-r)}\dd r .
\]
We define the renormalized quantum interaction by
\[
    \bW_{\lambda,\mathrm{quad}}^{\ren}
    =
    \frac12
    \int_\Omega\dd\nu(\omega)
    \int_{\bT^d}
    \left|
    \lambda\dG(\tau_rJ_\omega)
    -
    \lambda\tr_{\gH}\big((\tau_rJ_\omega)\gamma_0\big)
    \right|^2
    \dd r,
\]
where $|X|^2=X^*X$.  Its classical counterpart is
\[
    \cD_{\lambda,\mathrm{quad}}[u]
    =
    \frac12
    \int_\Omega\dd\nu(\omega)
    \int_{\bT^d}
    \left|
    \cM_\lambda^{\tau_rJ_\omega}[u]
    \right|^2
    \dd r ,
\]
where $\cM_{\lambda}^{\tau_rJ_\omega}$ is defined in
Lemma~\ref{lem:Mass-Lp}; see also
Remarks~\ref{rem:M-lambda-positive} and
\ref{rem:Mass-L2-complex-channels}.

Let
\[
    Z_{\lambda,\mathrm{quad}}
    =
    \tr_{\gF}
    \Big(e^{-\lambda\dG(h)-\bW_{\lambda,\mathrm{quad}}^{\ren}}\Big),
    \qquad
    z_{\lambda,\mathrm{quad}}
    =
    \int e^{-\cD_{\lambda,\mathrm{quad}}[u]}\dd\mu_\lambda(u),
\]
and set
\[
    z_{\mathrm{quad}}
    :=
    \int e^{-\cD_{0,\mathrm{quad}}[u]}\dd\mu_0(u).
\]
The corresponding quantum state and limiting classical measure are
\[
\begin{aligned}
    \Gamma_{\lambda,\mathrm{quad}}
    =
    Z_{\lambda,\mathrm{quad}}^{-1}
    e^{-\lambda\dG(h)-\bW_{\lambda,\mathrm{quad}}^{\ren}},\qquad
    \dd\mu_{\mathrm{quad}}(u)
    =
    z_{\mathrm{quad}}^{-1}
    e^{-\cD_{0,\mathrm{quad}}[u]}\dd\mu_0(u).
\end{aligned}
\]

Under the sole assumption
\[
    \int_\Omega\norm{J_\omega}_{L^\infty}^2\dd\nu(\omega)<\infty,
\]
the quadratic model satisfies the analogues of the two conclusions in
Theorem~\ref{thm:quantum-classical-main}.

\begin{remark}
This is the only assumption imposed on the quadratic channel interaction,
both in $d=2$ and in $d=3$; in particular, no smallness condition is
required.  As discussed in the Introduction, this contains the
positive-type interactions on the torus considered in \cite{LNR21}.
\end{remark}

\medskip

\begin{proof}[\textbf{Sketch of the proof}] 
Since this interaction is nonnegative, the convergence
\[z_{\lambda,\mathrm{quad}}\to z_{\mathrm{quad}}\] follows by the same
$L^1$-argument as Proposition~\ref{prop:tilted-density-common}, with no
exponential-integrability estimate.

\noindent\textbf{Relative free-energy lower bound:}
The lower bound for the relative free energy follows from
Golden--Thompson and the coherent-state representation.  First,
\[
\begin{aligned}
    &\quad
    \frac{Z_{\lambda,\mathrm{quad}}}{Z_0}
    -
    z_{\mathrm{quad}}
    \leq
    \tr_{\gF}
    \big(e^{-\bW_{\lambda,\mathrm{quad}}^{\ren}}\Gamma_0\big)
    -
    z_{\lambda,\mathrm{quad}}
    +
    z_{\lambda,\mathrm{quad}}
    -
    z_{\mathrm{quad}}
    \\
    &=
    \int
    \left[
    \left\langle
    \xi(u/\sqrt\lambda),
    e^{-\bW_{\lambda,\mathrm{quad}}^{\ren}}
    \xi(u/\sqrt\lambda)
    \right\rangle
    -
    e^{-\cD_{\lambda,\mathrm{quad}}[u]}
    \right]\dd\mu_\lambda(u)
    +
    o(1).
\end{aligned}
\]
For fixed $u\in L^2(\bT^d)$, let $\dd\rho_{u,\lambda}$ be the spectral
measure of $\bW_{\lambda,\mathrm{quad}}^{\ren}$ in the vector
$\xi(u/\sqrt\lambda)$, namely
\[
    \int_{\R}F(x)\dd\rho_{u,\lambda}(x)
    =
    \left\langle
    \xi(u/\sqrt\lambda),
    F(\bW_{\lambda,\mathrm{quad}}^{\ren})
    \xi(u/\sqrt\lambda)
    \right\rangle
\]
for every nonnegative Borel function $F$.  
Together with $\cD_{\lambda,\mathrm{quad}}[u]\geq0$, the spectral theorem
and the Lipschitz bound
$\abs{e^{-x}-e^{-y}}\leq\abs{x-y}$ for $x,y\geq0$ give
\[
\begin{aligned}
&
\left\langle
\xi(u/\sqrt\lambda),
e^{-\bW_{\lambda,\mathrm{quad}}^{\ren}}
\xi(u/\sqrt\lambda)
\right\rangle
-
e^{-\cD_{\lambda,\mathrm{quad}}[u]}
\\
&=
\int_{\R}
\left(e^{-x}-e^{-\cD_{\lambda,\mathrm{quad}}[u]}\right)
\dd\rho_{u,\lambda}(x)
\\
&\leq
\int_{\R}
\left|e^{-x}-e^{-\cD_{\lambda,\mathrm{quad}}[u]}\right|
\dd\rho_{u,\lambda}(x)
\\
&\leq
\int_{\R}
\left|x-\cD_{\lambda,\mathrm{quad}}[u]\right|
\dd\rho_{u,\lambda}(x)
\\
&=
\left\langle
\xi(u/\sqrt\lambda),
\left|
\bW_{\lambda,\mathrm{quad}}^{\ren}
-
\cD_{\lambda,\mathrm{quad}}[u]
\right|
\xi(u/\sqrt\lambda)
\right\rangle .
\end{aligned}
\]
The self-adjoint operators
\(
\left|
\lambda\dG(\tau_rJ_\omega)
-
m_{\lambda,\omega}
\right|^2
-
\big|
\cM_\lambda^{\tau_rJ_\omega}[u]
\big|^2
\)
belong, as $(r,\omega)$ varies, to the same abelian subalgebra,
because the one-body channel operators are multiplication operators.  Hence,
in the joint spectral representation,
\begin{align}
&
\left\langle
\xi(u/\sqrt\lambda),
\left|
\bW_{\lambda,\mathrm{quad}}^{\ren}
-
\cD_{\lambda,\mathrm{quad}}[u]
\right|
\xi(u/\sqrt\lambda)
\right\rangle
\nn\\
&\leq
\frac12
\int_\Omega\dd\nu(\omega)
\int_{\bT^d}\dd r
\left\langle
\xi(u/\sqrt\lambda),
\big|
|
\lambda\dG(\tau_rJ_\omega)
-
m_{\lambda,\omega}
|^2
-
|
\cM_\lambda^{\tau_rJ_\omega}[u]
|^2
\big|
\xi(u/\sqrt\lambda)
\right\rangle .\label{eq:quad-upper-channel}
\end{align}
We next estimate the right-hand side of \eqref{eq:quad-upper-channel}.
The coherent-state identities give
\begin{equation}\label{eq:quad-indentity-1}
    \left\langle
    \xi(u/\sqrt\lambda),\big|
    \dG(\tau_r J_\omega)
    \big\vert^2
    \xi(u/\sqrt\lambda)\right\rangle
    =
    \lambda^{-2}\big|\inprod{u}{\tau_r J_\omega u}\big\vert^2+\lambda^{-1}\inprod{u}{|\tau_r J_\omega\vert^2 u},
\end{equation}
\begin{equation*}
    \left\langle
    \xi(u/\sqrt\lambda),
    \dG(\tau_r J_\omega)
    \xi(u/\sqrt\lambda)\right\rangle
    =\lambda^{-1}\inprod{u}{\tau_r J_\omega u},
\end{equation*}
which imply 
\begin{equation}\label{eq:quad-channel-identity}
   \begin{aligned}
\left\langle
\xi(u/\sqrt\lambda),
\left|
\lambda\dG(\tau_rJ_\omega)
-
m_{\lambda,\omega}
\right|^2
\xi(u/\sqrt\lambda)
\right\rangle
=\big|\cM_\lambda^{\tau_r J_\omega}[u]\big|^2+\lambda \inprod{u}{\left|\tau_r J_\omega\right|^2u}.
   \end{aligned}
\end{equation}
For each channel, the scalar inequality 
$\big||z|^2-|a|^2\big|\leq2|a||z-a|+|z-a|^2$ and \eqref{eq:quad-channel-identity} give
\[
\begin{aligned}
&
\left\langle
\xi(u/\sqrt\lambda),
\Big|
\left|
\lambda\dG(\tau_rJ_\omega)
-
m_{\lambda,\omega}
\right|^2
-
\big|
\cM_\lambda^{\tau_rJ_\omega}[u]
\big|^2
\Big|
\xi(u/\sqrt\lambda)
\right\rangle
\\
&\leq
2\big|\cM_\lambda^{\tau_rJ_\omega}[u]\big|
\left\langle
\xi(u/\sqrt\lambda),
\left|
\lambda\dG(\tau_rJ_\omega)
-
\inprod{u}{\tau_r J_\omega u}
\right|
\xi(u/\sqrt\lambda)
\right\rangle\\
&+
\left\langle
\xi(u/\sqrt\lambda),
\left|
\lambda\dG(\tau_rJ_\omega)
-
\inprod{u}{\tau_r J_\omega u}
\right|^2
\xi(u/\sqrt\lambda)
\right\rangle\\
&\leq
2\big|\cM_\lambda^{\tau_rJ_\omega}[u]\big|
\left(\lambda\inprod{u}{|\tau_rJ_\omega|^2u}\right)^{1/2}
+
\lambda\inprod{u}{|\tau_rJ_\omega|^2u}.
\end{aligned}
\]
The last inequality follows from the spectral theorem, the
Cauchy--Schwarz inequality, and \eqref{eq:quad-indentity-1}.

After integration in $u,r,\omega$, Lemma~\ref{lem:Mass-Lp} with $p=2$ and the Cauchy--Schwarz inequality give
\[  
    \frac{Z_{\lambda,\mathrm{quad}}}{Z_0}
    -
    z_{\mathrm{quad}}\lesssim
    \Big(\int_\Omega\norm{J_\omega}_{L^\infty}^2\dd\nu(\omega)\Big)
    \big[
    (\tr_{\gH}h^{-2})^{1/2}(\lambda^2N_0)^{1/2}
    +
    \lambda^2N_0
    \big]+o(1),
\]
which tends to zero by the asymptotic behavior of $N_0$, see Lemma~\ref{lem:Clambda-trace-bound}. Hence,
\[
    \liminf_{\lambda\downarrow0}
    \Big(-\log\frac{Z_{\lambda,\mathrm{quad}}}{Z_0}\Big)
    \geq
    -\log z_{\mathrm{quad}},
\]
and this gives the relative free-energy lower bound.

\medskip

\noindent\textbf{Relative free-energy upper bound:}
For the upper bound, define the coherent trial state by
\[
    \widetilde\Gamma_{\lambda,\mathrm{quad}}
    =
    \int
    \left|\xi(u/\sqrt\lambda)\right\rangle
    \left\langle\xi(u/\sqrt\lambda)\right|
    \dd\tilde{\mu}_\lambda (u),\qquad 
    \dd\tilde{\mu}_\lambda(u)=
    z_{\lambda,\mathrm{quad}}^{-1}
    e^{-\cD_{\lambda,\mathrm{quad}}[u]}
    \dd\mu_\lambda(u).
\]
Lemma~\ref{lem:relative-entropy-Berezin-Lieb} gives
\[
\begin{aligned}
    \cH(\widetilde\Gamma_{\lambda,\mathrm{quad}},\Gamma_0)
    \leq\cH_{\cl}(\tilde{\mu}_\lambda,\mu_\lambda)=
    -\log z_{\lambda,\mathrm{quad}}
    -
    \int
    \cD_{\lambda,\mathrm{quad}}[u]
    \dd\tilde{\mu}_\lambda(u).
\end{aligned}
\]
Moreover, \eqref{eq:quad-channel-identity} gives
\[
    \tr_{\gF}
    \big(\bW_{\lambda,\mathrm{quad}}^{\ren}
    \widetilde\Gamma_{\lambda,\mathrm{quad}}\big)
    =
    \int
    \cD_{\lambda,\mathrm{quad}}[u]
    \dd\tilde{\mu}_\lambda(u)
    +
    \frac{\lambda}{2}
\int_\Omega\dd\nu(\omega)\int_{\bT^d}
\int_{\gH}
\bangle{u,|\tau_rJ_\omega|^2u}\dd\tilde{\mu}_\lambda(u)\dd r.
\]
Since $\cD_{\lambda,\mathrm{quad}}\geq0$ and $z_{\lambda,\mathrm{quad}}\to z_{\mathrm{quad}}>0$,
\[
\begin{aligned}
&\frac{\lambda}{2}
\int_\Omega\dd\nu(\omega)\int_{\bT^d}
\int_{\gH}
\bangle{u,|\tau_rJ_\omega|^2u}\dd\tilde{\mu}_\lambda(u)\dd r
\\
&\quad\leq
\frac{\lambda}{2z_{\lambda,\mathrm{quad}}}
\int_\Omega\dd\nu(\omega)\int_{\bT^d}
\int_{\gH}
\bangle{u,|\tau_rJ_\omega|^2u}\dd\mu_\lambda(u)\dd r
\\
&\quad\leq
\frac{(2\pi)^d\lambda^2 N_0}{2z_{\lambda,\mathrm{quad}}}
\int_\Omega\norm{J_\omega}_{L^\infty}^2\dd\nu(\omega)
\longrightarrow0 .
\end{aligned}
\]
The Gibbs variational principle therefore yields
\[
\begin{aligned}
    -\log\frac{Z_{\lambda,\mathrm{quad}}}{Z_0}
    &\leq
    \cH(\widetilde\Gamma_{\lambda,\mathrm{quad}},\Gamma_0)
    +
    \tr_{\gF}(\bW_{\lambda,\mathrm{quad}}^{\ren}\widetilde\Gamma_{\lambda,\mathrm{quad}})
    \\
    &\leq\cH_{\cl}(\tilde{\mu}_\lambda,\mu_\lambda)+\int
    \cD_{\lambda,\mathrm{quad}}[u]
    \dd\tilde{\mu}_\lambda(u)+o(1)\\
    &=
    -\log z_{\lambda,\mathrm{quad}}+o(1)
    \longrightarrow -\log z_{\mathrm{quad}} .
\end{aligned}
\]
Together with the lower bound, this proves the relative free-energy
convergence in the quadratic case.

\medskip 

\noindent\textbf{Hilbert--Schmidt convergence:}
As in \cite{LNR21}, we will use the following estimate:
\begin{align*}
\norm{\Gamma_{\lambda,\mathrm{quad}}^{(k)}-\widetilde\Gamma_{\lambda,\mathrm{quad}}^{(k)}}_{\gS^2}^2
\lesssim_k
\norm{\Gamma_{\lambda,\mathrm{quad}}-\widetilde\Gamma_{\lambda,\mathrm{quad}}}_{\gS^1}
\sum_{\ell=k}^{2k}
\left(
\norm{\Gamma_{\lambda,\mathrm{quad}}^{(\ell)}}_{\gS^2}
+
\norm{\widetilde\Gamma_{\lambda,\mathrm{quad}}^{(\ell)}}_{\gS^2}
\right).
\end{align*}

The quadratic interaction is
nonnegative, and the positivity of the heat kernel and the Trotter product
formula give the pointwise kernel domination
\[
    0\leq
    \Gamma_{\lambda,\mathrm{quad}}^{(\ell)}
    (\underline X_\ell;\underline Y_\ell)
    \leq
    C_\ell\,
    \Gamma_0^{(\ell)}
    (\underline X_\ell;\underline Y_\ell).
\]
This supplies the uniform Hilbert--Schmidt bounds required in
Lemma~11.4 of \cite{LNR21}.  Combining these bounds with the trace-norm
convergence of the Gibbs state to the coherent trial state gives the
convergence of reduced density matrices of every fixed order.  
\end{proof}

\medskip

\noindent\textbf{Additional issues in the cubic case.}
The three-body interaction treated in this paper follows the same coherent-state
principle, but two new issues appear.  First, neither the centered classical
cubic channel interaction nor the centered quantum three-body interaction is
nonnegative.  The convergence of the Gibbs weights therefore requires 
exponential estimates, both in the classical and the quantum
settings; these are proved in Propositions~\ref{prop:full-exp} and
\ref{prop:quantum-exponential}.  Second, the pointwise kernel domination for the
interacting density matrices can no longer be obtained directly from
the Trotter product formula.  The substitute is the monotonicity of the
three-body interaction, combined with the Ginibre loop
representation in Appendix~\ref{app:Ginibre-loop-representation}.  This
yields the pointwise kernel domination in Proposition~\ref{prop:Ginibre-domination},
which plays the corresponding role in the quadratic case.

\bigskip

\section{Classical model: properties of the nonlinear Gibbs measures}\label{sec:classical-model}

In this section, we construct the nonlinear Gibbs measure and establish
the estimates needed below.

\subsection{Gaussian fields with general covariance}

We first recall the notation for Gaussian fields with a general covariance
operator.  Let $\mathsf C$ be a nonnegative self-adjoint compact operator
on $\mathfrak H$, and assume that $\mathsf C\in\mathfrak S^p(\mathfrak H)$
for some $p\ge 1$. Let $(\varphi_j,c_j)$ be a spectral decomposition of the positive
spectral subspace of $\mathsf C$, so that
\[
    \mathsf C\varphi_j=c_j\varphi_j,
    \qquad c_j>0,
\]
We denote by $\mu_{\mathsf C}$ the centered complex Gaussian field with
covariance $\mathsf C$.  In spectral coordinates, this field may be
realized as
\[
    u_{\mathsf C}
    =
    \sum_j c_j^{1/2}g_j\varphi_j,
\]
where $(g_j)_j$ are independent standard complex Gaussian variables.

When $p=1$, the field is $\mathfrak H$-valued.  For $p>1$, the series
above is not necessarily convergent in $\mathfrak H$.  It is naturally
realized as a random element of the Hilbert space
\[
    \mathfrak H_{\mathsf C}^{1-p}
    :=
    \left\{
    u=\sum_j u_j\varphi_j:
    \sum_j c_j^{p-1}|u_j|^2<\infty
    \right\},
\]
because
\[
    \int_{\mathfrak H_{\mathsf C}^{1-p}}
    \|u\|_{\mathfrak H_{\mathsf C}^{1-p}}^2
    \,\mathrm d\mu_{\mathsf C}(u)
    =
    \sum_j c_j^p<\infty.
\]
In the application $\mathsf C=h^{-1}$, this is the usual Sobolev-type
space
\[
    \mathfrak H_{h^{-1}}^{1-p}
    =
    \left\{
    u=\sum_{k\in\mathbb Z^d}u_ke_k:
    \sum_{k\in\mathbb Z^d}
    (1+|k|^2)^{1-p}|u_k|^2<\infty
    \right\}.
\]

We shall use the following standard consequence of Wick's theorem.  If
$P$ is a finite-rank spectral projection of $\mathsf C$, then
\[
    \int_{\mathfrak H_{\mathsf C}^{1-p}}
    |(Pu)^{\otimes k}\rangle
    \langle (Pu)^{\otimes k}|
    \,\mathrm d\mu_{\mathsf C}(u)
    =
    k!\,(P\mathsf C P)^{\otimes k}
\]
as an operator on $(P\mathfrak H)^{\otimes_s k}$.  Passing to the limit in $\mathfrak S^p$ yields
\[
    \int_{\mathfrak H_{\mathsf C}^{1-p}}
    |u^{\otimes k}\rangle
    \langle u^{\otimes k}|
    \,\mathrm d\mu_{\mathsf C}(u)
    =
    k!\,\mathsf C^{\otimes k}
\]
as an element of $\mathfrak S^p(\mathfrak H^{\otimes_s k})$, with
\[
    \left\|
    k!\,\mathsf C^{\otimes k}
    \right\|_{\mathfrak S^p(\mathfrak H^{\otimes_s k})}
    \le
    k!\,\|\mathsf C\|_{\mathfrak S^p(\mathfrak H)}^k.
\]
A rigorous formulation of this identity is given in
\cite[Section~3.1]{LNR15}. Throughout this paper we consider the following covariances
\[
    \mathsf C_\lambda
    =
    \lambda(e^{\lambda h}-1)^{-1}=\lambda\gamma_0, \quad\lambda>0\quad;
    \qquad
    \mathsf C_0=h^{-1}.
\]
In these two cases we keep the shorter notation
\[
    \mu_{\mathsf C_\lambda}=\mu_\lambda,
    \qquad
    \mu_{h^{-1}}=\mu_0.
\]

We next define the finite-dimensional centered interaction. For $K<\infty$ and $\lambda\in[0,1]$, define
\begin{align}
    \cD_{\lambda,K}[u]
    =
    \frac16
    \int_\Omega\dd\nu(\omega)
    \int_{\bT^d}
    \left(\cM_{\lambda,K}^{\tau_rJ_\omega}[u]\right)^3
    \dd r ,
\end{align}
where
\[
    \cM_{\lambda,K}^A[u]
    =
    \inprod{P_Ku}{AP_Ku}
    -
    \mathbb E_{\mu_\lambda}\inprod{P_Ku}{AP_Ku},
    \qquad
    P_K=\mathds1_{\{h\leq K\}}.
\]
We then remove the cutoff and establish the exponential integrability
needed to normalize the Gibbs measure.

We begin with the uniform $L^p$-control and convergence
of the mass renormalization.

\subsection{Mass renormalization}

\medskip

\begin{lemma}[Mass renormalization in $L^p$]\label{lem:Mass-Lp}
Let $A$ be a bounded self-adjoint operator on $\gH$.  Then
$\cM_{\lambda,K}^A$ is a Cauchy sequence in $L^p(\mu_\lambda)$, uniformly for
$\lambda\in[0,1]$, for every $1\leq p<\infty$. The limit is denoted by $\cM_{\lambda}^A$. Moreover,
\begin{equation}\label{eq:MK-uniform-Lp}
    \sup_K\sup_{\lambda\in[0,1]}
    \norm{\cM_{\lambda,K}^A}_{L^p(\mu_\lambda)}
    \leq
    \left[2\left(1+p\,4^p\Gamma(p)\right)\right]^{1/p}
    \norm{A}
    \left(\tr h^{-2}\right)^{1/2}.
\end{equation}
In particular, for every $\omega,r$,
\begin{equation}\label{eq:MK-channel-uniform-Lp}
    \sup_K\sup_{\lambda\in[0,1]}
    \norm{\cM_{\lambda,K}^{\tau_rJ_\omega}}_{L^p(\mu_\lambda)}
    \leq
    \left[2\left(1+p\,4^p\Gamma(p)\right)\right]^{1/p}
    \norm{J_\omega}_{L^\infty}
    \left(\tr h^{-2}\right)^{1/2}.
\end{equation}
\end{lemma}

\begin{remark}\label{rem:M-lambda-positive}
    Since $\mu_\lambda$ is supported on $\gH$ for $\lambda>0$, then, for $u\in\gH$,
    \[
        \cM_{\lambda}^A[u]
        =
        \inprod{u}{Au}-\tr_{\gH}(A\mathsf{C}_\lambda).
    \]
    In particular, for $\lambda>0$ and $u\in\gH$,
    \[
        \cM_{\lambda}^{\tau_rJ_\omega}[u]
        =
        \inprod{u}{\tau_rJ_\omega u}-\tr_{\gH}((\tau_rJ_\omega)\mathsf{C}_\lambda)
        =
        \inprod{u}{\tau_rJ_\omega u}-m_{\lambda,\omega}.   
    \]
\end{remark}

\begin{remark}\label{rem:Mass-L2-complex-channels}
For the quadratic channel interaction, only the $L^2$ version of the
mass renormalization is needed.  In that case the self-adjointness
assumption on $A$ can be removed, and the estimate follows from the
Hilbert--Schmidt isometry for centered complex Gaussian quadratic forms;
see \cite[Lemma~5.2]{LNR21}.  Consequently, in the quadratic setting our
method allows the channel functions $J_\omega$ to be complex-valued
bounded functions.
\end{remark}

\medskip

\begin{proof}
It is enough to consider $0<\lambda\leq1$; the case $\lambda=0$ is the
same argument with $\mathsf{C}_{0}=h^{-1}$.  For finite $K$, write
\[
    P_Ku=P_K\mathsf{C}_{\lambda}^{1/2}g,
\]
where $g$ is the standard complex Gaussian vector in $P_K\gH$.  Then
\begin{align*}
\cM_{\lambda,K}^A
&=
\inprod{P_K\mathsf{C}_{\lambda}^{1/2}g}{A P_K\mathsf{C}_{\lambda}^{1/2}g}
-
\tr\left(AP_K\mathsf{C}_{\lambda} P_K\right)                                           \\
&=
\inprod{g}{P_K\mathsf{C}_{\lambda}^{1/2}A\mathsf{C}_{\lambda}^{1/2}P_Kg}
-
\tr\left(P_K\mathsf{C}_{\lambda}^{1/2}A\mathsf{C}_{\lambda}^{1/2}P_K\right).
\end{align*}
The finite-rank operator
\[
    P_K\mathsf{C}_{\lambda}^{1/2}A\mathsf{C}_{\lambda}^{1/2}P_K
\]
is self-adjoint.  If its eigenvalues are $\alpha_{\lambda,j,K}$, unitary
invariance of the standard complex Gaussian vector gives
\[
    \cM_{\lambda,K}^A
    \stackrel{\mathrm{law}}=
    \sum_j\alpha_{\lambda,j,K}(|g_j|^2-1).
\]
Moreover, by the Schatten Hölder inequality and $\mathsf{C}_{\lambda}\leq h^{-1}$,
\begin{align*}
\sum_j\alpha_{\lambda,j,K}^2
=
\norm{P_K\mathsf{C}_{\lambda}^{1/2}A\mathsf{C}_{\lambda}^{1/2}P_K}_{\gS^2}^2                    
\leq
\norm A^2\tr(\mathsf{C}_{\lambda}^2)                                                   
\leq
\norm A^2\tr h^{-2}.
\end{align*}
Lemma~\ref{lem:Gaussian-Lp} therefore gives
\[
    \norm{\cM_{\lambda,K}^A}_{L^p(\mu_\lambda)}
    \leq
    \left[2\left(1+p4^p\Gamma(p)\right)\right]^{1/p}
    \norm{A}(\tr h^{-2})^{1/2}.
\]
Taking the supremum over $K$ and $\lambda\in[0,1]$ proves
\eqref{eq:MK-uniform-Lp}.

It remains to prove the Cauchy property.  For $L\geq K$,
\begin{align*}
\cM_{\lambda,L}^A-\cM_{\lambda,K}^A
&=
\inprod{g}{
\left(
P_L\mathsf{C}_{\lambda}^{1/2}A\mathsf{C}_{\lambda}^{1/2}P_L
-
P_K\mathsf{C}_{\lambda}^{1/2}A\mathsf{C}_{\lambda}^{1/2}P_K
\right)g}                                                                    \\
&\quad-
\tr\left(
P_L\mathsf{C}_{\lambda}^{1/2}A\mathsf{C}_{\lambda}^{1/2}P_L
-
P_K\mathsf{C}_{\lambda}^{1/2}A\mathsf{C}_{\lambda}^{1/2}P_K
\right).
\end{align*}
Applying Lemma~\ref{lem:Gaussian-Lp} to the eigenvalues of the self-adjoint
difference gives
\begin{align}
&
\norm{\cM_{\lambda,L}^A-\cM_{\lambda,K}^A}_{L^p(\mu_\lambda)}
\nn\\
&\leq
\left[2\left(1+p4^p\Gamma(p)\right)\right]^{1/p}
\norm{
P_L\mathsf{C}_{\lambda}^{1/2}A\mathsf{C}_{\lambda}^{1/2}P_L
-
P_K\mathsf{C}_{\lambda}^{1/2}A\mathsf{C}_{\lambda}^{1/2}P_K
}_{\gS^2}.
\label{eq:lambda-MK-Cauchy-HS}
\end{align}
Since $P_K$ commutes with $\mathsf{C}_{\lambda}$,
\begin{align*}
&
P_L\mathsf{C}_{\lambda}^{1/2}A\mathsf{C}_{\lambda}^{1/2}P_L
-
P_K\mathsf{C}_{\lambda}^{1/2}A\mathsf{C}_{\lambda}^{1/2}P_K                                      \\
&=
P_K\mathsf{C}_{\lambda}^{1/2}A(P_L-P_K)\mathsf{C}_{\lambda}^{1/2}
+
(P_L-P_K)\mathsf{C}_{\lambda}^{1/2}A\mathsf{C}_{\lambda}^{1/2}P_L .
\end{align*}
By the Schatten Hölder inequality,
\begin{align*}
\norm{
P_K\mathsf{C}_{\lambda}^{1/2}A(P_L-P_K)\mathsf{C}_{\lambda}^{1/2}
}_{\gS^2}                                                                   
&\leq
\norm A\,
\norm{P_K\mathsf{C}_{\lambda}^{1/2}}_{\gS^4}
\norm{(P_L-P_K)\mathsf{C}_{\lambda}^{1/2}}_{\gS^4}                                      \\
&\leq
\norm A\,
(\tr h^{-2})^{1/4}
\left(\tr((1-P_K)h^{-2})\right)^{1/4},
\end{align*}
and
\begin{align*}
\norm{
(P_L-P_K)\mathsf{C}_{\lambda}^{1/2}A\mathsf{C}_{\lambda}^{1/2}P_L
}_{\gS^2}                                                                    
&\leq
\norm A\,
\norm{(P_L-P_K)\mathsf{C}_{\lambda}^{1/2}}_{\gS^4}
\norm{\mathsf{C}_{\lambda}^{1/2}P_L}_{\gS^4}                                            \\
&\leq
\norm A\,
\left(\tr((1-P_K)h^{-2})\right)^{1/4}
(\tr h^{-2})^{1/4}.
\end{align*}
Consequently,
\begin{align*}
&
\sup_{0\leq\lambda\leq1}
\norm{
P_L\mathsf{C}_{\lambda}^{1/2}A\mathsf{C}_{\lambda}^{1/2}P_L
-
P_K\mathsf{C}_{\lambda}^{1/2}A\mathsf{C}_{\lambda}^{1/2}P_K
}_{\gS^2}                                                                    \\
&\leq
2\norm A
(\tr h^{-2})^{1/4}
\left(\tr((1-P_K)h^{-2})\right)^{1/4}
\to0
\end{align*}
as $K\to\infty$, uniformly in $L\geq K$.  Combining this with
\eqref{eq:lambda-MK-Cauchy-HS} proves that $\cM_{\lambda,K}^A$ is Cauchy in
$L^p(\mu_\lambda)$, uniformly for $\lambda\in[0,1]$. Taking $A=\tau_rJ_\omega$ and using
$
    \norm{\tau_rJ_\omega}=\norm{J_\omega}_{L^\infty}
$
proves \eqref{eq:MK-channel-uniform-Lp}.
\end{proof}

The previous lemma provides the uniform bounds required to pass from
individual centered densities to their cubic channel average.

\medskip

\begin{proposition}[Limit of the centered cubic channel interaction]\label{prop:D-L1}
For each $\lambda\in[0,1]$, one can choose a jointly measurable version of
\[
    (u,\omega,r)
    \longmapsto
    \cM_\lambda^{\tau_rJ_\omega}[u].
\]
With this choice, $\cD_{\lambda,K}$ converges in $L^1(\mu_\lambda)$,
uniformly for $\lambda\in[0,1]$, and its limit satisfies
\begin{equation}\label{eq:Dlambda-channel-representation}
    \cD_\lambda[u]
    =
    \frac16
    \int_\Omega\int_{\bT^d}
    \big(\cM_\lambda^{\tau_rJ_\omega}[u]\big)^3
    \dd r\,\dd\nu(\omega)
\end{equation}
in $L^1(\mu_\lambda)$ and, after fixing representatives, for
$\mu_\lambda$-a.e. $u$.
\end{proposition}

\medskip

\begin{proof}
For fixed $\lambda\in[0,1]$, set
\[
    F_{\lambda,K}(u,\omega,r)
    :=
    \cM_{\lambda,K}^{\tau_rJ_\omega}[u].
\]
The joint measurability in Assumption~\ref{ass:J} and the finite rank of
$P_K$ imply that $F_{\lambda,K}$ is jointly measurable.  Let
\[
    \mathfrak m_\lambda
    :=
    \mu_\lambda\otimes\nu\otimes\mathrm{Leb}_{\bT^d}.
\]
The estimates in the proof of Lemma~\ref{lem:Mass-Lp} give, for $L\geq K$,
\begin{align*}
\norm{
F_{\lambda,L}(\,\cdot\,,\omega,r)
-
F_{\lambda,K}(\,\cdot\,,\omega,r)
}_{L^3(\mu_\lambda)}
\lesssim
\norm{J_\omega}_{L^\infty}
(\tr h^{-2})^{1/4}
\left(\tr\big((1-P_K)h^{-2}\big)\right)^{1/4},
\end{align*}
uniformly in $\lambda$, $\omega$, and $r$.  Hence
$(F_{\lambda,K})_K$ is Cauchy in
$L^3(\mathfrak m_\lambda)$, uniformly in $\lambda$, because
\begin{align*}
&
\int_\Omega\int_{\bT^d}
\norm{
F_{\lambda,L}(\,\cdot\,,\omega,r)
-
F_{\lambda,K}(\,\cdot\,,\omega,r)
}_{L^3(\mu_\lambda)}^3
\dd r\,\dd\nu(\omega)
\\
&\qquad\lesssim
(2\pi)^d
(\tr h^{-2})^{3/4}
\left(\tr\big((1-P_K)h^{-2}\big)\right)^{3/4}
\int_\Omega\norm{J_\omega}_{L^\infty}^3\dd\nu(\omega),
\end{align*}
and the right-hand side tends to zero.  Denote the product-space limit by
$F_\lambda$.  By uniqueness of the $L^3$ limit, for
$\dd\nu(\omega)\dd r$-a.e. $(\omega,r)$ it agrees with
$\cM_\lambda^{\tau_rJ_\omega}$; this gives the asserted jointly measurable
version.

Lemma~\ref{lem:Mass-Lp} also yields
\[
    \sup_{K}\sup_{\lambda\in[0,1]}
    \int_\Omega\int_{\bT^d}
    \norm{F_{\lambda,K}(\,\cdot\,,\omega,r)}_{L^3(\mu_\lambda)}^3
    \dd r\,\dd\nu(\omega)
    \lesssim
    (2\pi)^d(\tr h^{-2})^{3/2}
    \int_\Omega\norm{J_\omega}_{L^\infty}^3\dd\nu(\omega).
\]
Consequently, using
\[
    |x^3-y^3|
    \leq
    |x-y|\big(|x|^2+|x||y|+|y|^2\big)
\]
and H\"older's inequality on the product space, we obtain
\[
    F_{\lambda,K}^3\longrightarrow F_\lambda^3
    \quad\text{in }
    L^1(\mathfrak m_\lambda),
\]
uniformly in $\lambda$.  Fubini's theorem therefore shows that
\[
    \cD_\lambda^\sharp[u]
    :=
    \frac16\int_\Omega\int_{\bT^d}
    F_\lambda(u,\omega,r)^3\dd r\,\dd\nu(\omega)
\]
is well-defined in $L^1(\mu_\lambda)$ and that
\[
    \sup_{\lambda\in[0,1]}
    \norm{\cD_{\lambda,K}-\cD_\lambda^\sharp}_{L^1(\mu_\lambda)}
    \longrightarrow0.
\]
We henceforth denote $\cD_\lambda^\sharp$ by $\cD_\lambda$, proving \eqref{eq:Dlambda-channel-representation}.
\end{proof}

\subsection{Exponential estimates}

\medskip

\begin{lemma}[One-channel exponential tail bound]\label{lem:Chernoff-channel}
Assume $\norm{J_\omega}_{L^\infty}>0$.  For every $K,\omega,r$ and every $R\geq0$,
\begin{equation}\label{eq:Chernoff-basic}
    \sup_{\lambda\in[0,1]}\mu_\lambda\Big(
    -\frac{\cM_{\lambda,K}^{\tau_rJ_\omega}}
    {\norm{J_\omega}_{L^\infty}}
    \geq R
    \Big)
    \leq
    \inf_{t\geq0}
    \exp\Big[
    -tR+
    \sum_{p\in\Z^d}
\min\left\{
\frac{t}{h(p)},
\frac{t^2}{2h(p)^2}
\right\}
    \Big].
\end{equation}
\end{lemma}

\medskip

\begin{proof}
The operator
\[
    P_K \mathsf{C}_{\lambda}^{1/2}(\tau_rJ_\omega) \mathsf{C}_{\lambda}^{1/2}P_K
\]
is finite-rank, self-adjoint and nonnegative.  Denote its eigenvalues by
$\alpha_{\lambda,j,K}^{\omega,r}$.  Since
\[
    \langle P_Ku,\tau_rJ_\omega P_Ku\rangle
    \stackrel{\mathrm{law}}=
    \sum_j\alpha_{\lambda,j,K}^{\omega,r}|g_j|^2,
\]
one has
\[
    \cM_{\lambda,K}^{\tau_rJ_\omega}
    \stackrel{\mathrm{law}}=
    \sum_j\alpha_{\lambda,j,K}^{\omega,r}(|g_j|^2-1).
\]
For $t\geq0$,
\begin{align*}
\mathbb E_{\mu_\lambda}
\exp\Big(
-t\frac{\cM_{\lambda,K}^{\tau_rJ_\omega}}{\norm{J_\omega}_{L^\infty}}
\Big)                                                                      
=
\prod_j
\frac{
\exp\left(
t\alpha_{\lambda,j,K}^{\omega,r}/\norm{J_\omega}_{L^\infty}
\right)
}{
1+t\alpha_{\lambda,j,K}^{\omega,r}/\norm{J_\omega}_{L^\infty}
}                                                                            
=
\exp\Big[
\sum_j
f\Big(
\frac{t\alpha_{\lambda,j,K}^{\omega,r}}{\norm{J_\omega}_{L^\infty}}
\Big)
\Big],
\end{align*}
where $f(x)=x-\log(1+x)$. The quadratic form inequality
\[
    0\leq
    P_K \mathsf{C}_{\lambda}^{1/2}(\tau_rJ_\omega) \mathsf{C}_{\lambda}^{1/2}P_K
    \leq
    \norm{J_\omega}_{L^\infty}P_Kh^{-1}P_K
\]
holds because $0\leq\tau_rJ_\omega\leq\norm{J_\omega}_{L^\infty}$.  Since
$f$ is increasing on $[0,\infty)$ and $f(x)\leq\min\{x,x^2/2\}$ for nonnegative $x$, the min--max principle gives
\[
    \sum_j
    f\Big(
    \frac{t\alpha_{j,K}^{\omega,r}}{\norm{J_\omega}_{L^\infty}}
    \Big)
    \leq
    \sum_{p:\,h(p)\leq K}f\Big(\frac{t}{h(p)}\Big)
    \leq
    \sum_{p\in\Z^d}
\min\left\{
\frac{t}{h(p)},
\frac{t^2}{2h(p)^2}
\right\}.
\]
Markov's inequality then yields, for every $R\geq0$,
\begin{align*}
\mu_\lambda\Big(
-\frac{\cM_{\lambda,K}^{\tau_rJ_\omega}}{\norm{J_\omega}_{L^\infty}}
\geq R
\Big)                                                                     
&=
\mu_\lambda\Big[
\exp\left(
-t\frac{\cM_{\lambda,K}^{\tau_rJ_\omega}}{\norm{J_\omega}_{L^\infty}}
\right)
\geq e^{tR}
\Big]                                                                      \\
&\leq
\exp\Big[
-tR+
\sum_{p\in\Z^d}
\min\left\{
\frac{t}{h(p)},
\frac{t^2}{2h(p)^2}
\right\}
\Big].
\end{align*}
Taking the infimum over $t\geq0$ proves \eqref{eq:Chernoff-basic}.
\end{proof}

Combining this estimate with the lattice bounds in
Lemma~\ref{lem:lattice-sums} yields the one-channel exponential moments
below.

\medskip

\begin{proposition}[One-channel exponential estimates]\label{prop:one-channel-exp}
Assume $\norm{J_\omega}_{L^\infty}>0$.

If $d=3$, then for every $0\leq\theta<1/43200$,
\begin{equation}
    \sup_{\lambda\in[0,1]}\mathbb E_{\mu_\lambda}\left[
    \exp\left(
    \theta
    \frac{\left(\cM_{\lambda,K}^{\tau_rJ_\omega}\right)_-^3}
    {\norm{J_\omega}_{L^\infty}^3}
    \right)\right]
    \leq
    \frac1{1-43200\,\theta}.
\end{equation}
If $d=2$, then for every $\theta>0$,
\begin{equation}\label{eq:one-channel-2d}
    \sup_{\lambda\in[0,1]}\mathbb E_{\mu_\lambda}\left[
    \exp\left(
    \theta
    \frac{\left(\cM_{\lambda,K}^{\tau_rJ_\omega}\right)_-^3}
    {\norm{J_\omega}_{L^\infty}^3}
    \right)\right]
    \leq
    2\exp\left(\theta R_\theta^3\right),
\end{equation}
where
\begin{equation}\label{eq:Rtheta-def}
    R_\theta
    =
    320\log\left(320\max\{e,\sqrt\theta\}\right).
\end{equation}
The estimates are uniform in $K,r,\omega$.
\end{proposition}

\medskip

\begin{proof}
For $d=3$, Lemmas~\ref{lem:Chernoff-channel} and
\ref{lem:lattice-sums} imply
\[
    \sup_{\lambda\in[0,1]}\mu_\lambda\Big(
    -\frac{\cM_{\lambda,K}^{\tau_rJ_\omega}}{\norm{J_\omega}_{L^\infty}}
    \geq R
    \Big)
    \leq
    \inf_{t\geq0}e^{-tR+80t^{3/2}} .
\]
For $R\geq0$, inserting $t=(R/120)^2$ and $s=R^3$ gives
\begin{align*}
    \sup_{\lambda\in[0,1]}\mu_\lambda\Big(
    \frac{(\cM_{\lambda,K}^{\tau_rJ_\omega})_-^3}
    {\norm{J_\omega}_{L^\infty}^3}
    \geq s
    \Big)
    \leq
    e^{-s/43200}.
\end{align*}
The layer-cake formula gives
\begin{align*}
&\sup_{\lambda\in[0,1]}\mathbb E_{\mu_\lambda}
\exp\left[
\theta
\frac{(\cM_{\lambda,K}^{\tau_rJ_\omega})_-^3}
{\norm{J_\omega}_{L^\infty}^3}
\right]                                                                     
\leq
1+
\theta\int_0^\infty
e^{\theta s}
\sup_{\lambda\in[0,1]}\mu_\lambda\Big(
\frac{(\cM_{\lambda,K}^{\tau_rJ_\omega})_-^3}{\norm{J_\omega}_{L^\infty}^3}
\geq s
\Big)\dd s                                                                \\
&\qquad\leq
1+\theta\int_0^\infty e^{-(1/43200-\theta)s}\dd s
=
\frac1{1-43200\theta}.
\end{align*}

For $d=2$, Lemmas~\ref{lem:Chernoff-channel} and
\ref{lem:lattice-sums} give
\[
    \sup_{\lambda\in[0,1]}\mu_\lambda\Big(
    -\frac{\cM_{\lambda,K}^{\tau_rJ_\omega}}{\norm{J_\omega}_{L^\infty}}
    \geq R
    \Big)
    \leq
    \inf_{t\geq0}e^{-tR+20t\log(2+t)} .
\]
When $R\geq80\log3$, choosing $t=e^{R/80}$ gives
\[
    \log(2+t)\leq\log(3t)=\log3+\frac{R}{80}\leq\frac{R}{40},
\]
and hence
\[
    -tR+20t\log(2+t)\leq-\frac12Re^{R/80}.
\]
With
$R_\theta=320\log\big(320\max\{e,\sqrt\theta\}\big)$, one has
$R_\theta>160$ and
\[
    \frac{\dd}{\dd R}\frac{e^{R/80}}{R^2}
    =
    \frac{e^{R/80}}{R^3}\Big(\frac{R}{80}-2\Big)\geq0
    \qquad (R\geq R_\theta).
\]
Moreover,
\[
    e^{R_\theta/80}
    =
    \big(320\max\{e,\sqrt\theta\}\big)^4
    \geq
    4\theta R_\theta^2 .
\]
Thus $e^{R/80}\geq4\theta R^2$ for $R\geq R_\theta$, and consequently
\[
    \sup_{\lambda\in[0,1]}\mu_\lambda\Big(
    -\frac{\cM_{\lambda,K}^{\tau_rJ_\omega}}{\norm{J_\omega}_{L^\infty}}
    \geq R
    \Big)
    \leq
    e^{-2\theta R^3},
    \qquad R\geq R_\theta .
\]
Therefore,
\begin{align*}
\sup_{\lambda\in[0,1]}\mathbb E_{\mu_\lambda}
\exp\left[
\theta
\frac{(\cM_{\lambda,K}^{\tau_rJ_\omega})_-^3}
{\norm{J_\omega}_{L^\infty}^3}
\right]                                                                     
\leq
e^{\theta R_\theta^3}
+
\theta\int_{R_\theta^3}^\infty e^{\theta s}e^{-2\theta s}\dd s
\leq
2e^{\theta R_\theta^3},
\end{align*}
which proves \eqref{eq:one-channel-2d}.
\end{proof}

\medskip

Recall, for $\delta>0$,
\[
    \Theta_{d,\delta}
    =
    \frac{1+\delta}{6}(2\pi)^d
    \int_\Omega\norm{J_\omega}_{L^\infty}^3\dd\nu(\omega).
\]
The parameter $\Theta_{d,\delta}$ collects the total channel strength
arising from Jensen's inequality.  Combining the one-channel estimates
then gives the following uniform bound.

\medskip

\begin{proposition}[Uniform exponential estimates]\label{prop:full-exp}
Under the assumptions of Theorem~\ref{thm:quantum-classical-main},
\begin{equation}\label{eq:full-exp-2d}
    \sup_K
    \sup_{\lambda\in[0,1]}\mathbb E_{\mu_\lambda}\left[
    e^{-(1+\delta)\cD_{\lambda,K}}\right]
    \leq
    \begin{cases}
    2\exp\left(\Theta_{2,\delta}R_{\Theta_{2,\delta}}^3\right), & d=2, \\
    (1-43200\,\Theta_{3,\delta})^{-1}, & d=3.
    \end{cases}
\end{equation}
Here $R_\theta$ is defined in \eqref{eq:Rtheta-def}.
\end{proposition}

\medskip

\begin{proof}
If $\int_\Omega\norm{J_\omega}_{L^\infty}^3\dd\nu(\omega)=0$, then
$\cM_{\lambda,K}^{\tau_rJ_\omega}=0$ for $\dd\nu(\omega)\dd r$-a.e.
$(\omega,r)$, and hence $\cD_{\lambda,K}=0$.  The estimates follow.  Assume the
integral is positive.  Since $-x^3\leq x_-^3$ for $x\in\R$,
\begin{align*}
-(1+\delta)\cD_{\lambda,K}
&=
-\frac{1+\delta}{6}
\int_\Omega\dd\nu(\omega)\int_{\bT^d}
\left(\cM_{\lambda,K}^{\tau_rJ_\omega}\right)^3\dd r                                  
\leq
\frac{1+\delta}{6}
\int_\Omega\dd\nu(\omega)\int_{\bT^d}
\left(\cM_{\lambda,K}^{\tau_rJ_\omega}\right)_-^3\dd r                                  \\
&=
\Theta_{d,\delta}
\frac{
\int_\Omega\dd\nu(\omega)\int_{\bT^d}
\norm{J_\omega}_{L^\infty}^3
\frac{\left(\cM_{\lambda,K}^{\tau_rJ_\omega}\right)_-^3}
{\norm{J_\omega}_{L^\infty}^3}
\dd r
}{
(2\pi)^d\int_\Omega\norm{J_\omega}_{L^\infty}^3\dd\nu(\omega)
},
\end{align*}
where the quotient is defined as zero on
$\{\norm{J_\omega}_{L^\infty}=0\}$.  Jensen's inequality applied to the
probability measure proportional to
$\norm{J_\omega}_{L^\infty}^3\dd\nu(\omega)\dd r$ gives
\begin{align*}
\mathbb E_{\mu_\lambda}e^{-(1+\delta)\cD_{\lambda,K}}                                      
\leq
\frac{
\int_\Omega\dd\nu(\omega)\int_{\bT^d}\dd r\,
\norm{J_\omega}_{L^\infty}^3
\mathbb E_{\mu_\lambda}
\exp\Big[
\Theta_{d,\delta}
\frac{\left(\cM_{\lambda,K}^{\tau_rJ_\omega}\right)_-^3}
{\norm{J_\omega}_{L^\infty}^3}
\Big]
}{
(2\pi)^d\int_\Omega\norm{J_\omega}_{L^\infty}^3\dd\nu(\omega)
}.
\end{align*}
Therefore, Proposition~\ref{prop:one-channel-exp} gives \eqref{eq:full-exp-2d}.
\end{proof}

\medskip

\begin{proposition}[Classical cubic Gibbs measure]
Assume the same assumptions as in Theorem~\ref{thm:quantum-classical-main}.  Then
\begin{equation}\label{eq:exp-conv}
    \lim_{K\to\infty}
    \sup_{\lambda\in[0,1]}
    \norm{
    e^{-\cD_{\lambda,K}}
    -
    e^{-\cD_{\lambda}}
    }_{L^1(\mu_\lambda)}
    =
    0 .
\end{equation}
Moreover,
\begin{equation}
    z_{\lambda}
    :=
    \int e^{-\cD_{\lambda}[u]}\dd\mu_\lambda(u),
    \qquad
    z
    :=
    \int e^{-\cD_{0}[u]}\dd\mu_0(u)
\end{equation}
satisfy $0<z_{\lambda},z<\infty$.  The probability measure
\begin{equation}
    \dd\mu(u)
    =
    z^{-1}e^{-\cD_{0}[u]}\dd\mu_0(u)
\end{equation}
is well-defined.
\end{proposition}

\medskip

\begin{proof}
We first prove \eqref{eq:exp-conv}, which follows from a uniform-integrability argument. Proposition~\ref{prop:full-exp} and Fatou's lemma
applied to a subsequence give
\[
    \sup_{\lambda\in[0,1]}
    \mathbb E_{\mu_\lambda}
    \left[
    e^{-(1+\delta)\cD_\lambda}
    \right]
    <\infty .
\]

For $R>0$ and $\varepsilon>0$,
\begin{align*}
&
\mu_\lambda\left(
\abs{e^{-\cD_{\lambda,K}}-e^{-\cD_\lambda}}>\varepsilon
\right)
\\
&\leq
\mu_\lambda\left(
\abs{\cD_{\lambda,K}-\cD_\lambda}>\frac{\varepsilon}{2R}
\right)
+
\mu_\lambda\left(e^{-\cD_{\lambda,K}}>R\right)
+
\mu_\lambda\left(e^{-\cD_\lambda}>R\right).
\end{align*}
Indeed, on the complement of the three events on the right,
\[
    \abs{e^{-\cD_{\lambda,K}}-e^{-\cD_\lambda}}
    \leq
    \abs{\cD_{\lambda,K}-\cD_\lambda}
    \left(e^{-\cD_{\lambda,K}}+e^{-\cD_\lambda}\right)
    \leq
    \varepsilon .
\]
By Markov's inequality,
\[
    \mu_\lambda\left(e^{-\cD_{\lambda,K}}>R\right)
    \leq
    R^{-(1+\delta)}
    \mathbb E_{\mu_\lambda}
    \left[
    e^{-(1+\delta)\cD_{\lambda,K}}
    \right],
\]
and
\[
    \mu_\lambda\left(e^{-\cD_\lambda}>R\right)
    \leq
    R^{-(1+\delta)}
    \mathbb E_{\mu_\lambda}
    \left[
    e^{-(1+\delta)\cD_\lambda}
    \right].
\]
Therefore
\[
    e^{-\cD_{\lambda,K}}
    -
    e^{-\cD_\lambda}
    \to0
\]
in $\mu_\lambda$-probability, uniformly for $\lambda\in[0,1]$.

For $R>0$,
\begin{align*}
&
\norm{
e^{-\cD_{\lambda,K}}
-
e^{-\cD_\lambda}
}_{L^1(\mu_\lambda)}
\\
&\leq
\mathbb E_{\mu_\lambda}
\left[
\min\left\{
\abs{e^{-\cD_{\lambda,K}}-e^{-\cD_\lambda}},
R
\right\}
\right]
\\
&\quad+
\mathbb E_{\mu_\lambda}
\left[
e^{-\cD_{\lambda,K}}
\1_{\{e^{-\cD_{\lambda,K}}>R\}}
\right]
+
\mathbb E_{\mu_\lambda}
\left[
e^{-\cD_\lambda}
\1_{\{e^{-\cD_\lambda}>R\}}
\right].
\end{align*}
For every $\varepsilon>0$,
\begin{align*}
\mathbb E_{\mu_\lambda}
\left[
\min\left\{
\abs{e^{-\cD_{\lambda,K}}-e^{-\cD_\lambda}},
R
\right\}
\right]
\leq
\varepsilon
+
R\,
\mu_\lambda\left(
\abs{e^{-\cD_{\lambda,K}}-e^{-\cD_\lambda}}>\varepsilon
\right).
\end{align*}
Moreover,
\[
    \mathbb E_{\mu_\lambda}
    \left[
    e^{-\cD_{\lambda,K}}
    \1_{\{e^{-\cD_{\lambda,K}}>R\}}
    \right]
    \leq
    R^{-\delta}
    \mathbb E_{\mu_\lambda}
    \left[
    e^{-(1+\delta)\cD_{\lambda,K}}
    \right],
\]
and
\[
    \mathbb E_{\mu_\lambda}
    \left[
    e^{-\cD_\lambda}
    \1_{\{e^{-\cD_\lambda}>R\}}
    \right]
    \leq
    R^{-\delta}
    \mathbb E_{\mu_\lambda}
    \left[
    e^{-(1+\delta)\cD_\lambda}
    \right].
\]
Taking first $K\to\infty$, then $\varepsilon\downarrow0$, and then
$R\to\infty$, gives \eqref{eq:exp-conv}.

The estimate above implies $e^{-\cD_\lambda}\in L^1(\mu_\lambda)$ for every
$\lambda\in[0,1]$.  Since $e^{-\cD_\lambda}>0$ $\mu_\lambda$-a.e.,
\[
    0<z_\lambda<\infty
    \qquad(\lambda\in[0,1]).
\]
In particular,
\[
    0<z<\infty .
\]
Thus $\mu$ is a probability measure.
\end{proof}

\bigskip

\section{Convergence of the classical partition functions}\label{sec:convergence-classical-partition}

We next compare the $\lambda$-dependent Gaussian objects with their
limiting counterparts.  To avoid changing probability spaces as
$\lambda\to0$, all Gaussian fields are realized on a common Fourier
coordinate space.

Let
$
    \gX=\CC^{\Z^d},
$
and let $\mathbb P$ be the product probability measure under which the
coordinate functions $g_p$, $p\in\Z^d$, are independent standard complex
Gaussian variables:
\[
    \dd\mathbb P(g)
    =
    \bigotimes_{p\in\Z^d}
    \pi^{-1}e^{-\abs{g_p}^2}\dd\Re g_p\,\dd\Im g_p .
\]
Put $\mathsf{C}_{0}=h^{-1}$.  For $0\leq\lambda\leq1$, write
\[
    \mathsf{C}_{\lambda} e_p=c_\lambda(p)e_p,
    \qquad
    c_\lambda(p)=
    \begin{cases}
    \displaystyle\frac{\lambda}{e^{\lambda h(p)}-1},&\lambda>0,\\[2mm]
    h(p)^{-1},&\lambda=0,
    \end{cases}
\]
and
\[
    u_\lambda(g)
    =
    \sum_{p\in\Z^d}c_\lambda(p)^{1/2}g_pe_p.
\]

For a bounded self-adjoint operator $A$ and an integer $K$, define
\[
    \widehat{\cM}_{\lambda,K}^A(g)
    =
    \cM_{\lambda,K}^A \left( u_{\lambda}(g) \right) \qquad (0\leq \lambda\leq1).
\]
Similarly, define
\[
    \widehat{\cD}_\lambda
    =
    \frac16
    \int_\Omega\dd\nu(\omega)
    \int_{\bT^d}
    \left(
    \widehat{\cM}_\lambda^{\tau_rJ_\omega}
    \right)^3\dd r .
\]

\medskip

\begin{lemma}
\label{lem:common-Gaussian-coordinate}
For every $\sigma>d/2-1$, the series
\[
    u_\lambda(g)
    =
    \sum_{p\in\Z^d}c_\lambda(p)^{1/2}g_pe_p
\]
converges in $L^2(\gX,\mathbb P;H^{-\sigma})$, and
\[
    (u_\lambda)_\#\mathbb P=\mu_\lambda
    \qquad(0\leq\lambda\leq1).
\]

For every $1\leq p<\infty$, $\widehat{\cM}_{\lambda,K}^A$ converges uniformly for $\lambda\in[0,1]$, as
$K\to\infty$ in $L^p(\gX,\mathbb P)$.  Its limit is denoted by
$\widehat{\cM}_\lambda^A$, and
\[
    \widehat{\cM}_\lambda^A=\cM_\lambda^A\circ u_\lambda
    \qquad(\lambda>0),
    \qquad
    \widehat{\cM}_0^A=\cM_0^A\circ u_0
\]
in $L^p(\gX,\mathbb P)$.
Similarly,
\[
    \widehat{\cD}_\lambda=\cD_\lambda\circ u_\lambda
    \qquad(\lambda>0),
    \qquad
    \widehat{\cD}_0=\cD_0\circ u_0
\]
in $L^1(\gX,\mathbb P)$, and
\begin{equation}
    \widehat{\cD}_\lambda\to\widehat{\cD}_0
    \qquad\text{in }L^1(\gX,\mathbb P).
\end{equation}
Moreover,
\[
    z_\lambda
    =
    \mathbb E e^{-\widehat{\cD}_\lambda}
    \qquad(\lambda>0),
    \qquad
    z
    =
    \mathbb E e^{-\widehat{\cD}_0}.
\]
\end{lemma}

\medskip

We postpone the proof of Lemma~\ref{lem:common-Gaussian-coordinate} to
Appendix~\ref{app:Gaussian}.

\medskip

\begin{proposition}[Convergence of the classical partition functions]\label{prop:tilted-density-common}
As $\lambda \downarrow 0$, we have $z_\lambda \to z$. Moreover, for every $s\in[1,1+\delta)$,
\begin{equation}
    z_\lambda^{-1}e^{-\widehat{\cD}_\lambda}\to z^{-1}e^{-\widehat{\cD}_0}
    \qquad\text{in }L^s(\gX,\mathbb P).
\end{equation}
Therefore, for some $\lambda_*>0$,
\begin{equation}\label{eq:tilted-density-common-uniform}
    \sup_{0<\lambda\leq\lambda_*}
    \norm{z_\lambda^{-1}e^{-\widehat{\cD}_\lambda}}_{L^{1+\delta}(\gX)}
    <\infty .
\end{equation}
\end{proposition}

\medskip

\begin{proof}
By Lemma~\ref{lem:common-Gaussian-coordinate},
\[
    \widehat{\cD}_\lambda\to\widehat{\cD}_0
    \qquad\text{in }L^1(\gX).
\]
Hence $\widehat{\cD}_\lambda\to\widehat{\cD}_0$ in probability, and
$e^{-\widehat{\cD}_\lambda}\to e^{-\widehat{\cD}_0}$ in probability.
Proposition~\ref{prop:full-exp} and Fatou's lemma
applied to a subsequence of $\widehat{\cD}_{\lambda,K}\to\widehat{\cD}_\lambda$ a.e.
give
\[
    \sup_{0\leq\lambda\leq1}
    \mathbb E e^{-(1+\delta)\widehat{\cD}_\lambda}
    <\infty .
\]
Therefore, Vitali's theorem gives
\[
    e^{-\widehat{\cD}_\lambda}\to e^{-\widehat{\cD}_0}
    \qquad\text{in }L^s(\gX)
\]
for every $s\in[1,1+\delta)$.  Hence
\[
    z_\lambda
    =
    \mathbb E e^{-\widehat{\cD}_\lambda}
    \to
    \mathbb E e^{-\widehat{\cD}_0}
    =
    z.
\]
Since $z>0$, there is $\lambda_*>0$ such that
\[
    z_\lambda\geq z/2
    \qquad(0<\lambda\leq\lambda_*).
\]
Thus
\[
\norm{z_\lambda^{-1}e^{-\widehat{\cD}_\lambda}-z^{-1}e^{-\widehat{\cD}_0}}_{L^s}
\leq
\abs{z_\lambda^{-1}-z^{-1}}
\norm{e^{-\widehat{\cD}_\lambda}}_{L^s}
+
z^{-1}
\norm{e^{-\widehat{\cD}_\lambda}-e^{-\widehat{\cD}_0}}_{L^s}
\to0.
\]
The bound \eqref{eq:tilted-density-common-uniform} follows from
\[
    \norm{z_\lambda^{-1}e^{-\widehat{\cD}_\lambda}}_{L^{1+\delta}}
    =
    z_\lambda^{-1}
    \norm{e^{-\widehat{\cD}_\lambda}}_{L^{1+\delta}},
\]
the lower bound $z_\lambda\geq z/2$ for $0<\lambda\leq\lambda_*$, and
Proposition~\ref{prop:full-exp}.
\end{proof}

\bigskip

\section{Quantum model: properties of the Gibbs states}\label{sec:quantum-model}

\begin{proposition}[Self-adjointness and the Gibbs state]\label{prop:quantum-wellposedness}
For $\lambda>0$, the operator $\bH_\lambda^{\ren}$ is self-adjoint and bounded from below.
Moreover,
\[
    Z_\lambda
    :=
    \tr_{\gF}(e^{-\bH_\lambda^{\ren}})
    <\infty,
    \qquad
    \Gamma_\lambda
    :=
    Z_\lambda^{-1}e^{-\bH_\lambda^{\ren}}
\]
defines a state on $\gF$.
\end{proposition}

\medskip

\begin{proof}
For $b,t\geq0$, the polynomial $F_b(t)=(t-b)^3+b^3$ satisfies
\begin{equation}
    F_b'(t)=3(t-b)^2,
    \qquad
    F_b(t)-\frac14t^3
    =
    \frac34t(t-2b)^2 .
\end{equation}
Thus $F_b$ is nondecreasing on $[0,\infty)$ and $F_b(t)\geq t^3/4$.
On $\gH^{\otimes_s n}$,
\[
    \lambda\dG(\tau_rJ_\omega)
    =
    \lambda\sum_{j=1}^n J_\omega(x_j-r)
\]
is a nonnegative multiplication operator.  Therefore
\begin{align}
&\bW_{\lambda,n}^{\ren}
+
\frac{(2\pi)^d}{6}
\int_\Omega m_{\lambda,\omega}^3\dd\nu(\omega)                              \nn\\
&=
\frac16
\int_\Omega\dd\nu(\omega)\int_{\bT^d}
F_{m_{\lambda,\omega}}
\Big(
\lambda\sum_{j=1}^nJ_\omega(x_j-r)
\Big)
\dd r
\geq0 .
\end{align}
In particular,
\[
    \bW_{\lambda,n}^{\ren}
    \geq
    -\frac{(2\pi)^d}{6}
    \int_\Omega m_{\lambda,\omega}^3\dd\nu(\omega).
\]
Since
\[
    0\leq m_{\lambda,\omega}
    =
    \lambda\tr((\tau_rJ_\omega)\gamma_0)
    \leq
    \lambda\norm{J_\omega}_{L^\infty}N_0 ,
\]
we have
\[
    \int_\Omega m_{\lambda,\omega}^3\dd\nu(\omega)
    \leq
    (\lambda N_0)^3
    \int_\Omega\norm{J_\omega}_{L^\infty}^3\dd\nu(\omega)<\infty
\]
for fixed $\lambda>0$.

The operator $\bW_{\lambda,n}^{\ren}$ is a bounded multiplication operator
on each $n$-particle sector.  Indeed,
\[
    0\leq
    \lambda\sum_{j=1}^nJ_\omega(x_j-r)
    \leq
    \lambda n\norm{J_\omega}_{L^\infty},
\]
and hence
\begin{align*}
\Big|
\Big[
\lambda\sum_{j=1}^nJ_\omega(x_j-r)-m_{\lambda,\omega}
\Big]^3
\Big|
&\leq
\left(
\lambda n\norm{J_\omega}_{L^\infty}
+
\lambda N_0\norm{J_\omega}_{L^\infty}
\right)^3                                                              \\
&=
\lambda^3(n+N_0)^3\norm{J_\omega}_{L^\infty}^3 .
\end{align*}
After integration in $r,\omega$, this gives a finite sector bound.  Thus
$\bH_{\lambda,n}^{\ren}=\lambda\sum_{j=1}^nh_j+\bW_{\lambda,n}^{\ren}$ is
self-adjoint on the domain of $\sum_{j=1}^nh_j$, and the direct sum over
$n\geq0$ is self-adjoint and bounded below by
\[
    \bH_\lambda^{\ren}
    \geq
    \lambda\dG(h)
    -
    \frac{(2\pi)^d}{6}
    \int_\Omega m_{\lambda,\omega}^3\dd\nu(\omega).
\]
Consequently,
\[
    Z_\lambda
    \leq
    \exp\left[
    \frac{(2\pi)^d}{6}
    \int_\Omega m_{\lambda,\omega}^3\dd\nu(\omega)
    \right]Z_0
    <\infty ,
\]
and $\Gamma_\lambda=Z_\lambda^{-1}e^{-\bH_\lambda^{\ren}}$ is a state.
\end{proof}

\medskip

\subsection{Free Gibbs states and coherent states}

A coherent state is a Weyl-rotation of the vacuum $\ket{0}$ in the Fock space $\gF$
\[
  \xi(u):=W(u)\ket{0}:=\exp (a^{\dagger}(u)-a(u))\ket{0}=e^{-\norm{u}^2/2}\bigoplus_{n\geq 0}\frac{1}{\sqrt{n!}}u^{\otimes n},
\]
for $u\in\gH$. 

The Weyl operator $W(u)$ is a unitary operator translating creation and annihilation operators
\[
  W(f)^* a^\dagger (g) W(f)=a^\dagger (g)+\inprod{f}{g}, \qquad W(f)^*a(g)W(f)=a(g)+\inprod{g}{f}.
\]
The $k$-particle density matrix of the coherent state $\xi(u)$ is 
\begin{equation}\label{eq:k-density-coherent-state}
    \ket{\xi(u)}\bra{\xi(u)}^{(k)}=\frac{1}{k!}\ket{u^{\otimes k}}\bra{u^{\otimes k}}.
\end{equation}

\medskip

\begin{lemma}[Coherent representation of the free state]\label{lem:free-coherent-representation}
For every $\lambda>0$,
\begin{equation}\label{eq:free-coherent-representation}
    \Gamma_0
    =
    \int
    \left|\xi(u/ \sqrt{\lambda})\right\rangle
    \left\langle\xi(u/ \sqrt{\lambda})\right|
    \dd\mu_\lambda(u).
\end{equation}
\end{lemma}

\medskip

\begin{proof}
Both sides of \eqref{eq:free-coherent-representation} are gauge-invariant
quasi-free states.  Hence it suffices to show that their one-particle density matrices coincide. We have from \eqref{eq:k-density-coherent-state} that 
\[
   \int
   \left|\xi(u/ \sqrt{\lambda})\right\rangle
   \left\langle\xi(u/ \sqrt{\lambda})\right|^{(1)}
   \dd\mu_\lambda(u)
   =\frac{1}{\lambda}\int \ket{u}\bra{u}\dd\mu_{\lambda}(u)
   =\frac{\mathsf{C}_{\lambda}}{\lambda}=\gamma_0.
\]
Thus, the right-hand side of \eqref{eq:free-coherent-representation} has one-particle
density matrix $\gamma_0$, which is precisely the one-particle density matrix of the free Gibbs state $\Gamma_0$.
\end{proof}

The remaining coherent-state estimates involve the expected particle
number $N_0$ of the free Gibbs state; we record its asymptotics next.

\medskip

\begin{lemma}[Asymptotic estimate for the number of particles]\label{lem:Clambda-trace-bound}
As $\lambda\downarrow0$,
\begin{equation}
    N_0:=\tr_{\gF}(\cN\Gamma_0)=\tr_{\gH}(\gamma_0)
    =
    \begin{cases}
    -\frac{\pi\log(\lambda)}{\lambda}+O(\lambda^{-1}),&d=2,\\[1mm]
    \frac{\pi^{3/2}\zeta(3/2)}{\lambda^{3/2}}+O(\lambda^{-1}),&d=3.
    \end{cases}
\end{equation}
\end{lemma}

\medskip

\begin{proof}
    Recall that
    \begin{align*}
    \tr_{\gH}(\gamma_0)=\sum_{p\in\Z^d}\frac1{e^{\lambda(|p|^2+1)}-1}&=\frac{1}{\lambda^{d/2}}\int_{\R^d}\frac{\dd k}{e^{|k|^2+\lambda}-1}+O(\lambda^{-1})\\
    &=
    \begin{cases}
    -\frac{\pi\log(\lambda)}{\lambda}+O(\lambda^{-1}),&d=2,\\[1mm]
    \frac{\pi^{3/2}\zeta(3/2)}{\lambda^{3/2}}+O(\lambda^{-1}),&d=3.
    \end{cases}
    \end{align*}
    See \cite[Lemma B.1]{LNR21} for the details of the asymptotic expansion.
\end{proof}

\subsection{Exponential estimates}

We now prove the quantum analogue of the classical exponential estimate in Proposition~\ref{prop:full-exp}. 

\medskip

\begin{proposition}[Quantum exponential estimate]\label{prop:quantum-exponential}
Under the assumptions of Theorem~\ref{thm:quantum-classical-main},
\begin{equation}\label{eq:quantum-full-exponential}
    \sup_{0<\lambda\leq1}
    \tr_{\gF}
    \left(
    e^{-(1+\delta)\bW_\lambda^{\ren}}\Gamma_0
    \right)
    \leq
    \begin{cases}
    2\exp\big(\Theta_{2,\delta}R_{\Theta_{2,\delta}}^3\big),&d=2,\\[1mm]
    (1-43200\Theta_{3,\delta})^{-1},&d=3.
    \end{cases}
\end{equation}
Here $R_\theta$ is defined in \eqref{eq:Rtheta-def}.
\end{proposition}

\medskip

\begin{proof}
Let $A$ be a bounded multiplication operator satisfying
\[
    0\leq A\leq a,
    \qquad
    a>0 .
\]
For every $t\geq0$, Lemma~\ref{lem:quasi-free-trace-identity}, applied to
$\lambda tA/a$, gives
\begin{align}
&
\tr_{\gF}
\Big[
\exp\Big(
-t\frac{\lambda(\dG(A)-\tr(A\gamma_0))}{a}
\Big)\Gamma_0
\Big]
\nn\\
&=
\exp\Big(
\frac{\lambda t}{a}\tr(A\gamma_0)
\Big)
\tr_{\gF}
\Big(
e^{-\dG(\lambda tA/a)}\Gamma_0
\Big)
\nn\\
&=
\exp\Big(
\frac{\lambda t}{a}\tr(A\gamma_0)
-
\tr_{\gH}
\log\big[
1+
\gamma_0^{1/2}
(1-e^{-\lambda tA/a})
\gamma_0^{1/2}
\big]
\Big).
\label{eq:quantum-one-channel-log-start}
\end{align}
Since $0\leq A/a\leq1$ and $x\mapsto1-e^{-\lambda tx}$ is concave on
$[0,1]$,
\[
    1-e^{-\lambda tA/a}
    \geq
    (1-e^{-\lambda t})\frac Aa .
\]
Hence
\[
    \gamma_0^{1/2}
    (1-e^{-\lambda tA/a})
    \gamma_0^{1/2}
    \geq
    (1-e^{-\lambda t})
    \gamma_0^{1/2}\frac Aa\gamma_0^{1/2}.
\]
Since $X\mapsto\log(1+X)$ is operator monotone on nonnegative operators,
\begin{align}
\tr_{\gH}
\log\left[
1+
\gamma_0^{1/2}
(1-e^{-\lambda tA/a})
\gamma_0^{1/2}
\right]
\geq
\tr_{\gH}
\log\left[
1+
(1-e^{-\lambda t})
\gamma_0^{1/2}\frac Aa\gamma_0^{1/2}
\right].
\label{eq:operator-log-monotone-step}
\end{align}
Combining \eqref{eq:quantum-one-channel-log-start} and
\eqref{eq:operator-log-monotone-step}, we obtain
\begin{align}
&
\log
\tr_{\gF}
\Big[
\exp\Big(
-t\frac{\lambda(\dG(A)-\tr(A\gamma_0))}{a}
\Big)\Gamma_0
\Big]
\nn\\
&\leq
\lambda t\,\tr_{\gH}
\Big(
\gamma_0^{1/2}\frac Aa\gamma_0^{1/2}
\Big)
-
\tr_{\gH}
\log\Big[
1+
(1-e^{-\lambda t})
\gamma_0^{1/2}\frac Aa\gamma_0^{1/2}
\Big].
\label{eq:one-channel-log-after-monotone}
\end{align}
Let $(\beta_j)_j$ be the eigenvalues of
$\gamma_0^{1/2}(A/a)\gamma_0^{1/2}$, counted with multiplicity and arranged
in decreasing order.  The quadratic form inequality
\[
    0
    \leq
    \gamma_0^{1/2}\frac Aa\gamma_0^{1/2}
    \leq
    \gamma_0
\]
and the min--max principle give
\[
    0\leq \beta_j\leq n_j,
\]
where $(n_j)_j$ is the decreasing rearrangement of
\[
    n_p=(e^{\lambda h(p)}-1)^{-1},
    \qquad p\in\Z^d .
\]
For $x\geq0$,
\[
    \frac{\dd}{\dd x}
    \left[
    \lambda tx-\log\bigl(1+(1-e^{-\lambda t})x\bigr)
    \right]
    =
    \lambda t
    -
    \frac{1-e^{-\lambda t}}
    {1+(1-e^{-\lambda t})x}
    \geq
    \lambda t-(1-e^{-\lambda t})
    \geq0.
\]
Therefore \eqref{eq:one-channel-log-after-monotone} implies
\begin{align}
\log
\tr_{\gF}
\Big[
\exp\Big(
-t\frac{\lambda(\dG(A)-\tr(A\gamma_0))}{a}
\Big)\Gamma_0
\Big]
\leq
\sum_{p\in\Z^d}
\big[
\lambda t\,n_p
-
\log\bigl(1+(1-e^{-\lambda t})n_p\bigr)
\big].
\label{eq:one-channel-log-sum}
\end{align}
For $s\geq0$ and $n\geq0$,
\[
    sn-\log(1+(1-e^{-s})n)
    \leq
    sn.
\]
Moreover,
\begin{align*}
sn-\log(1+(1-e^{-s})n)
&=
\int_0^s\int_0^\sigma
\left[
\frac{e^{-\tau}n}
{1+(1-e^{-\tau})n}
+
\frac{e^{-2\tau}n^2}
{(1+(1-e^{-\tau})n)^2}
\right]
\dd\tau\dd\sigma
\nn\\
&\leq
\int_0^s\int_0^\sigma n(1+n)\dd\tau\dd\sigma
=
\frac{s^2}{2}n(1+n).
\end{align*}
Taking $s=\lambda t$ and $n=n_p$, we get
\[
    \lambda t\,n_p
    \leq
    t\,\frac{\lambda}{e^{\lambda h(p)}-1}
    \leq
    \frac{t}{h(p)}
\]
and
\begin{align*}
\frac{(\lambda t)^2}{2}n_p(1+n_p)
=
\frac{(\lambda t)^2}{2}
\frac{e^{\lambda h(p)}}{(e^{\lambda h(p)}-1)^2}
=
\frac{t^2}{2h(p)^2}
\left(
\frac{\lambda h(p)}
{2\sinh(\lambda h(p)/2)}
\right)^2
\leq
\frac{t^2}{2h(p)^2}.
\end{align*}
Thus \eqref{eq:one-channel-log-sum} yields
\begin{align}
\tr_{\gF}
\Big[
\exp\Big(
-t\frac{\lambda(\dG(A)-\tr(A\gamma_0))}{a}
\Big)
\Gamma_0
\Big]
\leq
\exp\Big(
\sum_{p\in\Z^d}
\min\left\{
\frac{t}{h(p)},
\frac{t^2}{2h(p)^2}
\right\}
\Big).
\label{eq:quantum-one-channel-laplace}
\end{align}

For $\nu$-a.e. $\omega$ with
$\norm{J_\omega}_{L^\infty}>0$, take
\[
    A=\tau_rJ_\omega,
    \qquad
    a=\norm{J_\omega}_{L^\infty}.
\]
By Lemma~\ref{lem:lattice-sums} and
\eqref{eq:quantum-one-channel-laplace},
\begin{align}
\tr_{\gF}
\Big[
\exp\Big(
-t
\frac{
\lambda\dG(\tau_rJ_\omega)-m_{\lambda,\omega}
}
{\norm{J_\omega}_{L^\infty}}
\Big)
\Gamma_0
\Big]
\leq
\begin{cases}
\exp\big(20t\log(2+t)\big),&d=2,\\[1mm]
\exp\big(80t^{3/2}\big),&d=3.
\end{cases}
\label{eq:quantum-one-channel-laplace-dimensional}
\end{align}
If $d=3$, then for every $R\geq0$,
\begin{align*}
\tr_{\gF}
\left[
\1_{
\left\{
-
\frac{
\lambda\dG(\tau_rJ_\omega)-m_{\lambda,\omega}
}
{\norm{J_\omega}_{L^\infty}}
\geq R
\right\}
}
\Gamma_0
\right]
\leq
\inf_{t\geq0}\exp(-tR+80t^{3/2})
\leq
\exp(-R^3/43200),
\end{align*}
where the last inequality follows by choosing $t=(R/120)^2$.  Hence, for
$0\leq\theta<1/43200$,
\begin{align}
&
\tr_{\gF}
\Big[
\exp\Big(
\theta
\frac{
\left(
\lambda\dG(\tau_rJ_\omega)-m_{\lambda,\omega}
\right)_-^3
}
{\norm{J_\omega}_{L^\infty}^3}
\Big)
\Gamma_0
\Big]
\nn\\
&=
1+
\theta
\int_0^\infty
e^{\theta s}
\tr_{\gF}
\left[
\1_{
\left\{
\frac{
\left(
\lambda\dG(\tau_rJ_\omega)-m_{\lambda,\omega}
\right)_-^3
}
{\norm{J_\omega}_{L^\infty}^3}
\geq s
\right\}
}
\Gamma_0
\right]
\dd s
\nn\\
&\leq
1+
\theta
\int_0^\infty
e^{-(1/43200-\theta)s}\dd s
=
\frac1{1-43200\theta}.
\label{eq:quantum-one-channel-3d}
\end{align}
If $d=2$, then \eqref{eq:quantum-one-channel-laplace-dimensional} gives
\[
\tr_{\gF}
\left[
\1_{
\left\{
-
\frac{
\lambda\dG(\tau_rJ_\omega)-m_{\lambda,\omega}
}
{\norm{J_\omega}_{L^\infty}}
\geq R
\right\}
}
\Gamma_0
\right]
\leq
\inf_{t\geq0}\exp(-tR+20t\log(2+t)).
\]
For $R\geq80\log3$, taking $t=e^{R/80}$ gives
\[
    \log(2+t)\leq\log(3t)\leq\frac{R}{40},
\]
and hence
\[
    \inf_{t\geq0}\exp(-tR+20t\log(2+t))
    \leq
    \exp\Big(-\frac12Re^{R/80}\Big).
\]
For $R\geq R_\theta$,
\[
    e^{R/80}\geq4\theta R^2,
\]
and therefore
\[
\tr_{\gF}
\left[
\1_{
\left\{
-
\frac{
\lambda\dG(\tau_rJ_\omega)-m_{\lambda,\omega}
}
{\norm{J_\omega}_{L^\infty}}
\geq R
\right\}
}
\Gamma_0
\right]
\leq
e^{-2\theta R^3}.
\]
It follows that
\begin{align}
\tr_{\gF}
\Big[
\exp\Big(
\theta
\frac{
\left(
\lambda\dG(\tau_rJ_\omega)-m_{\lambda,\omega}
\right)_-^3
}
{\norm{J_\omega}_{L^\infty}^3}
\Big)
\Gamma_0
\Big]
\leq
e^{\theta R_\theta^3}
+
\theta
\int_{R_\theta^3}^{\infty}e^{\theta s}e^{-2\theta s}\dd s
\leq
2e^{\theta R_\theta^3}.
\label{eq:quantum-one-channel-2d}
\end{align}

If
$
    \int_\Omega\norm{J_\omega}_{L^\infty}^3\dd\nu(\omega)=0,
$
then $\bW_\lambda^{\ren}=0$ and \eqref{eq:quantum-full-exponential} follows.
Assume now that this integral is positive.  Since $-x^3\leq x_-^3$,
\begin{align*}
-(1+\delta)\bW_\lambda^{\ren}
&\leq
\frac{1+\delta}{6}
\int_\Omega\dd\nu(\omega)\int_{\bT^d}
\left(
\lambda\dG(\tau_rJ_\omega)-m_{\lambda,\omega}
\right)_-^3
\dd r
\nn\\
&=
\Theta_{d,\delta}
\frac{
\int_\Omega\dd\nu(\omega)\int_{\bT^d}
\norm{J_\omega}_{L^\infty}^3
\frac{
\left(
\lambda\dG(\tau_rJ_\omega)-m_{\lambda,\omega}
\right)_-^3
}
{\norm{J_\omega}_{L^\infty}^3}
\dd r
}{
(2\pi)^d\int_\Omega\norm{J_\omega}_{L^\infty}^3\dd\nu(\omega)
},
\end{align*}
with the quotient set equal to zero on
$\{\norm{J_\omega}_{L^\infty}=0\}$.  Since all
$\dG(\tau_rJ_\omega)$ commute, Jensen's inequality in the joint spectral
calculus gives
\begin{align}
&
\tr_{\gF}
\big(
e^{-(1+\delta)\bW_\lambda^{\ren}}\Gamma_0
\big)
\nn\\
&\leq
\frac{
\int_\Omega\dd\nu(\omega)\int_{\bT^d}
\norm{J_\omega}_{L^\infty}^3
\tr_{\gF}
\Big[
\exp\Big(
\Theta_{d,\delta}
\frac{
\left(
\lambda\dG(\tau_rJ_\omega)-m_{\lambda,\omega}
\right)_-^3
}
{\norm{J_\omega}_{L^\infty}^3}
\Big)
\Gamma_0
\Big]
\dd r
}{
(2\pi)^d\int_\Omega\norm{J_\omega}_{L^\infty}^3\dd\nu(\omega)
}.
\label{eq:quantum-jensen-channel}
\end{align}
If $d=2$, then \eqref{eq:quantum-one-channel-2d} with
$\theta=\Theta_{2,\delta}$ and \eqref{eq:quantum-jensen-channel} give
\[
    \tr_{\gF}
    \big(
    e^{-(1+\delta)\bW_\lambda^{\ren}}\Gamma_0
    \big)
    \leq
    2\exp\big(\Theta_{2,\delta}R_{\Theta_{2,\delta}}^3\big).
\]
If $d=3$, then $\Theta_{3,\delta}<1/43200$, and
\eqref{eq:quantum-one-channel-3d} with
$\theta=\Theta_{3,\delta}$ and \eqref{eq:quantum-jensen-channel} give
\[
    \tr_{\gF}
    \big(
    e^{-(1+\delta)\bW_\lambda^{\ren}}\Gamma_0
    \big)
    \leq
    (1-43200\Theta_{3,\delta})^{-1}.
\]
This proves \eqref{eq:quantum-full-exponential}.
\end{proof}

\bigskip

\section{Quantum-to-Classical convergence of the relative free energy}\label{sec:Q-C-free-energy}

We now prove the convergence of the relative free energy. 

\subsection{Relative free-energy lower bound}

\medskip

\begin{lemma}[Coherent $q$-moment comparison]\label{lem:coherent-q-moment-comparison}
Let $q\geq1$.  Under Assumption~\ref{ass:J}, for every $0<\lambda\leq1$,
\begin{align}
&
\Big(\int
\left\langle
\xi(u/ \sqrt{\lambda}),
\left|
\bW_\lambda^{\ren}
-
\cD_\lambda[u]
\right|^q
\xi(u/ \sqrt{\lambda})
\right\rangle
\dd\mu_\lambda(u)\Big)^{1/q}
\nn\\
&\lesssim_q
\Big(
\frac{(2\pi)^d}{6}
\int_\Omega\norm{J_\omega}_{L^\infty}^3\dd\nu(\omega)
\Big)
\nn\\
&\times
\Big[
(\tr_{\gH}h^{-2})
\big((\lambda^2N_0)^{1/2}+\lambda\big)
+
(\tr_{\gH}h^{-2})^{1/2}
\big(\lambda^2N_0+\lambda^{2}\big)
+
(\lambda^2N_0)^{3/2}
+
\lambda^{3}
\Big].
\label{eq:coherent-q-moment-comparison}
\end{align}
Here the implicit constant in $\lesssim_q$ depends only on $q$.
\end{lemma}

\medskip

\begin{proof}
Let $A=A^*$ be bounded and satisfy
\[
    0\leq A\leq a,
    \qquad
    a>0.
\]
For every $\theta\in\R$,
\[
    e^{\theta\lambda\dG(A)}
    =
    \Gamma(e^{\theta\lambda A}),
\]
and the coherent-vector identity gives
\begin{align}
\left\langle
\xi(u/ \sqrt{\lambda}),
e^{\theta(\lambda\dG(A)-\inprod{u}{Au})}
\xi(u/ \sqrt{\lambda})
\right\rangle
=
\exp\Big[
\frac1\lambda
\inprod{u}{
\big(e^{\theta\lambda A}-1-\theta\lambda A\big)u}
\Big].
\label{eq:coherent-channel-laplace}
\end{align}
If $\abs{\theta}\lambda a\leq1$, then functional calculus and
\[
    0\leq e^x-1-x\leq x^2,
    \qquad
    -1\leq x\leq1,
\]
give
\[
    0
    \leq
    e^{\theta\lambda A}-1-\theta\lambda A
    \leq
    \theta^2\lambda^2A^2 .
\]
Thus \eqref{eq:coherent-channel-laplace} implies
\begin{equation}\label{eq:coherent-channel-laplace-bound}
\left\langle
\xi(u/ \sqrt{\lambda}),
e^{\theta(\lambda\dG(A)-\inprod{u}{Au})}
\xi(u/ \sqrt{\lambda})
\right\rangle
\leq
\exp\left(
\theta^2\lambda\inprod{u}{A^2u}
\right)
\end{equation}
whenever $\abs{\theta}\lambda a\leq1$.

For every $R\geq0$, \eqref{eq:coherent-channel-laplace-bound} and the
spectral theorem give
\begin{align*}
\left\langle
\xi(u/ \sqrt{\lambda}),
\1_{\{\lambda\dG(A)-\inprod{u}{Au}\geq R\}}
\xi(u/ \sqrt{\lambda})
\right\rangle
\leq
\inf_{0\leq\theta\leq(\lambda a)^{-1}}
\exp\left(
-\theta R+\theta^2\lambda\inprod{u}{A^2u}
\right),
\end{align*}
and
\begin{align*}
\left\langle
\xi(u/ \sqrt{\lambda}),
\1_{\{-\lambda\dG(A)+\inprod{u}{Au}\geq R\}}
\xi(u/ \sqrt{\lambda})
\right\rangle
\leq
\inf_{0\leq\theta\leq(\lambda a)^{-1}}
\exp\left(
-\theta R+\theta^2\lambda\inprod{u}{A^2u}
\right).
\end{align*}
Choosing
\[
    \theta=\frac{R}{2\lambda\inprod{u}{A^2u}}
    \quad\text{if}\quad
    R\leq \frac{2\inprod{u}{A^2u}}{a},
\]
and
\[
    \theta=\frac1{\lambda a}
    \quad\text{if}\quad
    R>\frac{2\inprod{u}{A^2u}}{ a},
\]
with the first choice omitted when $\inprod{u}{A^2u}=0$, yields
\begin{align}
\left\langle
\xi(u/ \sqrt{\lambda}),
\1_{\{\abs{\lambda\dG(A)-\inprod{u}{Au}}\geq R\}}
\xi(u/ \sqrt{\lambda})
\right\rangle
\leq
2\exp\Big(
-
\min\left\{
\frac{R^2}{4\lambda\inprod{u}{A^2u}},
\frac{R}{2\lambda a}
\right\}
\Big).
\end{align}
Hence, for
$\alpha\in\{q,2q,3q\}$,
\begin{align}
&
\left\langle
\xi(u/ \sqrt{\lambda}),
\abs{
\lambda\dG(A)-\inprod{u}{Au}
}^{\alpha}
\xi(u/ \sqrt{\lambda})
\right\rangle
\nn\\
&=
\alpha\int_0^\infty R^{\alpha-1}
\left\langle
\xi(u/ \sqrt{\lambda}),
\1_{\{\abs{\lambda\dG(A)-\inprod{u}{Au}}\geq R\}}
\xi(u/ \sqrt{\lambda})
\right\rangle
\dd R
\nn\\
&\lesssim_q
(\lambda\inprod{u}{A^2u})^{\alpha/2}
+
(\lambda a)^\alpha .
\label{eq:coherent-channel-moment}
\end{align}
Since
\[
    0\leq A\leq a,
    \qquad
    \inprod{u}{A^2u}\leq a^2\norm u^2,
\]
\eqref{eq:coherent-channel-moment} gives
\begin{align}
\left\langle
\xi(u/ \sqrt{\lambda}),
\left|
\frac{\lambda\dG(A)-\inprod{u}{Au}}{a}
\right|^\alpha
\xi(u/ \sqrt{\lambda})
\right\rangle
\lesssim_q
(\lambda\norm u^2)^{\alpha/2}
+
\lambda^\alpha .
\label{eq:normalized-coherent-channel-moment}
\end{align}

Since
\[
    \lambda\dG(A)-\tr_{\gH}(A \mathsf{C}_{\lambda})
    =
    \cM_\lambda^A[u]
    +
    \lambda\dG(A)-\inprod{u}{Au},
\]
we have
\begin{align}
&
\left[
\lambda\dG(A)-\tr_{\gH}(A \mathsf{C}_{\lambda})
\right]^3
-
\big(\cM_\lambda^A[u]\big)^3
\nn\\
&=
3\big(\cM_\lambda^A[u]\big)^2
\left[
\lambda\dG(A)-\inprod{u}{Au}
\right]
+
3\cM_\lambda^A[u]
\left[
\lambda\dG(A)-\inprod{u}{Au}
\right]^2
\nn\\
&\quad+
\left[
\lambda\dG(A)-\inprod{u}{Au}
\right]^3 .
\end{align}
Applying the scalar inequality
\[
    \abs{x+y+z}^q
    \lesssim_q
    \abs x^q+\abs y^q+\abs z^q
\]
in the spectral calculus of
$\lambda\dG(A)-\inprod{u}{Au}$, and using
\eqref{eq:normalized-coherent-channel-moment}, we get
\begin{align}
&
\left\langle
\xi(u/ \sqrt{\lambda}),
\left|
\frac{
\left[
\lambda\dG(A)-\tr_{\gH}(A \mathsf{C}_{\lambda})
\right]^3
-
\big(\cM_\lambda^A[u]\big)^3
}{a^3}
\right|^q
\xi(u/ \sqrt{\lambda})
\right\rangle
\nn\\
&\lesssim_q
\left|
\frac{\cM_\lambda^A[u]}{a}
\right|^{2q}
\left[
(\lambda\norm u^2)^{q/2}
+\lambda^q
\right]
+
\left|
\frac{\cM_\lambda^A[u]}{a}
\right|^q
\left[
(\lambda\norm u^2)^q
+\lambda^{2q}
\right]
+
(\lambda\norm u^2)^{3q/2}
+\lambda^{3q}.
\label{eq:single-channel-cubic-pointwise}
\end{align}

By Lemma~\ref{lem:Mass-Lp}, applied to $A/a$,
\[
    \sup_{0<\lambda\leq1}
    \left\|
    \frac{\cM_\lambda^A}{a}
    \right\|_{L^p(\mu_\lambda)}
    \lesssim_p
    (\tr_{\gH}h^{-2})^{1/2}
\]
for every finite $p\geq1$.  Since $\mu_\lambda$ has covariance
$\mathsf{C}_{\lambda}$, writing
\[
    u=\sum_{k\in\Z^d}(\lambda n_k)^{1/2}g_k e_k
\]
with independent standard complex Gaussian variables $g_p$, Minkowski's
inequality gives
\begin{equation}\label{eq:lambda-L2norm-moment}
    \norm{\norm u^2}_{L^p(\mu_\lambda)}
    \leq
    \sum_{k\in\Z^d}\lambda n_k\norm{\abs{g_k}^2}_{L^p}
    \lesssim_p
    \lambda N_0.
\end{equation}
Using Hölder's inequality in \eqref{eq:single-channel-cubic-pointwise}, with
the exponents $4q,2q,q$ for $\cM_\lambda^A/a$ and with
\eqref{eq:lambda-L2norm-moment}, gives
\begin{align}
&
\int
\left\langle
\xi(u/ \sqrt{\lambda}),
\left|
\frac{
\left[
\lambda\dG(A)-\tr_{\gH}(A \mathsf{C}_{\lambda})
\right]^3
-
\big(\cM_\lambda^A[u]\big)^3
}{a^3}
\right|^q
\xi(u/ \sqrt{\lambda})
\right\rangle
\dd\mu_\lambda(u)
\nn\\
&\lesssim_q
(\tr_{\gH}h^{-2})^q
\big((\lambda^2N_0)^{q/2}+\lambda^q\big)
+
(\tr_{\gH}h^{-2})^{q/2}
\big((\lambda^2N_0)^q+\lambda^{2q}\big)
+
(\lambda^2N_0)^{3q/2}
+
\lambda^{3q}.
\label{eq:single-channel-cubic-integrated}
\end{align}

Since the operators $\tau_rJ_\omega$ are multiplication operators, the
family $\{\dG(\tau_rJ_\omega)\}_{\omega,r}$ has a joint spectral calculus.
Jensen's inequality in that joint spectral calculus gives
\begin{align}
&
\left|
\bW_\lambda^{\ren}
-
\cD_\lambda[u]
\right|^q
\nn\\
&\leq
\left(
\frac{(2\pi)^d}{6}
\int_\Omega\norm{J_\omega}_{L^\infty}^3\dd\nu(\omega)
\right)^q
\nn\\
&\quad\times
\frac{
\displaystyle
\int_\Omega\int_{\bT^d}
\norm{J_\omega}_{L^\infty}^3
\left|
\frac{
\left[
\lambda\dG(\tau_rJ_\omega)-m_{\lambda,\omega}
\right]^3
-
\big(\cM_\lambda^{\tau_rJ_\omega}[u]\big)^3
}{
\norm{J_\omega}_{L^\infty}^3
}
\right|^q
\dd r\,\dd\nu(\omega)
}{
\displaystyle
(2\pi)^d
\int_\Omega\norm{J_\omega}_{L^\infty}^3\dd\nu(\omega)
}.
\end{align}
Taking the coherent-state expectation,
integrating in $u$, and applying \eqref{eq:single-channel-cubic-integrated} with $A=\tau_rJ_\omega$ and $a=\norm{J_\omega}_{L^\infty}$, we
prove \eqref{eq:coherent-q-moment-comparison}.
\end{proof}

We now combine this comparison with the exponential bounds from the previous
sections and the Golden--Thompson inequality.

\medskip

\begin{proposition}[Relative free-energy lower bound]\label{prop:free-energy-lower}
It holds that
\begin{equation}\label{eq:free-energy-lower}
    \liminf_{\lambda\downarrow0}
    \Big(
    -\log\frac{Z_\lambda}{Z_0}
    \Big)
    \geq
    -\log z .
\end{equation}
\end{proposition}

\medskip

\begin{proof}
Golden--Thompson's inequality gives
\begin{align}
0\leq\frac{Z_\lambda}{Z_0}
=
\frac{\tr_{\gF}\left(e^{-\lambda\dG(h)-\bW_\lambda^{\ren}}\right)}{Z_0}
\leq
\frac{
\tr_{\gF}\left(e^{-\lambda\dG(h)}e^{-\bW_\lambda^{\ren}}\right)
}{Z_0}
=
\tr_{\gF}\big(e^{-\bW_\lambda^{\ren}}\Gamma_0\big).
\label{eq:free-energy-lower-GT}
\end{align}
By Lemma~\ref{lem:free-coherent-representation},
\begin{align}
\tr_{\gF}\big(e^{-\bW_\lambda^{\ren}}\Gamma_0\big)
&=
\int
\left\langle
\xi(u/ \sqrt{\lambda}),
e^{-\bW_\lambda^{\ren}}
\xi(u/ \sqrt{\lambda})
\right\rangle
\dd\mu_\lambda(u).
\end{align}

Fix $u\in L^2(\bT^d)$.  Let $\dd\rho_{u,\lambda}$ be the spectral measure
of $\bW_\lambda^{\ren}$ in the vector
$\xi(u/\sqrt\lambda)$, namely
\[
    \int_{\R} F(x)\dd\rho_{u,\lambda}(x)
    =
    \left\langle
    \xi(u/ \sqrt{\lambda}),
    F(\bW_\lambda^{\ren})
    \xi(u/ \sqrt{\lambda})
    \right\rangle
\]
for every nonnegative Borel function $F$.  Since
\[
    \int_{\R}\dd\rho_{u,\lambda}(x)=1,
\]
the spectral theorem gives
\begin{align}
\left|
\left\langle
\xi(u/ \sqrt{\lambda}),
e^{-\bW_\lambda^{\ren}}
\xi(u/ \sqrt{\lambda})
\right\rangle
-
e^{-\cD_\lambda[u]}
\right|
&=
\left|
\int_{\R}
\left(e^{-x}-e^{-\cD_\lambda[u]}\right)
\dd\rho_{u,\lambda}(x)
\right|\nn\\
&
\leq
\int_{\R}
\left|
e^{-x}-e^{-\cD_\lambda[u]}
\right|
\dd\rho_{u,\lambda}(x).
\label{eq:spectral-difference-start}
\end{align}
For $x,y\in\R$,
\[
    \abs{e^{-x}-e^{-y}}
    =
    \left|
    \int_y^x e^{-s}\dd s
    \right|
    \leq
    \abs{x-y}\big(e^{-x}+e^{-y}\big).
\]
Taking $y=\cD_\lambda[u]$ in \eqref{eq:spectral-difference-start}, we get
\begin{align}
&
\left|
\left\langle
\xi(u/ \sqrt{\lambda}),
e^{-\bW_\lambda^{\ren}}
\xi(u/ \sqrt{\lambda})
\right\rangle
-
e^{-\cD_\lambda[u]}
\right|
\nn\\
&\leq
\int_{\R}
\abs{x-\cD_\lambda[u]}e^{-x}
\dd\rho_{u,\lambda}(x)
+
e^{-\cD_\lambda[u]}
\int_{\R}
\abs{x-\cD_\lambda[u]}
\dd\rho_{u,\lambda}(x).
\label{eq:spectral-difference-scalar-bound}
\end{align}
Hölder's inequality with respect to $\dd\rho_{u,\lambda}$ gives
\begin{align}
&
\int_{\R}
\abs{x-\cD_\lambda[u]}e^{-x}
\dd\rho_{u,\lambda}(x)
\nn\\
&\leq
\Big(
\int_{\R}e^{-(1+\delta)x}\dd\rho_{u,\lambda}(x)
\Big)^{\frac{1}{1+\delta}}
\Big(
\int_{\R}\abs{x-\cD_\lambda[u]}^q\dd\rho_{u,\lambda}(x)
\Big)^{\frac1q}
\nn\\
&=
\left\langle
\xi(u/ \sqrt{\lambda}),
e^{-(1+\delta)\bW_\lambda^{\ren}}
\xi(u/ \sqrt{\lambda})
\right\rangle^{\frac{1}{1+\delta}}
\left\langle
\xi(u/ \sqrt{\lambda}),
\left|
\bW_\lambda^{\ren}-\cD_\lambda[u]
\right|^q
\xi(u/ \sqrt{\lambda})
\right\rangle^{\frac{1}{q}},
\label{eq:spectral-holder-first}
\end{align}
where $q=(1+\delta)/\delta$. Also,
\begin{align}
\int_{\R}
\abs{x-\cD_\lambda[u]}
\dd\rho_{u,\lambda}(x)
&\leq
\Big(
\int_{\R}
\abs{x-\cD_\lambda[u]}^q
\dd\rho_{u,\lambda}(x)
\Big)^{\frac{1}{q}}
\nn\\
&=
\left\langle
\xi(u/ \sqrt{\lambda}),
\left|
\bW_\lambda^{\ren}-\cD_\lambda[u]
\right|^q
\xi(u/ \sqrt{\lambda})
\right\rangle^{1/q}.
\label{eq:spectral-holder-second}
\end{align}
Combining \eqref{eq:spectral-difference-scalar-bound},
\eqref{eq:spectral-holder-first}, and \eqref{eq:spectral-holder-second},
\begin{align}
&
\left|
\left\langle
\xi(u/ \sqrt{\lambda}),
e^{-\bW_\lambda^{\ren}}
\xi(u/ \sqrt{\lambda})
\right\rangle
-
e^{-\cD_\lambda[u]}
\right|
\nn\\
&\leq
\Big(\left\langle
\xi(u/ \sqrt{\lambda}),
e^{-(1+\delta)\bW_\lambda^{\ren}}
\xi(u/ \sqrt{\lambda})
\right\rangle^{\frac{1}{1+\delta}}+e^{-\cD_\lambda[u]}\Big)
\nn\\
&\times
\left\langle
\xi(u/ \sqrt{\lambda}),
\left|
\bW_\lambda^{\ren}-\cD_\lambda[u]
\right|^q
\xi(u/ \sqrt{\lambda})
\right\rangle^{1/q}.
\label{eq:free-energy-lower-pointwise-comparison}
\end{align}

Integrating \eqref{eq:free-energy-lower-pointwise-comparison} with respect to
$\mu_\lambda$, and applying Hölder's inequality in $u$, yields
\begin{align}
&
\left|
\tr_{\gF}\big(e^{-\bW_\lambda^{\ren}}\Gamma_0\big)
-
z_\lambda
\right|
\nn\\
&\leq
\Big[
\int
\left\langle
\xi(u/ \sqrt{\lambda}),
e^{-(1+\delta)\bW_\lambda^{\ren}}
\xi(u/ \sqrt{\lambda})
\right\rangle
\dd\mu_\lambda(u)
\Big]^{1/(1+\delta)}
\nn\\
&\quad\times
\Big[
\int
\left\langle
\xi(u/ \sqrt{\lambda}),
\left|
\bW_\lambda^{\ren}-\cD_\lambda[u]
\right|^q
\xi(u/ \sqrt{\lambda})
\right\rangle
\dd\mu_\lambda(u)
\Big]^{1/q}
\nn\\
&\quad+
\Big[
\int e^{-(1+\delta)\cD_\lambda[u]}\dd\mu_\lambda(u)
\Big]^{1/(1+\delta)}
\nn\\
&\quad\times
\Big[
\int
\left\langle
\xi(u/ \sqrt{\lambda}),
\left|
\bW_\lambda^{\ren}-\cD_\lambda[u]
\right|^q
\xi(u/ \sqrt{\lambda})
\right\rangle
\dd\mu_\lambda(u)
\Big]^{1/q}.
\label{eq:free-energy-lower-integrated-comparison}
\end{align}
By Lemma~\ref{lem:free-coherent-representation},
\[
\int
\left\langle
\xi(u/ \sqrt{\lambda}),
e^{-(1+\delta)\bW_\lambda^{\ren}}
\xi(u/ \sqrt{\lambda})
\right\rangle
\dd\mu_\lambda(u)
=
\tr_{\gF}
\left(
e^{-(1+\delta)\bW_\lambda^{\ren}}\Gamma_0
\right).
\]
Proposition~\ref{prop:quantum-exponential} and
Proposition~\ref{prop:full-exp} give
\[
    \sup_{0<\lambda\leq1}
    \tr_{\gF}
    \left(
    e^{-(1+\delta)\bW_\lambda^{\ren}}\Gamma_0
    \right)
    <\infty
\]
and
\[
    \sup_{0<\lambda\leq1}
    \int e^{-(1+\delta)\cD_\lambda[u]}\dd\mu_\lambda(u)
    <\infty .
\]
Lemma~\ref{lem:coherent-q-moment-comparison} and Lemma~\ref{lem:Clambda-trace-bound} give
\[
\int
\left\langle
\xi(u/ \sqrt{\lambda}),
\left|
\bW_\lambda^{\ren}-\cD_\lambda[u]
\right|^q
\xi(u/ \sqrt{\lambda})
\right\rangle
\dd\mu_\lambda(u)
\to0 .
\]
Therefore, \eqref{eq:free-energy-lower-integrated-comparison} implies
\[
    \tr_{\gF}
    \left(e^{-\bW_\lambda^{\ren}}\Gamma_0\right)
    -
    z_\lambda
    \to0 ,
\]
which, together with \eqref{eq:free-energy-lower-GT} and Proposition~\ref{prop:tilted-density-common}, implies \eqref{eq:free-energy-lower}.
\end{proof}

\subsection{Relative free-energy upper bound}

For the upper bound, we use the coherent-state trial state
$\widetilde\Gamma_\lambda$, for $\lambda>0$, defined as follows:
\[
    \widetilde\Gamma_\lambda
    =
    \frac{1}{z_\lambda}
    \int
    \left|\xi(u/ \sqrt{\lambda})\right\rangle
    \left\langle\xi(u/ \sqrt{\lambda})\right|
    e^{-\cD_\lambda[u]}\dd\mu_\lambda(u).
\]

\medskip

\begin{lemma}[Coherent-state cubic comparison]\label{lem:coherent-cubic-comparison}
For every $u\in L^2(\bT^d)$, $\lambda>0$, and every $\nu$-a.e. $\omega$, $r\in\bT^d$,
\begin{align}
&
\left\langle
\xi(u/ \sqrt{\lambda}),
\left[
\lambda\dG(\tau_rJ_\omega)-m_{\lambda,\omega}
\right]^3
\xi(u/ \sqrt{\lambda})
\right\rangle                                                                  \nn\\
&=
\big(\cM_\lambda^{\tau_rJ_\omega}[u]\big)^3
+
3\lambda\cM_\lambda^{\tau_rJ_\omega}[u]
\inprod{u}{(\tau_rJ_\omega)^2u}
+
\lambda^2\inprod{u}{(\tau_rJ_\omega)^3u}.
\label{eq:coherent-cubic-exact}
\end{align}
Consequently, we have, for every $\delta>0$,
\begin{align}
&
\left|
\tr_{\gF}(\bW_\lambda^{\ren}\widetilde\Gamma_\lambda)
-
\int\cD_\lambda[u]z_\lambda^{-1}e^{-\cD_\lambda[u]}\dd\mu_\lambda(u)
\right|                                                                        \nn\\
&\lesssim_{\delta}
\frac{(2\pi)^d}{z_\lambda}
\left(\int_\Omega\norm{J_\omega}_{L^\infty}^3\dd\nu(\omega)\right)            
\left[
(\tr h^{-2})^{1/2}\lambda^2 N_0
+
\lambda^3 N_0
\right]\norm{e^{-\cD_\lambda}}_{L^{1+\delta}(\mu_\lambda)}   .
\label{eq:coherent-cubic-remainder}
\end{align}
Under the assumptions of Theorem~\ref{thm:quantum-classical-main}, the right-hand side of \eqref{eq:coherent-cubic-remainder} tends to zero.
\end{lemma}

\medskip

\begin{proof}
For a bounded multiplication operator $A$, the coherent vector identities are
obtained from $a(f)\xi(u/\sqrt\lambda)=\lambda^{-1/2}\inprod{f}{u}\xi(u/\sqrt\lambda)$
and the canonical commutation relations.  They read
\begin{align*}
\left\langle\xi(u/ \sqrt{\lambda}),\dG(A)
\xi(u/ \sqrt{\lambda})\right\rangle
&=
\lambda^{-1}\inprod{u}{Au},                                                 \\
\left\langle\xi(u/ \sqrt{\lambda}),\dG(A)^2
\xi(u/ \sqrt{\lambda})\right\rangle
&=
\lambda^{-2}\inprod{u}{Au}^2
+
\lambda^{-1}\inprod{u}{A^2u},                                                \\
\left\langle\xi(u/ \sqrt{\lambda}),\dG(A)^3
\xi(u/ \sqrt{\lambda})\right\rangle
&=
\lambda^{-3}\inprod{u}{Au}^3
+
3\lambda^{-2}\inprod{u}{Au}\inprod{u}{A^2u}
+
\lambda^{-1}\inprod{u}{A^3u}.
\end{align*}
The second identity follows from
\[
    \dG(A)^2
    =
    \int A(x)A(y)a_x^\dagger a_y^\dagger a_ya_x\dd x\dd y
    +
    \dG(A^2),
\]
and the third follows from
\begin{align*}
\dG(A)^3
&=
\int A(x)A(y)A(z)
a_x^\dagger a_y^\dagger a_z^\dagger a_za_ya_x\dd x\dd y\dd z                 \\
&\quad+
3\int A(x)^2A(y)a_x^\dagger a_y^\dagger a_ya_x\dd x\dd y
+
\dG(A^3).
\end{align*}
Substituting $A=\tau_rJ_\omega$ and expanding
$[\lambda\dG(A)-m_{\lambda,\omega}]^3$, with
$m_{\lambda,\omega}=\tr_{\gH}(A\mathsf{C}_{\lambda})$, gives
\begin{align*}
&
\left\langle\xi(u/ \sqrt{\lambda}),
[\lambda\dG(A)-m_{\lambda,\omega}]^3
\xi(u/ \sqrt{\lambda})\right\rangle                         \\
&=
\left(\inprod{u}{Au}-m_{\lambda,\omega}\right)^3
+
3\lambda
\left(\inprod{u}{Au}-m_{\lambda,\omega}\right)
\inprod{u}{A^2u}
+
\lambda^2\inprod{u}{A^3u},
\end{align*}
which is \eqref{eq:coherent-cubic-exact}.

For the remainder estimate, $0\leq\tau_rJ_\omega\leq\norm{J_\omega}_{L^\infty}$
implies
\[
    \inprod{u}{(\tau_rJ_\omega)^2u}
    \leq
    \norm{J_\omega}_{L^\infty}^2\norm u^2,
    \qquad
    \inprod{u}{(\tau_rJ_\omega)^3u}
    \leq
    \norm{J_\omega}_{L^\infty}^3\norm u^2 .
\]
Let $q=(1+\delta)/\delta$.  Hölder's inequality gives
\begin{align*}
&\int
\abs{\cM_\lambda^{\tau_rJ_\omega}[u]}\norm u^2
e^{-\cD_\lambda[u]}\dd\mu_\lambda(u)                                      \\
&\leq
\norm{e^{-\cD_\lambda}}_{L^{1+\delta}(\mu_\lambda)}
\norm{\cM_\lambda^{\tau_rJ_\omega}}_{L^{2q}(\mu_\lambda)}
\Big(
\int\norm u^{4q}\dd\mu_\lambda(u)
\Big)^{1/(2q)}.
\end{align*}
By Lemma~\ref{lem:Mass-Lp},
\[
    \norm{\cM_\lambda^{\tau_rJ_\omega}}_{L^{2q}(\mu_\lambda)}
    \lesssim_{\delta}
    \norm{J_\omega}_{L^\infty}(\tr h^{-2})^{1/2},
\]
Since
$\norm u^2$ is the sum of independent complex Gaussian squares, an estimate similar to \eqref{eq:lambda-L2norm-moment} gives
\[
    \Big(
    \int\norm u^{4q}\dd\mu_\lambda(u)
    \Big)^{1/(2q)}
    \lesssim_{\delta}
    \lambda N_0 .
\]
Hence,
\begin{align*}
\int
\abs{\cM_\lambda^{\tau_rJ_\omega}[u]}\norm u^2
e^{-\cD_\lambda[u]}\dd\mu_\lambda(u)                                      \lesssim_{\delta}
\norm{e^{-\cD_\lambda}}_{L^{1+\delta}(\mu_\lambda)}
\norm{J_\omega}_{L^\infty}
(\tr h^{-2})^{1/2}
\lambda N_0 .
\end{align*}
Similarly,
\[
    \int\norm u^2e^{-\cD_\lambda[u]}\dd\mu_\lambda(u)
    \lesssim_{\delta}
    \norm{e^{-\cD_\lambda}}_{L^{1+\delta}(\mu_\lambda)}
    \lambda N_0 .
\]
Integrating in $r,\omega$, gives
\eqref{eq:coherent-cubic-remainder}.  Proposition~\ref{prop:full-exp}
bounds the $L^{1+\delta}$-norm, Lemma~\ref{lem:Clambda-trace-bound} gives
$\lambda^2N_0\to0$, and $z_\lambda\to z>0$; the right-hand side of
\eqref{eq:coherent-cubic-remainder} therefore tends to zero.
\end{proof}

\medskip

\begin{proposition}[Relative free-energy upper bound]\label{prop:coherent-upper-bound}
We have
\begin{equation}\label{eq:coherent-upper-bound}
    \limsup_{\lambda\downarrow0}
    \left(
    -\log\frac{Z_\lambda}{Z_0}
    \right)
    \leq
    -\log z .
\end{equation}
\end{proposition}

\medskip

\begin{proof}
Set
\[
    f_\lambda(u)
    :=
    z_\lambda^{-1}e^{-\cD_\lambda[u]}.
\]
The exponential estimate in Proposition~\ref{prop:full-exp} and the
elementary bound, valid for some $C_\delta<\infty$,
\[
    |x|e^{-x}
    \leq
    C_\delta\big(1+e^{-(1+\delta)x}\big),
    \qquad x\in\R,
\]
imply
\[
    \int |\cD_\lambda[u]|f_\lambda(u)\dd\mu_\lambda(u)<\infty,
    \qquad
    \int f_\lambda|\log f_\lambda|\dd\mu_\lambda<\infty.
\]
Lemma~\ref{lem:coherent-cubic-comparison}, together with the lower bound
on $\bW_\lambda^{\ren}$ from
Proposition~\ref{prop:quantum-wellposedness}, also shows that
$\bW_\lambda^{\ren}$ is integrable in the state
$\widetilde\Gamma_\lambda$.

Define the coherent-state preparation map
\[
    \Phi_\lambda(f)
    :=
    \int
    \left|\xi(u/\sqrt\lambda)\right\rangle
    \left\langle\xi(u/\sqrt\lambda)\right|
    f(u)\dd\mu_\lambda(u),
    \qquad f\in L^1(\mu_\lambda).
\]
It is completely positive and trace preserving, since
\[
    \tr_{\gF}(\Phi_\lambda(f))
    =
    \int f\,\dd\mu_\lambda.
\]
Furthermore,
\[
    \Phi_\lambda(1)=\Gamma_0,
    \qquad
    \Phi_\lambda(f_\lambda)=\widetilde\Gamma_\lambda
\]
by Lemma~\ref{lem:free-coherent-representation}.  Lemma~\ref{lem:relative-entropy-Berezin-Lieb} gives
\begin{align*}
\cH(\widetilde\Gamma_\lambda,\Gamma_0)
\leq
\cH_{\cl}(f_\lambda\mu_\lambda,\mu_\lambda)
=
\int f_\lambda\log f_\lambda\,\dd\mu_\lambda
=
-\log z_\lambda
-
\int\cD_\lambda[u]f_\lambda(u)\dd\mu_\lambda(u).
\end{align*}
Adding $\tr_{\gF}(\bW_\lambda^{\ren}\widetilde\Gamma_\lambda)$ and using
Lemma~\ref{lem:coherent-cubic-comparison} yields
\begin{align*}
\cH(\widetilde\Gamma_\lambda,\Gamma_0)
+
\tr_{\gF}(\bW_\lambda^{\ren}\widetilde\Gamma_\lambda)
&\leq
-\log z_\lambda
+
\Big[
\tr_{\gF}(\bW_\lambda^{\ren}\widetilde\Gamma_\lambda)
-
\int\cD_\lambda[u]f_\lambda(u)\dd\mu_\lambda(u)
\Big]
\\
&\leq
-\log z_\lambda+o(1)
=
-\log z+o(1).
\end{align*}
The Gibbs variational principle therefore gives
\begin{align*}
-\log\frac{Z_\lambda}{Z_0}
\leq
\cH(\widetilde\Gamma_\lambda,\Gamma_0)
+
\tr_{\gF}(\bW_\lambda^{\ren}\widetilde\Gamma_\lambda)
\leq
-\log z+o(1),
\end{align*}
which proves \eqref{eq:coherent-upper-bound}.
\end{proof}

The matching lower and upper bounds determine the free energy.  The same
identity also shows that the true Gibbs state is close in trace norm to the
coherent trial state.

\medskip

\begin{corollary}[Free-energy and state convergence]\label{cor:state-trace-convergence}
The free-energy convergence \eqref{eq:free-energy-main-limit} holds and
\begin{equation}\label{eq:state-trace-convergence}
    \norm{\Gamma_\lambda-\widetilde\Gamma_\lambda}_{\gS^1}\to0 .
\end{equation}
\end{corollary}

\medskip

\begin{proof}
The convergence of the free energy follows from
Propositions~\ref{prop:free-energy-lower} and
\ref{prop:coherent-upper-bound}.  For the trace-norm convergence, use
\[
    \log\Gamma_\lambda
    =
    -\bH_\lambda^{\ren}-\log Z_\lambda
    =
    -\lambda\dG(h)-\bW_\lambda^{\ren}-\log Z_\lambda
\]
and
\[
    \log\Gamma_0
    =
    -\lambda\dG(h)-\log Z_0 .
\]
Therefore
\begin{align*}
\cH(\widetilde\Gamma_\lambda,\Gamma_\lambda)
&=
\tr_{\gF}
\left[
\widetilde\Gamma_\lambda
(\log\widetilde\Gamma_\lambda-\log\Gamma_\lambda)
\right]                                                                      \\
&=
\tr_{\gF}
\left[
\widetilde\Gamma_\lambda
(\log\widetilde\Gamma_\lambda-\log\Gamma_0)
\right]
+
\tr_{\gF}(\bW_\lambda^{\ren}\widetilde\Gamma_\lambda)
+
\log\frac{Z_\lambda}{Z_0}                                                    \\
&=
\cH(\widetilde\Gamma_\lambda,\Gamma_0)
+
\tr_{\gF}(\bW_\lambda^{\ren}\widetilde\Gamma_\lambda)
+
\log\frac{Z_\lambda}{Z_0}.
\end{align*}
The upper-bound proof gives
\[
    \cH(\widetilde\Gamma_\lambda,\Gamma_0)
    +
    \tr_{\gF}(\bW_\lambda^{\ren}\widetilde\Gamma_\lambda)
    \leq
    -\log z+o(1),
\]
and the free-energy convergence gives
\[
    \log\frac{Z_\lambda}{Z_0}
    \to
    \log z .
\]
Since relative entropy is nonnegative, we have
\[
    0\leq
    \cH(\widetilde\Gamma_\lambda,\Gamma_\lambda)\leq -\log z+\log\frac{Z_\lambda}{Z_0}+o(1) \to0.
\]
Pinsker's inequality gives
\[
    \norm{\Gamma_\lambda-\widetilde\Gamma_\lambda}_{\gS^1}^2
    \leq
    2\cH(\widetilde\Gamma_\lambda,\Gamma_\lambda),
\]
and \eqref{eq:state-trace-convergence} follows.
\end{proof}

\bigskip

\section{Quantum-to-Classical convergence of reduced density matrices}\label{sec:Q-C-density}

In this section, we prove the Hilbert--Schmidt convergence of the fixed-order reduced density matrices. We will use the notations introduced in Section~\ref{sec:convergence-classical-partition}. 
Before we start, we briefly sketch the strategy of the proof. 

We know from last section that the Gibbs state $\Gamma_\lambda$ is close to the coherent trial state $\tilde{\Gamma}_\lambda$. 
It remains to upgrade this trace-norm convergence of the states to convergence of their fixed-order reduced density matrices.
Lemma~11.4 of \cite{LNR21} gives
\begin{align*}
\norm{\Gamma_\lambda^{(k)}-\widetilde\Gamma_\lambda^{(k)}}_{\gS^2}^2
\lesssim_k
\norm{\Gamma_\lambda-\widetilde\Gamma_\lambda}_{\gS^1}
\sum_{\ell=k}^{2k}
\left(
\norm{\Gamma_\lambda^{(\ell)}}_{\gS^2}
+
\norm{\widetilde\Gamma_\lambda^{(\ell)}}_{\gS^2}
\right).
\end{align*}
Thus, besides the trace-norm convergence from
Corollary~\ref{cor:state-trace-convergence}, we need uniform
Hilbert--Schmidt bounds for the fixed-order reduced density matrices. Through the
Ginibre loop representation proved in
Appendix~\ref{app:Ginibre-loop-representation}, this monotonicity yields
the analogous pointwise kernel domination
\begin{equation}
    0\leq
    \Gamma_\lambda^{(\ell)}(\underline X_\ell;\underline Y_\ell)
    \leq
    \Gamma_0^{(\ell)}(\underline X_\ell;\underline Y_\ell),
    \qquad
    \ell\geq1.
\end{equation}
See Proposition~\ref{prop:Ginibre-domination} in Appendix~\ref{app:Ginibre-loop-representation}. 

\medskip

\begin{lemma}[Gaussian cross moments]\label{lem:Gaussian-cross-moments}
Let $T\in\gS^2(\gH)$.  Let $g=(g_j)$ and $g'=(g'_j)$ be two independent
copies of the coordinate family on $\gX$.  For finite-rank $T$, define
\[
    Q_T(g,g')
    =
    \sum_{i,j}
    \overline{g_i}\,\inprod{e_i}{T e_j}\,g'_j .
\]
For general $T\in\gS^2$, $Q_T$ is the $L^{2r}(\gX\times\gX)$-limit of
$Q_{T_n}$ for any finite-rank sequence $T_n\to T$ in $\gS^2$.  For every
$r\geq1$,
\begin{equation}\label{eq:Gaussian-cross-moment}
    \int_{\gX\times\gX}|Q_T(g,g')|^{2r}
    \dd \mathbb{P}^{\otimes 2}(g,g')
    \leq
    \Gamma(r+1)^2\norm{T}_{\gS^2}^{2r}.
\end{equation}
Consequently,
\begin{equation}\label{eq:Gaussian-cross-moment-difference}
    \norm{Q_T-Q_S}_{L^{2r}(\gX\times\gX)}
    \leq
    \Gamma(r+1)^{1/r}\norm{T-S}_{\gS^2}.
\end{equation}
\end{lemma}

\medskip

\begin{proof}
Assume first that $T$ has finite rank.  Let
\[
    T=\sum_j s_j |f_j\rangle\langle e_j|
\]
be a singular value decomposition over the nonzero singular values.  By
unitary invariance of the standard complex Gaussian family,
\[
    Q_T
    \stackrel{\mathrm{law}}=
    \sum_j s_j\overline{\eta_j}\eta'_j,
\]
where $(\eta_j)$ and $(\eta'_j)$ are independent standard complex Gaussian
families.  Conditioning on $(\eta_j)$, the random variable
$
    \sum_j s_j\overline{\eta_j}\eta'_j
$
is a centered complex Gaussian variable with variance
$
    \sum_j s_j^2\abs{\eta_j}^2.
$
Therefore
\[
    \int_{\gX}
    \big|
    \sum_j s_j\overline{\eta_j}\eta'_j
    \big|^{2r}
    \dd \mathbb{P}(\eta')
    =
    \Gamma(r+1)
    \Big(\sum_j s_j^2\abs{\eta_j}^2\Big)^r.
\]
If $\norm T_{\gS^2}>0$, set
\[
    a_j=\frac{s_j^2}{\norm T_{\gS^2}^2},
    \qquad
    \sum_j a_j=1.
\]
Convexity of $x\mapsto x^r$ gives
\[
    \Big(\sum_j a_j\abs{\eta_j}^2\Big)^r
    \leq
    \sum_j a_j\abs{\eta_j}^{2r}.
\]
Thus
\[
    \int_{\gX}
    \Big(\sum_j s_j^2\abs{\eta_j}^2\Big)^r
    \dd \mathbb{P}(\eta)
    \leq
    \Gamma(r+1)\norm T_{\gS^2}^{2r}.
\]
This proves \eqref{eq:Gaussian-cross-moment} for finite-rank $T$.  The case
$T=0$ is immediate.

For $T\in\gS^2$, choose finite-rank $T_n\to T$ in $\gS^2$.  Applying
the finite-rank estimate to $T_n-T_m$ gives
\[
    \norm{Q_{T_n}-Q_{T_m}}_{L^{2r}}
    \leq
    \Gamma(r+1)^{1/r}\norm{T_n-T_m}_{\gS^2}.
\]
Hence $Q_{T_n}$ converges in $L^{2r}$, and the limit is independent of the
approximating sequence.  Passing to the limit gives
\eqref{eq:Gaussian-cross-moment}.  Applying the same estimate to $T-S$
gives \eqref{eq:Gaussian-cross-moment-difference}.
\end{proof}

For $0\leq\lambda,\lambda'\leq1$ and an orthogonal projection $P$, define
\begin{equation}
    Q_{\lambda,\lambda'}^P
    :=
    Q_{\mathsf{C}_{\lambda}^{1/2}P \mathsf{C}_{\lambda'}^{1/2}} .
\end{equation}
We write $Q_{\lambda,\lambda'}=Q_{\lambda,\lambda'}^1$.  Since
$\mathsf{C}_{\lambda}\leq h^{-1}$,
\begin{equation}\label{eq:Q-S2-bound}
    \norm{\mathsf{C}_{\lambda}^{1/2}P \mathsf{C}_{\lambda'}^{1/2}}_{\gS^2}
    \leq
    \norm{h^{-1/2}P h^{-1/2}}_{\gS^2}.
\end{equation}
For $P=1$,
\[
    \norm{\mathsf{C}_{\lambda}^{1/2} \mathsf{C}_{\lambda'}^{1/2}}_{\gS^2}^2
    \leq
    \tr h^{-2}.
\]
For $P=1-P_K$,
\begin{equation}\label{eq:Q-tail-S2-bound}
    \norm{\mathsf{C}_{\lambda}^{1/2}(1-P_K)\mathsf{C}_{\lambda'}^{1/2}}_{\gS^2}^2
    \leq
    \tr((1-P_K)h^{-2}).
\end{equation}
The tail estimate for $Q_{\lambda,\lambda'}^{1-P_K}$ allows us first to
prove convergence for fixed Fourier cutoff and then remove the cutoff
uniformly.

\medskip

\begin{proposition}\label{prop:classical-correlation-operators}
For every fixed $k\geq1$,
\begin{equation}
    \gamma_{\lambda}^{(k)}
    :=
    \int
    |u^{\otimes k}\rangle\langle u^{\otimes k}|
    z_\lambda^{-1}e^{-\cD_\lambda[u]}\dd\mu_\lambda(u)
    \to
    \gamma_\mu^{(k)}
    \qquad\text{in }\gS^2.
\end{equation}
\end{proposition}

\medskip

\begin{proof}
For $K<\infty$ and $0\leq\lambda\leq1$, set
\begin{equation}
    \gamma_{\lambda,K}^{(k)}
    =
    \int_{\gX}
    z_\lambda^{-1}e^{-\widehat{\cD}_\lambda}(g)
    |(P_Ku_\lambda(g))^{\otimes k}\rangle
    \langle(P_Ku_\lambda(g))^{\otimes k}|
    \dd \mathbb P(g).
\end{equation}

For $0\leq\lambda,\lambda'\leq1$,
\begin{align}
    \tr\left[
    \gamma_{\lambda,K}^{(k)}
    \gamma_{\lambda',K}^{(k)}
    \right]
    &=
    \int_{\gX\times\gX}
    z_\lambda^{-1}e^{-\widehat{\cD}_\lambda}(g) z_{\lambda'}^{-1}e^{-\widehat{\cD}_{\lambda'}}(g')
    \left|Q_{\lambda,\lambda'}^{P_K}(g,g')\right|^{2k}
    \dd \mathbb{P}^{\otimes 2}(g,g').
\end{align}
The identity follows from
\[
    \tr\left[
    |a^{\otimes k}\rangle\langle a^{\otimes k}|
    |b^{\otimes k}\rangle\langle b^{\otimes k}|
    \right]
    =
    \abs{\inprod{a}{b}}^{2k},
    \qquad
    a,b\in P_K\gH .
\]

Let $s\in(1,1+\delta)$ and set $q=s/(s-1)$.  By
Proposition~\ref{prop:tilted-density-common},
\[
    z_\lambda^{-1}e^{-\widehat{\cD}_\lambda}\to z^{-1}e^{-\widehat{\cD}_0}
    \qquad\text{in }L^s(\gX).
\]
Thus,
\[
    z_\lambda^{-1}e^{-\widehat{\cD}_\lambda}(g) z_{\lambda'}^{-1}e^{-\widehat{\cD}_{\lambda'}}(g')
    \to
    z^{-1}e^{-\widehat{\cD}_0}(g)z^{-1}e^{-\widehat{\cD}_0}(g')
    \qquad\text{in }L^s(\gX\times\gX)
\]
as $\lambda,\lambda'\downarrow0$.  For fixed $K$,
\[
    \mathsf{C}_{\lambda}^{1/2}P_K\mathsf{C}_{\lambda'}^{1/2}
    \to
    h^{-1/2}P_Kh^{-1/2}
    \qquad\text{in }\gS^2 .
\]
Lemma~\ref{lem:Gaussian-cross-moments} gives
\[
    Q_{\lambda,\lambda'}^{P_K}
    \to
    Q_{0,0}^{P_K}
    \qquad\text{in }L^{2kq}(\gX\times\gX),
\]
and hence
\[
    \abs{Q_{\lambda,\lambda'}^{P_K}}^{2k}
    \to
    \abs{Q_{0,0}^{P_K}}^{2k}
    \qquad\text{in }L^q(\gX\times\gX).
\]
Hölder's inequality in $\gX\times\gX$ yields
\begin{align}
    \tr\left[
    \gamma_{\lambda,K}^{(k)}
    \gamma_{\lambda',K}^{(k)}
    \right]
    \to
    \tr\left[
    \gamma_{0,K}^{(k)}
    \gamma_{0,K}^{(k)}
    \right]
\label{eq:finite-K-correlation-inner-conv}
\end{align}
as $\lambda,\lambda'\downarrow0$.  Taking $\lambda'=\lambda$ and
$\lambda'=0$ in \eqref{eq:finite-K-correlation-inner-conv} gives
\[
    \gamma_{\lambda,K}^{(k)}
    \to
    \gamma_{0,K}^{(k)}
    \qquad\text{in }\gS^2
\]
for every fixed $K$.

We now pass to $K\to\infty$.  For $0<\lambda\leq1$,
\[
    \gamma_{\lambda}^{(k)}
    =
    \int_{\gX}
    z_\lambda^{-1}e^{-\widehat{\cD}_\lambda}(g)
    |u_\lambda(g)^{\otimes k}\rangle
    \langle u_\lambda(g)^{\otimes k}|
    \dd \mathbb P(g),
\]
because $\mu_\lambda(L^2(\bT^d))=1$.  The Hilbert--Schmidt inner product formula
with $P=1$ and $P=P_K$ gives
\begin{align}
&
\norm{\gamma_{\lambda}^{(k)}-\gamma_{\lambda,K}^{(k)}}_{\gS^2}^2
\nn\\
&=
\int_{\gX\times\gX}
z_\lambda^{-1}e^{-\widehat{\cD}_\lambda}(g)z_\lambda^{-1}e^{-\widehat{\cD}_\lambda}(g')
\Big(
\abs{Q_{\lambda,\lambda}(g,g')}^{2k}
-
\abs{Q_{\lambda,\lambda}^{P_K}(g,g')}^{2k}
\Big)
\dd \mathbb{P}^{\otimes 2}(g,g').
\end{align}
Consequently,
\begin{align}
\norm{\gamma_{\lambda}^{(k)}-\gamma_{\lambda,K}^{(k)}}_{\gS^2}^2
\leq
\int_{\gX\times\gX}
z_\lambda^{-1}e^{-\widehat{\cD}_\lambda}(g)z_\lambda^{-1}e^{-\widehat{\cD}_\lambda}(g')
\Big|
\abs{Q_{\lambda,\lambda}}^{2k}
-
\abs{Q_{\lambda,\lambda}^{P_K}}^{2k}
\Big|
\dd \mathbb{P}^{\otimes 2}(g,g').
\end{align}
Since
$
    Q_{\lambda,\lambda}
    -
    Q_{\lambda,\lambda}^{P_K}
    =
    Q_{\lambda,\lambda}^{1-P_K}
$
and
\[
    \abs{\abs a^{2k}-\abs b^{2k}}
    \lesssim_k
    \abs{a-b}\big(\abs a^{2k-1}+\abs b^{2k-1}\big),
\]
Hölder's inequality, \eqref{eq:tilted-density-common-uniform}, and
Lemma~\ref{lem:Gaussian-cross-moments} imply
\begin{align}
\sup_{0<\lambda\leq\lambda_*}
\norm{\gamma_{\lambda}^{(k)}-\gamma_{\lambda,K}^{(k)}}_{\gS^2}^2
\lesssim_{k,\delta}
\left(\tr((1-P_K)h^{-2})\right)^{1/2}
\left(\tr h^{-2}\right)^{k-1/2}.
\end{align}
Indeed, \eqref{eq:Q-tail-S2-bound} controls the
$L^{2q}(\gX\times\gX)$-norm of $Q_{\lambda,\lambda}^{1-P_K}$, while
\eqref{eq:Q-S2-bound} controls the remaining $Q$-factors.  Therefore
\[
    \lim_{K\to\infty}
    \sup_{0<\lambda\leq\lambda_*}
    \norm{\gamma_{\lambda}^{(k)}-\gamma_{\lambda,K}^{(k)}}_{\gS^2}
    =
    0.
\]

The same argument with $\lambda=0$ and $P_L-P_K$ in place of $1-P_K$
gives
\[
    \norm{\gamma_{0,L}^{(k)}-\gamma_{0,K}^{(k)}}_{\gS^2}\to0
    \qquad(K,L\to\infty).
\]
Recall we define in \eqref{eq:classical-correlation-formal} that
\[
    \gamma_\mu^{(k)}
    :=
    \lim_{K\to\infty}\gamma_{0,K}^{(k)}
    \qquad\text{in }\gS^2.
\]

For every $K$,
\begin{align*}
\norm{\gamma_{\lambda}^{(k)}-\gamma_\mu^{(k)}}_{\gS^2}
\leq
\norm{\gamma_{\lambda}^{(k)}-\gamma_{\lambda,K}^{(k)}}_{\gS^2}
+
\norm{\gamma_{\lambda,K}^{(k)}-\gamma_{0,K}^{(k)}}_{\gS^2}
+
\norm{\gamma_{0,K}^{(k)}-\gamma_\mu^{(k)}}_{\gS^2}.
\end{align*}
First let $\lambda\downarrow0$ with $K$ fixed, and then let
$K\to\infty$.  This proves
\[
    \gamma_{\lambda}^{(k)}
    \to
    \gamma_\mu^{(k)}
    \qquad\text{in }\gS^2 .
\]
\end{proof}
The preceding proposition identifies the limiting correlations of the
coherent trial states.  We now pass from the trial states to the
Gibbs states using the trace-norm convergence and the uniform
Hilbert--Schmidt bounds.

\medskip

\begin{proposition}[Convergence of the fixed-order reduced density matrices]
For every fixed $k\geq1$, \eqref{eq:density-main-limit} holds.
\end{proposition}

\medskip

\begin{proof}
For $u\in L^2(\bT^d)$ and $\lambda>0$,
\[
    \Big(
    \left|\xi(u/ \sqrt{\lambda})\right\rangle
    \left\langle\xi(u/ \sqrt{\lambda})\right|
    \Big)^{(k)}
    =
    \frac1{k!\lambda^k}
    |u^{\otimes k}\rangle\langle u^{\otimes k}|;
\]
see \eqref{eq:k-density-coherent-state}.  Since $\mu_\lambda(L^2(\bT^d))=1$ for $\lambda>0$, integration gives
\[
    k!\lambda^k\widetilde\Gamma_\lambda^{(k)}
    =
    \gamma_{\lambda}^{(k)} .
\]
Proposition~\ref{prop:classical-correlation-operators} gives
\[
    k!\lambda^k\widetilde\Gamma_\lambda^{(k)}
    \to
    \gamma_\mu^{(k)}
    \qquad\text{in }\gS^2 .
\]

It remains to compare $\Gamma_\lambda$ and $\widetilde\Gamma_\lambda$.
Lemma~11.4 of \cite{LNR21} gives,
\begin{align}
\norm{\Gamma_\lambda^{(k)}-\widetilde\Gamma_\lambda^{(k)}}_{\gS^2}^2
\lesssim_k
\norm{\Gamma_\lambda-\widetilde\Gamma_\lambda}_{\gS^1}
\sum_{\ell=k}^{2k}
\big(
\norm{\Gamma_\lambda^{(\ell)}}_{\gS^2}
+
\norm{\widetilde\Gamma_\lambda^{(\ell)}}_{\gS^2}
\big).
\label{eq:LNR-density-comparison-used}
\end{align}
By \eqref{eq:interacting-HS-bound},
\[
    \sup_{0<\lambda\leq1}
    \lambda^\ell\norm{\Gamma_\lambda^{(\ell)}}_{\gS^2}<\infty .
\]
For $\widetilde\Gamma_\lambda$,
\[
    \lambda^\ell\norm{\widetilde\Gamma_\lambda^{(\ell)}}_{\gS^2}
    =
    \frac1{\ell!}
    \norm{\gamma_{\lambda}^{(\ell)}}_{\gS^2}.
\]
With $q=(1+\delta)/\delta$,
\[
    \norm{\gamma_{\lambda}^{(\ell)}}_{\gS^2}^2
    =
    \mathbb E_{g,g'}
    \left[
    z_\lambda^{-1}e^{-\widehat{\cD}_\lambda}(g)z_\lambda^{-1}e^{-\widehat{\cD}_\lambda}(g')
    \abs{Q_{\lambda,\lambda}(g,g')}^{2\ell}
    \right].
\]
Hölder's inequality and Lemma~\ref{lem:Gaussian-cross-moments}, applied to
$T=\mathsf{C}_{\lambda}$, give
\begin{align*}
\norm{\gamma_{\lambda}^{(\ell)}}_{\gS^2}
&\leq
z_\lambda^{-1}
\norm{e^{-\widehat{\cD}_\lambda}}_{L^{1+\delta}(\gX)}
\Gamma(\ell q+1)^{1/q}
\left(\tr \mathsf{C}_{\lambda}^2\right)^{\ell/2}.
\end{align*}
Since
$
    \tr \mathsf{C}_{\lambda}^2\leq\tr h^{-2},
$
$z_\lambda\to z>0$, and Proposition~\ref{prop:full-exp} holds,
\[
    \sup_{0<\lambda\leq\lambda_*}
    \lambda^\ell\norm{\widetilde\Gamma_\lambda^{(\ell)}}_{\gS^2}<\infty .
\]
Multiplying \eqref{eq:LNR-density-comparison-used} by $\lambda^{2k}$ gives
\[
\begin{aligned}
\lambda^{2k}
\norm{\Gamma_\lambda^{(k)}-\widetilde\Gamma_\lambda^{(k)}}_{\gS^2}^2
\lesssim_k
\norm{\Gamma_\lambda-\widetilde\Gamma_\lambda}_{\gS^1}
\sum_{\ell=k}^{2k}
\lambda^{2k-\ell}
\big(
\lambda^\ell\norm{\Gamma_\lambda^{(\ell)}}_{\gS^2}
+
\lambda^\ell\norm{\widetilde\Gamma_\lambda^{(\ell)}}_{\gS^2}
\big).
\end{aligned}
\]
The sum is uniformly bounded for $0<\lambda\leq\lambda_*$, and
\eqref{eq:state-trace-convergence} gives
\[
    \lambda^k
    \norm{\Gamma_\lambda^{(k)}-\widetilde\Gamma_\lambda^{(k)}}_{\gS^2}
    \to0 .
\]
Therefore
\[
\begin{aligned}
\norm{k!\lambda^k\Gamma_\lambda^{(k)}-\gamma_\mu^{(k)}}_{\gS^2}
\leq
k!\lambda^k
\norm{\Gamma_\lambda^{(k)}-\widetilde\Gamma_\lambda^{(k)}}_{\gS^2}
+
\norm{k!\lambda^k\widetilde\Gamma_\lambda^{(k)}-\gamma_\mu^{(k)}}_{\gS^2}
\to0.
\end{aligned}
\]
This proves \eqref{eq:density-main-limit}.
\end{proof}

It remains to prove Corollary~\ref{co:Lr-kernal-convergence}.

\begin{proof}[Proof of Corollary~\ref{co:Lr-kernal-convergence}]
    Recall that, by the properties of the free Gibbs state,
\[
    k!\lambda^k\Gamma_0^{(k)}
    =
    k!\mathsf C_\lambda^{\otimes k},
\]
see \cite[(5.33)]{LNR21}. Hence, using the symmetrized integral kernel of
\(\mathsf C_\lambda^{\otimes k}\), the triangle inequality and Fubini's
theorem give
\[
    \left\|
    k!\lambda^k\Gamma_0^{(k)}
    \right\|_{L^s((\bT^d)^{2k})}
    \leq
    k!\,
    \norm{\mathsf C_\lambda}_{L^s(\bT^d\times\bT^d)}^k.
\]
The one-particle kernels \(\mathsf C_\lambda\) are uniformly bounded in
\(L^s(\bT^d\times\bT^d)\) for every finite \(s\) when \(d=2\), and for
every \(1\leq s<3\) when \(d=3\).  Consequently, Proposition~\ref{prop:Ginibre-domination} yields
\[
    \sup_{0<\lambda\leq1}
    \left\|
    k!\lambda^k\Gamma_\lambda^{(k)}
    \right\|_{L^s((\bT^d)^{2k})}\leq
    \sup_{0<\lambda\leq1}
    \left\|
    k!\lambda^k\Gamma_0^{(k)}
    \right\|_{L^s((\bT^d)^{2k})}
    <\infty
\]
in the same range. Combining this with the Hilbert--Schmidt convergence in Theorem~\ref{thm:quantum-classical-main}, the finite-volume embedding for exponents below \(2\), and interpolation for exponents above \(2\), yields the asserted \(L^r\)-kernel convergence in \eqref{eq:kernal-Lr-convergence}.
\end{proof}

\appendix

\bigskip

\section{Gaussian estimates}\label{app:Gaussian}

We collect here the Gaussian estimates used in the construction 
of the nonlinear Gibbs measure and in the proof of the Hilbert--Schmidt convergence of the fixed-order reduced density matrices.

\medskip

\begin{lemma}[Quadratic complex Gaussian variables]\label{lem:Gaussian-Lp}
Let $g_1,\ldots,g_n$ be independent standard complex Gaussian variables and let $\alpha_1,\ldots,\alpha_n\in\R$.  Then for every $p\geq1$,
\begin{equation}\label{eq:Gaussian-Lp-constant}
    \left\|
    \sum_{j=1}^n\alpha_j(|g_j|^2-1)
    \right\|_{L^p}
    \leq
    \left[2\left(1+p\,4^p\Gamma(p)\right)\right]^{1/p}
    \Big(\sum_{j=1}^n\alpha_j^2\Big)^{1/2}.
\end{equation}
\end{lemma}

\medskip

\begin{proof}
Put $\sum_j\alpha_j^2>0$; the case $\sum_j\alpha_j^2=0$ gives zero on both sides of \eqref{eq:Gaussian-Lp-constant}.  For $|t|\leq1/2$,
\begin{align*}
\mathbb E
\exp\left[
t\frac{\sum_j\alpha_j(|g_j|^2-1)}
{\Big(\sum_j\alpha_j^2\Big)^{1/2}}
\right]
 =
\prod_j
\frac{
\exp\left(-t\alpha_j/\left(\sum_i\alpha_i^2\right)^{1/2}\right)
}{
1-t\alpha_j/\left(\sum_i\alpha_i^2\right)^{1/2}
}.
\end{align*}
For $|x|\leq1/2$,
\[
    -x-\log(1-x)
    =
    \int_0^x\frac{s}{1-s}\dd s
    \leq
    \int_0^{|x|}2s\dd s
    =
    x^2 .
\]
Hence
\begin{equation}\label{eq:normalized-quadratic-mgf}
    \mathbb E
    \exp\left[
    t\frac{\sum_j\alpha_j(|g_j|^2-1)}
    {\left(\sum_j\alpha_j^2\right)^{1/2}}
    \right]
    \leq e^{t^2},
    \qquad |t|\leq \frac12 .
\end{equation}
Let $Y=\sum_j\alpha_j(|g_j|^2-1)/(\sum_j\alpha_j^2)^{1/2}$.  For $0\leq y\leq1$, \eqref{eq:normalized-quadratic-mgf} with $t=y/2$ gives
\[
    \mathbb P(Y\geq y)
    \leq
    e^{-y^2/4}.
\]
For $y\geq1$, \eqref{eq:normalized-quadratic-mgf} with $t=1/2$ gives
\[
    \mathbb P(Y\geq y)
    \leq
    e^{-y/2+1/4}
    \leq
    e^{-y/4}.
\]
Applying the same estimates to $-Y$,
\[
    \mathbb P(|Y|\geq y)
    \leq
    2\mathbf 1_{\{0\leq y\leq1\}}
    +
    2e^{-y/4}\mathbf 1_{\{y\geq1\}}.
\]
Thus,
\begin{align*}
    \mathbb E|Y|^p
    &=
    p\int_0^\infty y^{p-1}\mathbb P(|Y|\geq y)\dd y
    \\
    &\leq
    2p\int_0^1 y^{p-1}\dd y
    +
    2p\int_1^\infty y^{p-1}e^{-y/4}\dd y\leq
    2\left(1+p\,4^p\Gamma(p)\right).
\end{align*}
This is \eqref{eq:Gaussian-Lp-constant}.
\end{proof}

\medskip

\begin{lemma}[Explicit lattice sums]\label{lem:lattice-sums}
For every $t\geq0$,
\begin{align}
    \sum_{p\in\Z^3}
\min\left\{
\frac{t}{h(p)},
\frac{t^2}{2h(p)^2}
\right\}
    &\leq 80\,t^{3/2},
    \label{eq:lattice-3d}
    \\
    \sum_{p\in\Z^2}
\min\left\{
\frac{t}{h(p)},
\frac{t^2}{2h(p)^2}
\right\}
    &\leq 20\,t\log(2+t).
    \label{eq:lattice-2d}
\end{align}
\end{lemma}

\medskip

\begin{proof}
In dimension $3$,
\[
    \#\{p\in\Z^3:\|p\|_{\ell^\infty}=m\}
    =
    (2m+1)^3-(2m-1)^3
    =
    24m^2+2
    \leq26m^2
\]
for $m\geq1$.  If $0\leq t\leq1$, then
\begin{align*}
\sum_{p\in\Z^3}
\min\left\{
\frac{t}{h(p)},
\frac{t^2}{2h(p)^2}
\right\}
&\leq
\frac{t^2}{2}
\sum_{p\in\Z^3}\frac1{(|p|^2+1)^2}                                           
\leq
\frac{t^2}{2}
\Big(
1+26\sum_{m\geq1}\frac{m^2}{(m^2+1)^2}
\Big)                                                                      \\
&\leq
\frac{t^2}{2}
\Big(
1+26\sum_{m\geq1}\frac1{m^2}
\Big)
\leq
22t^2
\leq
22t^{3/2}.
\end{align*}
If $t\geq1$, then
\begin{align*}
\sum_{\|p\|_{\ell^\infty}\leq\sqrt t}
\min\left\{
\frac{t}{h(p)},
\frac{t^2}{2h(p)^2}
\right\}
\leq
t\sum_{\|p\|_{\ell^\infty}\leq\sqrt t}\frac1{|p|^2+1}                         
\leq
t\Big(
1+26\sum_{1\leq m\leq\sqrt t}1
\Big)
\leq
27t^{3/2},
\end{align*}
and
\begin{align*}
\sum_{\|p\|_{\ell^\infty}>\sqrt t}
\min\left\{
\frac{t}{h(p)},
\frac{t^2}{2h(p)^2}
\right\}
\leq
13t^2
\sum_{m>\sqrt t}\frac1{m^2}                                                  
\leq
52t^{3/2}.
\end{align*}
The two estimates give \eqref{eq:lattice-3d}.

In dimension $2$,
\[
    \#\{p\in\Z^2:\|p\|_{\ell^\infty}=m\}
    =
    (2m+1)^2-(2m-1)^2
    =
    8m .
\]
If $0\leq t\leq1$, then
\begin{align*}
\sum_{p\in\Z^2}\min\left\{
\frac{t}{h(p)},
\frac{t^2}{2h(p)^2}
\right\}
\leq
\frac{t^2}{2}
\Big(
1+8\sum_{m\geq1}\frac{m}{(m^2+1)^2}
\Big)                                                                      
\leq
6t^2
\leq
20t\log(2+t).
\end{align*}
If $t\geq1$, the low-frequency part satisfies
\begin{align*}
\sum_{\|p\|_{\ell^\infty}\leq\sqrt t}
\min\left\{
\frac{t}{h(p)},
\frac{t^2}{2h(p)^2}
\right\}
\leq
t\Big(
1+8\sum_{1\leq m\leq\sqrt t}\frac{m}{m^2+1}
\Big)                                                                      
\leq
t\left(9+4\log t\right),
\end{align*}
while the high-frequency part gives
\begin{align*}
\sum_{\|p\|_{\ell^\infty}>\sqrt t}
\min\left\{
\frac{t}{h(p)},
\frac{t^2}{2h(p)^2}
\right\}
\leq
4t^2\sum_{m>\sqrt t}\frac{m}{(m^2+1)^2}                                          
\leq
4t^2\sum_{m>\sqrt t}\frac1{m^3}
\leq
8t .
\end{align*}
Since $17+4\log t\leq20\log(2+t)$ for $t\geq1$, \eqref{eq:lattice-2d}
follows.
\end{proof}

\medskip

It remains to prove Lemma~\ref{lem:common-Gaussian-coordinate}.

\medskip

\begin{proof}[Proof of Lemma~\ref{lem:common-Gaussian-coordinate}]
Since $0\leq \mathsf{C}_{\lambda}\leq h^{-1}$,
\[
    \mathbb E\|u_\lambda\|_{H^{-\sigma}}^2
    =
    \sum_{p\in\Z^d}h(p)^{-\sigma}c_\lambda(p)
    \leq
    \sum_{p\in\Z^d}h(p)^{-\sigma-1}
    <\infty .
\]
For $f_1,\ldots,f_m\in H^\sigma$, the vector
\[
    \big(\inprod{f_1}{u_\lambda},\ldots,\inprod{f_m}{u_\lambda}\big)
\]
is centered complex Gaussian and satisfies
\[
    \mathbb E\left[
    \overline{\inprod{f_i}{u_\lambda}}
    \inprod{f_j}{u_\lambda}
    \right]
    =
    \inprod{f_j}{\mathsf{C}_{\lambda} f_i}.
\]
Thus $(u_\lambda)_\#\mathbb P=\mu_\lambda$, and the uniform convergence of $\widehat{\cM}_{\lambda,K}^A$ follows from Lemma~\ref{lem:Gaussian-Lp}.

For $K<\infty$, set
\[
    \widehat{\cD}_{\eta,K}
    =
    \frac16
    \int_\Omega\dd\nu(\omega)
    \int_{\bT^d}
    \left(
    \widehat{\cM}_{\eta,K}^{\tau_rJ_\omega}
    \right)^3\dd r,
    \qquad 0\leq \eta\leq1 .
\]
By Proposition~\ref{prop:D-L1}, pulled back through $u_\eta$,
\[
    \lim_{K\to\infty}
    \sup_{\eta\in[0,1]}
    \norm{
    \widehat{\cD}_{\eta,K}
    -
    \widehat{\cD}_{\eta}
    }_{L^1(\gX)}
    =
    0 .
\]
For fixed $K$ and $A=\tau_rJ_\omega$,
\[
\begin{aligned}
&
\widehat{\cM}_{\lambda,K}^{A}
-
\widehat{\cM}_{0,K}^{A}
\\
&=
\inprod{g}{
\left(
P_K\mathsf{C}_{\lambda}^{1/2}A\mathsf{C}_{\lambda}^{1/2}P_K
-
P_Kh^{-1/2}Ah^{-1/2}P_K
\right)g}
\\
&\quad-
\tr_{\gH}
\left(
P_K\mathsf{C}_{\lambda}^{1/2}A\mathsf{C}_{\lambda}^{1/2}P_K
-
P_Kh^{-1/2}Ah^{-1/2}P_K
\right).
\end{aligned}
\]
Since $K<\infty$,
\[
    P_K\mathsf{C}_{\lambda}^{1/2}A\mathsf{C}_{\lambda}^{1/2}P_K
    \to
    P_Kh^{-1/2}Ah^{-1/2}P_K
    \qquad\text{in }\gS^2 .
\]
Lemma~\ref{lem:Gaussian-Lp} gives
\[
    \norm{
    \widehat{\cM}_{\lambda,K}^{\tau_rJ_\omega}
    -
    \widehat{\cM}_{0,K}^{\tau_rJ_\omega}
    }_{L^3(\gX)}
    \to0 .
\]
Using $x^3-y^3=(x-y)(x^2+xy+y^2)$, Hölder's inequality, and
\eqref{eq:MK-channel-uniform-Lp},
\[
\begin{aligned}
&
\mathbb E
\left|
\left(\widehat{\cM}_{\lambda,K}^{\tau_rJ_\omega}\right)^3
-
\left(\widehat{\cM}_{0,K}^{\tau_rJ_\omega}\right)^3
\right|
\\
&\leq
\norm{
\widehat{\cM}_{\lambda,K}^{\tau_rJ_\omega}
-
\widehat{\cM}_{0,K}^{\tau_rJ_\omega}
}_{L^3(\gX)}
\Big[
\norm{\widehat{\cM}_{\lambda,K}^{\tau_rJ_\omega}}_{L^3(\gX)}^2
+
\norm{\widehat{\cM}_{0,K}^{\tau_rJ_\omega}}_{L^3(\gX)}^2
\\
&
+
\norm{\widehat{\cM}_{\lambda,K}^{\tau_rJ_\omega}}_{L^3(\gX)}
\norm{\widehat{\cM}_{0,K}^{\tau_rJ_\omega}}_{L^3(\gX)}
\Big]
\to0,
\end{aligned}
\]
and the left side is bounded by
$
    C\norm{J_\omega}_{L^\infty}^3\tr h^{-2}.
$
Thus \eqref{eq:J-Linf-assumption} and dominated convergence imply
\[
    \widehat{\cD}_{\lambda,K}
    \to
    \widehat{\cD}_{0,K}
    \qquad\text{in }L^1(\gX)
\]
for every fixed $K$.  Hence
\[
\begin{aligned}
\norm{\widehat{\cD}_\lambda-\widehat{\cD}_0}_{L^1(\gX)}
\leq
\norm{\widehat{\cD}_\lambda-\widehat{\cD}_{\lambda,K}}_{L^1(\gX)}
+
\norm{\widehat{\cD}_{\lambda,K}-\widehat{\cD}_{0,K}}_{L^1(\gX)}
+
\norm{\widehat{\cD}_{0,K}-\widehat{\cD}_{0}}_{L^1(\gX)} .
\end{aligned}
\]
Taking $\limsup_{\lambda\downarrow0}$ and then $K\to\infty$ gives
\[
    \widehat{\cD}_\lambda\to\widehat{\cD}_0
    \qquad\text{in }L^1(\gX).
\]
The identifications
\[
    \widehat{\cD}_\lambda=\cD_\lambda\circ u_\lambda
    \quad(\lambda>0),
    \qquad
    \widehat{\cD}_0=\cD_0\circ u_0
\]
follow from the identifications of $\widehat{\cM}_\lambda^A$, the definition
of $\cD_\lambda$ for $\lambda>0$, and the definition of $\cD_0$ as the
$L^1(\mu_0)$-limit of $\cD_{0,K}$.  Finally,
\[
    \mathbb E e^{-\widehat{\cD}_\lambda}
    =
    \int e^{-\cD_\lambda[u]}\dd\mu_\lambda(u)
    =
    z_\lambda
    \qquad(\lambda>0),
\]
and
\[
    \mathbb E e^{-\widehat{\cD}_0}
    =
    \int e^{-\cD_0[u]}\dd\mu_0(u)
    =
    z.
\]
\end{proof}

\bigskip

\section{Berezin--Lieb inequalities for the relative entropy}
\label{app:relative-entropy-Berezin-Lieb}

We collect the lower- and upper-symbol inequalities in a common form.

\medskip

\begin{lemma}[Berezin--Lieb inequalities for the relative entropy]
\label{lem:relative-entropy-Berezin-Lieb}
Let $\mathfrak K$ be a separable Hilbert space, let $M$ be a separable
metric space, and let
\[
    M\ni x\longmapsto n_x\in S\mathfrak K
    :=
    \{u\in\mathfrak K:\norm u_{\mathfrak K}=1\}
\]
be continuous.  For every Borel probability measure $\mu$ on $M$, let
\[
    A_\mu
    :=
    \int_M
    \ket{n_x}\bra{n_x}\,\dd\mu(x)
    \in\gS^1(\mathfrak K).
\]
Then, for all Borel probability measures $\mu,\mu'$ on $M$,
\begin{equation}
    \cH(A_\mu,A_{\mu'})
    \leq
    \cH_{\cl}(\mu,\mu').
    \label{eq:relative-entropy-BL-upper}
\end{equation}
Assume in addition that $\zeta$ is a positive Borel measure on $M$
satisfying
\begin{equation}
    \int_M
    \ket{n_x}\bra{n_x}\,\dd\zeta(x)
    =
    \1_{\mathfrak K}
    \label{eq:relative-entropy-resolution-identity}
\end{equation}
in the weak operator sense.  For every state $A$ on $\mathfrak K$,
define its Husimi measure, or lower symbol, by
\(
    \dd m_A(x)
    :=
    \bangle{n_x,A n_x}\,\dd\zeta(x).
\)
Then, for all states $A,B$ on $\mathfrak K$,
\begin{equation}
    \cH_{\cl}(m_A,m_B)
    \leq
    \cH(A,B).
    \label{eq:relative-entropy-BL-lower}
\end{equation}
Consequently, under
\eqref{eq:relative-entropy-resolution-identity},
\begin{equation}
    \cH_{\cl}(m_{A_\mu},m_{A_{\mu'}})
    \leq
    \cH(A_\mu,A_{\mu'})
    \leq
    \cH_{\cl}(\mu,\mu').
    \label{eq:relative-entropy-BL-sandwich}
\end{equation}
\end{lemma}

\medskip

\begin{proof}
We use monotonicity of the relative entropy under CPTP maps,
\begin{equation}
    \cH(\Phi(A),\Phi(B))
    \leq
    \cH(A,B),
    \label{eq:relative-entropy-data-processing}
\end{equation}
where $\Phi:\gS^1(\mathfrak K_1)\to\gS^1(\mathfrak K_2)$ may act
between different Hilbert spaces \cite{Lindblad75,OhyaPetz}.

The proof of \eqref{eq:relative-entropy-BL-lower} is given in
\cite[Lemma~7.2]{LNR15}; we reproduce it here for
completeness. Let
$\{E_1,\ldots,E_K\}$ be a finite Borel partition of $M$.  The map
\[
    \Phi(A)
    :=
    \operatorname{diag}
    \left(
        \int_{E_1}\bangle{n_x,A n_x}\,\dd\zeta(x),
        \ldots,
        \int_{E_K}\bangle{n_x,A n_x}\,\dd\zeta(x)
    \right)
\]
from $\gS^1(\mathfrak K)$ to $\gS^1(\CC^K)$ is completely positive:
each diagonal entry is a positive linear functional, hence completely
positive.  It is trace preserving by
\eqref{eq:relative-entropy-resolution-identity}.  Therefore
\eqref{eq:relative-entropy-data-processing} gives
\[
    \cH(A,B)
    \geq
    \sum_{k=1}^K
    m_A(E_k)
    \log
    \left(
        \frac{m_A(E_k)}{m_B(E_k)}
    \right).
\]
Taking the supremum over finite Borel partitions yields
\eqref{eq:relative-entropy-BL-lower}, since
\begin{equation}
    \cH_{\cl}(\mu,\mu')
    =
    \sup_{\{E_1,\ldots,E_K\}}
    \sum_{k=1}^K
    \mu(E_k)
    \log
    \left(
        \frac{\mu(E_k)}{\mu'(E_k)}
    \right).
    \label{eq:finite-partition-relative-entropy}
\end{equation}
Here the supremum is taken over all finite measurable partitions of
\(M\).  We use the convention that a term with
\(\mu(E_k)=0\) is zero, whereas it is \(+\infty\) if
\(\mu(E_k)>0=\mu'(E_k)\).

We next prove \eqref{eq:relative-entropy-BL-upper}.  The assertion is
immediate when $\cH_{\cl}(\mu,\mu')=+\infty$, so assume that the latter
is finite.  Since
\[
    \left\|
        \ket{n_x}\bra{n_x}
        -
        \ket{n_y}\bra{n_y}
    \right\|_{\gS^1}
    =
    2\sqrt{1-\abs{\bangle{n_x,n_y}}^2}
    \leq
    2\norm{n_x-n_y}_{\mathfrak K},
\]
the map $x\mapsto\ket{n_x}\bra{n_x}$ is continuous in
$\gS^1(\mathfrak K)$.  Since \(M\) is separable and
\(
    x\longmapsto
    \ket{n_x}\bra{n_x}
\)
is continuous from \(M\) to \(\gS^1(\mathfrak K)\), its range is
separable in the trace norm.  We may therefore choose a sequence
\((y_\ell)_{\ell\geq1}\subseteq M\) such that
\(
    \left\{
        \ket{n_{y_\ell}}\bra{n_{y_\ell}}
        :
        \ell\geq1
    \right\}
\)
is dense in this range.  For every \(j\geq1\), define the pairwise
disjoint Borel sets
\[
\begin{aligned}
    F_{1,j}
    &:=
    \left\{
        x\in M:
        \left\|
            \ket{n_x}\bra{n_x}
            -
            \ket{n_{y_1}}\bra{n_{y_1}}
        \right\|_{\gS^1}
        <\frac1j
    \right\},
    \\
    F_{\ell,j}
    &:=
    \left\{
        x\in M:
        \left\|
            \ket{n_x}\bra{n_x}
            -
            \ket{n_{y_\ell}}\bra{n_{y_\ell}}
        \right\|_{\gS^1}
        <\frac1j
    \right\}
    \setminus
    \bigcup_{r=1}^{\ell-1}F_{r,j},
    \qquad \ell\geq2.
\end{aligned}
\]
These sets form a countable Borel partition of \(M\).  Indeed, their
pairwise disjointness follows from the definition, while their union
is all of \(M\) by the density of
\(
\{\ket{n_{y_\ell}}\bra{n_{y_\ell}}\}_{\ell\geq1}
\)
in the range of
\(x\mapsto\ket{n_x}\bra{n_x}\).

Since \(\mu+\mu'\) is finite and
\[
    \bigcup_{\ell=1}^N F_{\ell,j}
    \uparrow M
    \qquad\text{as }N\to\infty,
\]
continuity from below gives an integer \(N_j\) such that
\(
    (\mu+\mu')
    \left(
        M\setminus
        \bigcup_{\ell=1}^{N_j}F_{\ell,j}
    \right)
    <\frac1j.
\)
After removing the empty sets and relabelling, we obtain a
finite Borel partition
\(
    \{E_{1,j},\ldots,E_{K_j,j}\}
\)
of \(M\).

Choose \(x_{k,j}\in E_{k,j}\) for every \(k\).  If \(E_{k,j}\) is one
of the sets inherited from \(F_{\ell,j}\), then both \(x\) and
\(x_{k,j}\) belong to the trace-norm ball of radius \(j^{-1}\)
centered at
\(\ket{n_{y_\ell}}\bra{n_{y_\ell}}\).  Consequently, for every
\(x\in E_{k,j}\),
\[
\begin{aligned}
    &
    \left\|
        \ket{n_x}\bra{n_x}
        -
        \ket{n_{x_{k,j}}}\bra{n_{x_{k,j}}}
    \right\|_{\gS^1}
    \\
    &\qquad\leq
    \left\|
        \ket{n_x}\bra{n_x}
        -
        \ket{n_{y_\ell}}\bra{n_{y_\ell}}
    \right\|_{\gS^1}
    +
    \left\|
        \ket{n_{x_{k,j}}}\bra{n_{x_{k,j}}}
        -
        \ket{n_{y_\ell}}\bra{n_{y_\ell}}
    \right\|_{\gS^1}
    <\frac2j.
\end{aligned}
\]
On the remaining cell we use
\[
    \left\|
        \ket{n_x}\bra{n_x}
        -
        \ket{n_{x_{k,j}}}\bra{n_{x_{k,j}}}
    \right\|_{\gS^1}
    \leq
    \left\|\ket{n_x}\bra{n_x}\right\|_{\gS^1}
    +
    \left\|\ket{n_{x_{k,j}}}\bra{n_{x_{k,j}}}\right\|_{\gS^1}
    =
    2.
\]
It follows that
\begin{align}
    &
    \sum_{k=1}^{K_j}
    \int_{E_{k,j}}
    \left\|
        \ket{n_x}\bra{n_x}
        -
        \ket{n_{x_{k,j}}}\bra{n_{x_{k,j}}}
    \right\|_{\gS^1}
    \dd(\mu+\mu')(x)
    \notag
    \\
    &\qquad\leq
    \frac2j(\mu+\mu')(M)
    +
    2(\mu+\mu')
    \Big(
        M\setminus
        \bigcup_{\ell=1}^{N_j}F_{\ell,j}
    \Big)
    <
    \frac{2(\mu+\mu')(M)+2}{j}
    \longrightarrow0.
    \label{eq:relative-entropy-simple-approximation}
\end{align}
If $e_1,\ldots,e_{K_j}$ denotes the canonical basis of
$\CC^{K_j}$, then
\[
    C
    \longmapsto
    \sum_{k=1}^{K_j}
    \bangle{e_k,C e_k}
    \ket{n_{x_{k,j}}}\bra{n_{x_{k,j}}}
\]
is completely positive and trace preserving. Applying
\eqref{eq:relative-entropy-data-processing} to the diagonal states
with entries $(\mu(E_{k,j}))$ and $(\mu'(E_{k,j}))$, and then using
\eqref{eq:finite-partition-relative-entropy}, gives
\begin{align}
    &\cH\Big(
        \sum_{k=1}^{K_j}
        \mu(E_{k,j})\ket{n_{x_{k,j}}}\bra{n_{x_{k,j}}},
        \sum_{k=1}^{K_j}
        \mu'(E_{k,j})\ket{n_{x_{k,j}}}\bra{n_{x_{k,j}}}
    \Big)
    \notag
    \\
    &\qquad\leq
    \sum_{k=1}^{K_j}
    \mu(E_{k,j})
    \log\left(\frac{\mu(E_{k,j})}{\mu'(E_{k,j})}\right)
    \leq
    \cH_{\cl}(\mu,\mu').
    \label{eq:relative-entropy-discrete-upper}
\end{align}
Moreover,
\eqref{eq:relative-entropy-simple-approximation} implies
\[
\begin{aligned}
    \Big\|
        A_\mu
        -
        \sum_{k=1}^{K_j}
        \mu(E_{k,j})\ket{n_{x_{k,j}}}\bra{n_{x_{k,j}}}
    \Big\|_{\gS^1}
    +
    \Big\|
        A_{\mu'}
        -
        \sum_{k=1}^{K_j}
        \mu'(E_{k,j})\ket{n_{x_{k,j}}}\bra{n_{x_{k,j}}}
    \Big\|_{\gS^1}
    \longrightarrow0.
\end{aligned}
\]
The joint lower semicontinuity of the relative entropy under
trace-norm convergence now proves
\eqref{eq:relative-entropy-BL-upper}.  Applying
\eqref{eq:relative-entropy-BL-lower} to $A_\mu$ and $A_{\mu'}$ gives
\eqref{eq:relative-entropy-BL-sandwich}.
\end{proof}

\bigskip

\section{Quasi-free trace identity}

\begin{lemma}[Quasi-free trace identity]\label{lem:quasi-free-trace-identity}
Let $A=A^*$ be bounded and assume $A\geq0$.  Then
\[
    \gamma_0^{1/2}(1-e^{-A})\gamma_0^{1/2}\in\gS^1(\gH)
\]
and
\begin{equation}\label{eq:quasi-free-trace-identity}
    \tr_{\gF}\left(e^{-\dG(A)}\Gamma_0\right)
    =
    \exp\left(
    -
    \tr_{\gH}
    \log\left[
    1+
    \gamma_0^{1/2}(1-e^{-A})\gamma_0^{1/2}
    \right]
    \right).
\end{equation}
\end{lemma}

\begin{proof}
Since $0\leq1-e^{-A}\leq1$ and
$\gamma_0\in\gS^1(\gH)$,
\[
    0\leq
    \gamma_0^{1/2}(1-e^{-A})\gamma_0^{1/2}
    \leq
    \gamma_0.
\]
This proves the trace-class assertion.

By Lemma~\ref{lem:free-coherent-representation} and
$e^{-\dG(A)}=\Gamma(e^{-A})$,
\begin{align*}
    \tr_{\gF}\left(e^{-\dG(A)}\Gamma_0\right)
    &=
    \int
    \left\langle
        \xi(u/\sqrt{\lambda}),
        \Gamma(e^{-A})\xi(u/\sqrt{\lambda})
    \right\rangle
    \dd\mu_\lambda(u)
    \\
    &=
    \int
    \exp\left(
        -\frac1\lambda
        \inprod{u}{(1-e^{-A})u}
    \right)
    \dd\mu_\lambda(u).
\end{align*}
Indeed, for every $v\in\gH$,
\[
    \left\langle\xi(v),\Gamma(e^{-A})\xi(v)\right\rangle
    =
    e^{-\norm v^2}
    \sum_{n=0}^{\infty}
    \frac{\inprod{v}{e^{-A}v}^{\,n}}{n!}
    =
    e^{-\inprod{v}{(1-e^{-A})v}}.
\]

Let $(\rho_j)_{j\geq1}$ be the eigenvalues of the positive trace-class
operator
\[
    \gamma_0^{1/2}(1-e^{-A})\gamma_0^{1/2}.
\]
Since $\mu_\lambda$ has covariance
$\mathsf C_\lambda=\lambda\gamma_0$, represent $u$ by
$(\lambda\gamma_0)^{1/2}g$ in finite-dimensional approximations and
diagonalize the operator above.  Unitary invariance of the standard
complex Gaussian vector gives
\[
    \frac1\lambda\inprod{u}{(1-e^{-A})u}
    \stackrel{\mathrm{law}}=
    \sum_{j\geq1}\rho_j|g_j|^2,
\]
where the $g_j$'s are independent standard complex Gaussian variables.
Since $\sum_j\rho_j<\infty$, the series on the right is finite almost
surely, and dominated convergence applied to its partial sums gives
\begin{align*}
    \tr_{\gF}\left(e^{-\dG(A)}\Gamma_0\right)
    &=
    \prod_{j\geq1}
    \mathbb E\left[e^{-\rho_j|g_j|^2}\right]
    =
    \prod_{j\geq1}\frac1{1+\rho_j}
    \\
    &=
    \exp\left(
        -\sum_{j\geq1}\log(1+\rho_j)
    \right)
    \\
    &=
    \exp\left(
        -
        \tr_{\gH}
        \log\left[
            1+
            \gamma_0^{1/2}(1-e^{-A})\gamma_0^{1/2}
        \right]
    \right).
\end{align*}
Here we used
$\mathbb E[e^{-t|g_j|^2}]=(1+t)^{-1}$ for $t\geq0$.
This proves \eqref{eq:quasi-free-trace-identity}.
\end{proof}

\bigskip

\section{Ginibre loop representation}\label{app:Ginibre-loop-representation}

We recall the Ginibre loop representation, originating in
\cite{Ginibre65,ginibre1965reduced2,ginibre1965reduced3,
ginibre1971some}.  See also
\cite{FKSS26,frohlich2020path,garouniatis2026large} for recent
developments in loop representations for interacting Bose gases.

Let $h=-\Delta+1$ on $\bT^d$, and denote by
\[
        G_t(x,y)=e^{-th}(x,y)
\]
the heat kernel.  For $t>0$ and $x,y\in\bT^d$, let
$\mathbb{B}_{x,y}^t$ be the unnormalized Brownian bridge measure from $x$ to
$y$ in time $t$, normalized by
\[
        \int \dd\mathbb{B}_{x,y}^t(\omega)=G_t(x,y).
\]
The only property of $\mathbb{B}_{x,y}^t$ used below is the concatenation identity:
for $0<t_1,t_2<\infty$,
\begin{equation}
\label{eq:bridge-concat-two}
 \int_{\bT^d}
 \dd\mathbb{B}_{x,z}^{t_1}(\omega_1)\,
 \dd\mathbb{B}_{z,y}^{t_2}(\omega_2)\,\dd z
 =
 \dd\mathbb{B}_{x,y}^{t_1+t_2}(\omega_1\star\omega_2),
\end{equation}
where $\omega_1\star\omega_2$ denotes the path obtained by following
$\omega_1$ and then $\omega_2$.  Equivalently, for every bounded measurable
function $\Phi$ of the concatenated path,
\[
 \int_{\bT^d}\int
 \Phi(\omega_1\star\omega_2)
 \dd\mathbb{B}_{x,z}^{t_1}(\omega_1)
 \dd\mathbb{B}_{z,y}^{t_2}(\omega_2)\,\dd z
 =
 \int \Phi(\omega)\dd\mathbb{B}_{x,y}^{t_1+t_2}(\omega).
\]
This is the path-measure form of the semigroup identity
$\int_{\bT^d}G_{t_1}(x,z)G_{t_2}(z,y)\dd z=G_{t_1+t_2}(x,y)$.

Let $\underline{X}_n=(x_1,\ldots,x_n)$ and $\underline{Y}_n=(y_1,\ldots,y_n)$.  On the
$n$-particle sector, recall that the renormalized interaction is the symmetric multiplication
operator
\[
 \bW_{\lambda,n}^{{\rm ren}}(\underline{X}_n)
 =
 \frac16\int_\Omega\dd\nu(\omega)\int_{\bT^d}
 \Big(
 \lambda\sum_{j=1}^nJ_\omega(x_j-r)-m_{\lambda,\omega}
 \Big)^3
 \dd r .
\]
For every
$n\geq1$, set
$
        K_{\lambda,n}^{\rm bos}(\underline{X}_n;\underline{Y}_n)
$
to be the integral kernel on $\gH^{\otimes_s n}$ of
\[
        \exp\Big\{-\lambda\sum_{j=1}^n h_j
        -\bW_{\lambda,n}^{{\rm ren}}\Big\}.
\]
Since $\bW_{\lambda,n}^{{\rm ren}}$ is symmetric in $\underline{X}_n$, Feynman--Kac on
the full tensor product followed by the bosonic symmetrizer gives
\begin{align}
\label{eq:FK-bosonic-kernel}
K_{\lambda,n}^{\rm bos}(\underline{X}_n;\underline{Y}_n)
=
\frac1{n!}\sum_{\sigma\in S_n}
\int
\exp\left[
-\frac1\lambda\int_0^\lambda
\bW_{\lambda,n}^{{\rm ren}}
\bigl(\omega_1(s),\ldots,\omega_n(s)\bigr)\dd s
\right]                                                   
\prod_{j=1}^n
\dd\mathbb{B}_{x_j,y_{\sigma(j)}}^\lambda(\omega_j),
\end{align}
where $S_n$ is the symmetric group of $n$ elements.

\subsection{Closed-loop expansion of the Fock trace}

Recall that the partition function is defined by
\[
        Z_{\lambda}
        =
        \tr_{\gF}
        \exp\{-\lambda d\Gamma(h)-\bW_\lambda^{{\rm ren}}\}.
\]
Using \eqref{eq:FK-bosonic-kernel} with $\underline{Y}_n=\underline{X}_n$, we obtain
\begin{align}
\label{eq:Fock-trace-first}
Z_{\lambda}
&=
\sum_{n\geq0}
\int_{(\bT^d)^n}
K_{\lambda,n}^{\rm bos}(\underline{X}_n;\underline{X}_n)\dd \underline{X}_n                                           \notag\\
&=
\sum_{n\geq0}\frac1{n!}
\sum_{\sigma\in S_n}
\int_{(\bT^d)^n}
\int
\exp\left[
-\frac1\lambda\int_0^\lambda
\bW_{\lambda,n}^{{\rm ren}}
\bigl(\omega_1(s),\ldots,\omega_n(s)\bigr)\dd s
\right]                                                                           \notag\\
&\hspace{5.2cm}\times
\prod_{j=1}^n
\dd\mathbb{B}_{x_j,x_{\sigma(j)}}^\lambda(\omega_j)\dd \underline{X}_n .
\end{align}

Let a cycle of $\sigma$ be
\[
        c=(i_1\,i_2\,\cdots\,i_\ell).
\]
The part of the product in \eqref{eq:Fock-trace-first} associated with this cycle is
\[
        \dd\mathbb{B}_{x_{i_1},x_{i_2}}^\lambda(\omega_{i_1})
        \dd\mathbb{B}_{x_{i_2},x_{i_3}}^\lambda(\omega_{i_2})
        \cdots
        \dd\mathbb{B}_{x_{i_\ell},x_{i_1}}^\lambda(\omega_{i_\ell}).
\]
Repeated use of \eqref{eq:bridge-concat-two} gives, for every bounded measurable
function of the concatenated loop $\eta$,
\begin{align}
\label{eq:closed-cycle-concat}
&\int_{(\bT^d)^{\ell-1}}
\Phi(\eta)
\prod_{q=1}^{\ell}
\dd\mathbb{B}_{x_{i_q},x_{i_{q+1}}}^{\lambda}
(\omega_{i_q})
\,\dd x_{i_2}\cdots\dd x_{i_\ell}                                      
=
\int
\Phi(\eta)\,
\dd\mathbb{B}_{x_{i_1},x_{i_1}}^{\ell\lambda}(\eta),
\qquad i_{\ell+1}=i_1 .
\end{align}
For $s\in[0,\lambda]$, the loop $\eta$ of time length $\ell\lambda$
produces the $\ell$ positions
\[
        \eta(s),\eta(s+\lambda),\ldots,\eta(s+(\ell-1)\lambda).
\]

For a finite configuration $\eta=(\eta_1,\ldots,\eta_M)$ of closed loops,
let $\cX_s(\eta)$ be the finite multiset of all positions occupied at
time $s\in[0,\lambda]$, including all points separated by integer multiples
of $\lambda$ along each loop.  Define
\begin{equation}
        \cU_\lambda(\eta)
        =
        \frac1\lambda\int_0^\lambda
        \bW_{\lambda,|\cX_s(\eta)|}^{{\rm ren}}
        \bigl(\cX_s(\eta)\bigr)\dd s .
\end{equation}
Introduce the positive loop measure
\begin{equation}
\label{eq:loop-measure}
        \dd\mathfrak L_\lambda(\eta)
        =
        \sum_{\ell\geq1}\frac1\ell
        \int_{\bT^d}\dd\mathbb{B}_{x,x}^{\ell\lambda}(\eta)\dd x .
\end{equation}
The factor $1/\ell$ in \eqref{eq:loop-measure} is the residue of the
cycle counting: a cycle of length $\ell$ has $\ell$ possible starting
points.  More precisely, if a permutation has $r_\ell$ cycles of length
$\ell$, then
\[
        \frac1{n!}\#\{\sigma\in S_n:
        \sigma\hbox{ has }r_\ell\hbox{ cycles of length }\ell
        \hbox{ for all }\ell\}
        =
        \prod_{\ell\geq1}\frac1{\ell^{r_\ell}r_\ell!}.
\]
Using this identity and \eqref{eq:closed-cycle-concat} in
\eqref{eq:Fock-trace-first} yields
\begin{align}
\label{eq:Z-loop-expanded}
Z_{\lambda}
&=
\sum_{M=0}^{\infty}\frac1{M!}
\int
e^{-\cU_\lambda(\eta_1,\ldots,\eta_M)}
\prod_{a=1}^M\dd\mathfrak L_\lambda(\eta_a).
\end{align}

For the free partition function,
\[
        Z_0=Z_0(\lambda)=\tr_{\gF}e^{-\lambda d\Gamma(h)},
\]
the same computation gives
\begin{align}
\label{eq:Z0-loop}
Z_0
&=
\sum_{M=0}^{\infty}\frac1{M!}
\int\prod_{a=1}^M\dd\mathfrak L_\lambda(\eta_a)
=
\exp\Big(\int\dd\mathfrak L_\lambda(\eta)\Big)                         \notag\\
&=
\exp\Big(
\sum_{\ell\geq1}\frac1\ell\tr_{\gH}e^{-\ell\lambda h}
\Big)
=
\det(1-e^{-\lambda h})^{-1}.
\end{align}
For a nonnegative function $F$ on finite closed-loop configurations, set
\begin{equation}
\label{eq:E-loop}
\mathbb E_{\rm loop}[F(\eta)]
=
Z_0^{-1}
\sum_{M=0}^{\infty}\frac1{M!}
\int
F(\eta_1,\ldots,\eta_M)
\prod_{a=1}^M\dd\mathfrak L_\lambda(\eta_a).
\end{equation}
Then \eqref{eq:Z-loop-expanded} and \eqref{eq:Z0-loop} give
\begin{equation}
        \frac{Z_{\lambda}}{Z_0}
        =
        \mathbb E_{\rm loop}
        \left[e^{-\cU_\lambda(\eta)}\right].
\end{equation}

\subsection{Open paths and the partial trace}

Let $\underline{X}_k=(x_1,\ldots,x_k)$, $\underline{Y}_k=(y_1,\ldots,y_k)$, and
$\underline{Z}_m=(z_1,\ldots,z_m)$.  Put
\[
        \underline{U}_{k+m}=(x_1,\ldots,x_k,z_1,\ldots,z_m),
        \qquad
        \underline{V}_{k+m}=(y_1,\ldots,y_k,z_1,\ldots,z_m).
\]
By the definition of the $k$-particle reduced density matrix,
\begin{align}
\label{eq:partial-trace-start}
Z_{\lambda}\Gamma_\lambda^{(k)}(\underline{X}_k;\underline{Y}_k)
&=
\sum_{m\geq0}\binom{k+m}{k}
\int_{(\bT^d)^m}
K_{\lambda,k+m}^{\rm bos}(\underline{U}_{k+m};\underline{V}_{k+m})\dd \underline{Z}_m .
\end{align}
Substituting \eqref{eq:FK-bosonic-kernel} into
\eqref{eq:partial-trace-start} and using
$\binom{k+m}{k}/(k+m)!=1/(k!m!)$, we get
\begin{align}
\label{eq:partial-trace-FK}
&Z_{\lambda}\Gamma_\lambda^{(k)}(\underline{X}_k;\underline{Y}_k) \notag\\
&\quad=
\sum_{m\geq0}\frac1{k!m!}
\sum_{\sigma\in S_{k+m}}
\int_{(\bT^d)^m}\int
\exp\left[
-\frac1\lambda\int_0^\lambda
\bW_{\lambda,k+m}^{{\rm ren}}
\bigl(\omega_1(s),\ldots,\omega_{k+m}(s)\bigr)\dd s
\right]                                                                  \notag\\
&\qquad\qquad\times
\prod_{j=1}^{k+m}
\dd\mathbb{B}_{u_j,v_{\sigma(j)}}^\lambda(\omega_j)\,\dd \underline{Z}_m .
\end{align}

We now identify, inside \eqref{eq:partial-trace-FK}, the factors which
become open paths.  Suppose a segment of a cycle runs from the distinguished
label $i\in\{1,\ldots,k\}$ to the distinguished label $j\in\{1,\ldots,k\}$
through traced labels $a_1,\ldots,a_{\ell-1}$.  The corresponding factor in
\eqref{eq:partial-trace-FK} is
\[
 \dd\mathbb{B}_{x_i,z_{a_1}}^\lambda
 \dd\mathbb{B}_{z_{a_1},z_{a_2}}^\lambda
 \cdots
 \dd\mathbb{B}_{z_{a_{\ell-1}},y_j}^\lambda
 \,\dd z_{a_1}\cdots\dd z_{a_{\ell-1}} .
\]
Repeated use of \eqref{eq:bridge-concat-two} gives
\begin{align}
\label{eq:open-chain-concat}
&\int
 \dd\mathbb{B}_{x_i,z_{a_1}}^\lambda
 \dd\mathbb{B}_{z_{a_1},z_{a_2}}^\lambda
 \cdots
 \dd\mathbb{B}_{z_{a_{\ell-1}},y_j}^\lambda
 \,\dd z_{a_1}\cdots\dd z_{a_{\ell-1}}
 =
 \dd\mathbb{B}_{x_i,y_j}^{\ell\lambda}.
\end{align}
If $\ell=1$, \eqref{eq:open-chain-concat} means
$\dd\mathbb{B}_{x_i,y_j}^{\lambda}$, with no traced variable integrated.

For a fixed permutation $\pi\in S_k$ and lengths
$\ell_1,\ldots,\ell_k\geq1$, the $k$ open chains use
\[
        q=\sum_{i=1}^k(\ell_i-1)
\]
traced labels.  From the coefficient $1/(k!m!)$ in
\eqref{eq:partial-trace-FK}, the choice and ordering of these $q$ traced
labels gives
\[
        \frac1{k!m!}\cdot\frac{m!}{(m-q)!}
        =
        \frac1{k!(m-q)!}.
\]
The remaining $m-q$ traced labels form closed cycles and are then summed by
the same closed-loop expansion as in \eqref{eq:Z-loop-expanded}.  Hence the
positive open-path measure is
\begin{equation}
\label{eq:open-path-measure}
\dd\mathbb{B}_{X_k,Y_k}^{(k)}(\omega)
=
\frac1{k!}
\sum_{\pi\in S_k}
\sum_{\ell_1,\ldots,\ell_k\geq1}
\prod_{i=1}^k
\dd\mathbb{B}_{x_i,y_{\pi(i)}}^{\ell_i\lambda}(\omega_i).
\end{equation}

For a configuration $\omega$ of $k$ open paths and a configuration
$\eta$ of closed loops, define $\cX_s(\eta\cup\omega)$ by taking all
positions at time $s\in[0,\lambda]$, including all positions separated by
integer multiples of $\lambda$ along each loop or open path, and define
\begin{equation}
\cU_\lambda(\eta\cup\omega)
=
\frac1\lambda\int_0^\lambda
\bW_{\lambda,|\cX_s(\eta\cup\omega)|}^{{\rm ren}}
\bigl(\cX_s(\eta\cup\omega)\bigr)\dd s .
\end{equation}
Combining \eqref{eq:partial-trace-FK}, \eqref{eq:open-chain-concat},
\eqref{eq:loop-measure}, and \eqref{eq:open-path-measure}, we obtain
\begin{align}
\label{eq:partial-trace-loop-open}
Z_{\lambda}\Gamma_\lambda^{(k)}(\underline{X}_k;\underline{Y}_k)
&=
\int
\Big[
\sum_{M=0}^{\infty}\frac1{M!}
\int
e^{-\cU_\lambda(\eta_1,\ldots,\eta_M,\omega)}
\prod_{a=1}^{M}\dd\mathfrak L_\lambda(\eta_a)
\Big]\dd\mathbb{B}_{X_k,Y_k}^{(k)}(\omega).
\end{align}
Dividing \eqref{eq:partial-trace-loop-open} by
\eqref{eq:Z-loop-expanded} and using \eqref{eq:E-loop}, we arrive at the
Ginibre representation:

\medskip

\begin{proposition}[Ginibre representation for the reduced kernels]
For every $k\geq1$,
\begin{equation}
\label{eq:Ginibre-rep}
\Gamma_\lambda^{(k)}(\underline{X}_k;\underline{Y}_k)
=
\int
\frac{
\mathbb E_{\rm loop}
\left[e^{-\cU_\lambda(\eta\cup\omega)}\right]
}{
\mathbb E_{\rm loop}
\left[e^{-\cU_\lambda(\eta)}\right]
}
\dd\mathbb{B}_{X_k,Y_k}^{(k)}(\omega).
\end{equation}
\end{proposition}

\medskip

\subsection{Pointwise kernel domination by the free fixed-order reduced density matrices}

The free one-particle density matrix is
\[
        \gamma_0=(e^{\lambda h}-1)^{-1}
        =
        \sum_{\ell\geq1}e^{-\ell\lambda h}.
\]
Taking the total mass of \eqref{eq:open-path-measure} gives
\begin{align}
\label{eq:open-total-mass}
\int\dd\mathbb{B}_{X_k,Y_k}^{(k)}(\omega)
&=
\frac1{k!}\sum_{\pi\in S_k}
\sum_{\ell_1,\ldots,\ell_k\geq1}
\prod_{i=1}^k
G_{\ell_i\lambda}(x_i,y_{\pi(i)})                                      \notag\\
&=
\frac1{k!}\sum_{\pi\in\mathfrak S_k}
\prod_{i=1}^k
\Big(\sum_{\ell\geq1}G_{\ell\lambda}(x_i,y_{\pi(i)})\Big)              \notag\\
&=
\frac1{k!}\sum_{\pi\in\mathfrak S_k}
\prod_{i=1}^k\gamma_0(x_i,y_{\pi(i)})
=
\Gamma_0^{(k)}(\underline{X}_k;\underline{Y}_k).
\end{align}

We now use the special structure of the three-body interaction. Suppose $\mathcal X\subset \mathcal X'$. Since $J_\omega\geq0$,
\begin{align}
&\lambda\sum_{x\in\mathcal X'}J_\omega(x-r)
=
\lambda\sum_{x\in\mathcal X}J_\omega(x-r)
+
\lambda\sum_{x\in\mathcal X'\setminus\mathcal X}J_\omega(x-r)
\geq
\lambda\sum_{x\in\mathcal X}J_\omega(x-r).
\end{align}
It follows that
\begin{align}
\label{eq:sector-monotone}
\bW_{\lambda,|\mathcal X'|}^{{\rm ren}}(\mathcal X')\geq\bW_{\lambda,|\mathcal X|}^{{\rm ren}}(\mathcal X).
\end{align}
For every $s\in[0,\lambda]$,
$
        \cX_s(\eta)\subset \cX_s(\eta\cup\omega),
$
as finite multisets.  Combining this inclusion with
\eqref{eq:sector-monotone} gives
\[
\bW_{\lambda,|\cX_s(\eta\cup\omega)|}^{{\rm ren}}
\bigl(\cX_s(\eta\cup\omega)\bigr)
\geq
\bW_{\lambda,|\cX_s(\eta)|}^{{\rm ren}}
\bigl(\cX_s(\eta)\bigr).
\]
After integration in $s$, this is
$
        \cU_\lambda(\eta\cup\omega)\geq \cU_\lambda(\eta).
$
Using the normalized loop series \eqref{eq:E-loop},
\[
0\leq
\mathbb E_{\rm loop}
\left[e^{-\cU_\lambda(\eta\cup\omega)}\right]
\leq
\mathbb E_{\rm loop}
\left[e^{-\cU_\lambda(\eta)}\right].
\]
The denominator is strictly positive, since the integrand is positive and the
empty-loop term is present.  Hence the ratio in \eqref{eq:Ginibre-rep}
belongs to $[0,1]$.

\medskip

\begin{proposition}[Pointwise kernel domination]\label{prop:Ginibre-domination}
For every $k\geq1$, the integral kernels of $\Gamma_\lambda^{(k)}$ and
$\Gamma_0^{(k)}$ are nonnegative and satisfy
\[
        0\leq \Gamma_\lambda^{(k)}(\underline{X}_k;\underline{Y}_k)
        \leq \Gamma_0^{(k)}(\underline{X}_k;\underline{Y}_k)
\]
for a.e. $\underline{X}_k,\underline{Y}_k\in(\bT^d)^k$.
\end{proposition}

\medskip

\begin{proof}
Using \eqref{eq:Ginibre-rep}, the ratio estimate above, and
\eqref{eq:open-total-mass}, we get
\[
\begin{aligned}
0
\leq
\Gamma_\lambda^{(k)}(\underline{X}_k;\underline{Y}_k)                                                    
\leq
\int\dd\mathbb{B}_{\underline{X}_k,\underline{Y}_k}^{(k)}(\omega)                                                
=
\Gamma_0^{(k)}(\underline{X}_k;\underline{Y}_k).
\end{aligned}
\]
\end{proof}

\medskip

\begin{corollary}[Hilbert--Schmidt and projected number bounds]
For every $k\geq1$,
\begin{equation}\label{eq:interacting-HS-bound}
    \lambda^k
    \norm{\Gamma_\lambda^{(k)}}_{\gS^2}
    \leq
    \big(\tr h^{-2}\big)^{k/2}.
\end{equation}
\end{corollary}

\medskip

\begin{proof}
By Proposition~\ref{prop:Ginibre-domination},
\[
    \norm{\Gamma_\lambda^{(k)}}_{\gS^2}^2
    =
    \int_{(\bT^d)^k}\int_{(\bT^d)^k}
    |\Gamma_\lambda^{(k)}(\underline{X}_k;\underline{Y}_k)|^2\dd \underline{X}_k\dd \underline{Y}_k
    \leq
    \norm{\Gamma_0^{(k)}}_{\gS^2}^2 .
\]
The free quasi-free density matrices satisfy
$\Gamma_0^{(k)}=(\gamma_0)^{\otimes k}$ on $\gH^{\otimes_s k}$. Hence
\begin{align*}
\lambda^k\norm{\Gamma_\lambda^{(k)}}_{\gS^2}
\leq
\lambda^k\norm{\gamma_0}_{\gS^2}^k
=
\Big(
\sum_{p\in\Z^d}
\Big[
\frac{\lambda}{e^{\lambda h(p)}-1}
\Big]^2
\Big)^{k/2}                                                              
\leq
\Big(
\sum_{p\in\Z^d}h(p)^{-2}
\Big)^{k/2}.
\end{align*}
This proves \eqref{eq:interacting-HS-bound}.
\end{proof}

%\bibliographystyle{plain} % apalike, ieee, plain, alpha, unsrt, abbrv
%\bibliography{cubCFT}

\end{document}